\documentclass[12pt,a4paper,reqno,oneside]{amsart}
\usepackage[utf8]{inputenc}
\usepackage[english]{babel}
\usepackage{amsmath,mathrsfs}
\usepackage{amsfonts}
\usepackage{amssymb}

\usepackage{amsxtra,latexsym,amscd,amsthm,marvosym,multirow}
\usepackage[a4paper,inner=2.3cm, outer=2.3cm, top=2.3cm, bottom=2.3cm]{geometry}
\usepackage{hhline}
\usepackage{setspace}
\usepackage{graphicx}
\usepackage{framed, fancybox}
\usepackage{epstopdf}
\usepackage{enumerate}
\usepackage{csquotes}
\usepackage[shortlabels]{enumitem}
\usepackage{cases}
\usepackage{hyperref}
\usepackage{tikz}
\usetikzlibrary{decorations.pathreplacing}
\usepackage{xifthen,tikz}
\usetikzlibrary{lindenmayersystems}
\usetikzlibrary[shadings]
\usetikzlibrary{calc,intersections,through,backgrounds}
\usepackage{morefloats}
\usepackage{subcaption, algorithmicx}
\usepackage{color}
\theoremstyle{plain}
\newtheorem{theorem}{Theorem}[section]
\newtheorem{lemma}[theorem]{Lemma}
\newtheorem{proposition}[theorem]{Proposition}

\theoremstyle{definition}
\newtheorem{assumpA}{Assumption}

\newtheorem{example}{\bf Example\/}

\theoremstyle{remark}
\newtheorem{remark}[theorem]{\bf Remark}
\newtheorem{notation}[theorem]{\bf Notation}

\def\beq{\begin{eqnarray*}}\def\eeq{\end{eqnarray*}}
\def\bq{\begin{equation}}\def\eq{\end{equation}}

\newcommand{\N}{\mathbb{N}}
\newcommand{\Z}{\mathbb{Z}}
\newcommand{\R}{\mathbb{R}}
\newcommand{\C}{\mathbb{C}}

\newcommand{\eps}{\varepsilon}
\newcommand{\dd}{\mathrm{d}}

\newcommand{\Dom}{\mathcal{D}}
\newcommand{\arcsinh}{\mathop{\mathrm{arcsinh}}\nolimits}
\newcommand{\arccosh}{\mathop{\mathrm{arccosh}}\nolimits}
\newenvironment{psmallmatrix}
  {\left(\begin{smallmatrix}}
  {\end{smallmatrix}\right)}
\usepackage[normalem]{ulem}
\definecolor{DarkGreen}{rgb}{0,0.5,0.1} 

\newcommand\soutD{\bgroup\markoverwith
{\textcolor{DarkGreen}{\rule[.5ex]{2pt}{1pt}}}\ULon}
\newcommand{\Hm}[1]{\leavevmode{\marginpar{\tiny%
$\hbox to 0mm{\hspace*{-0.5mm}$\leftarrow$\hss}%
\vcenter{\vrule depth 0.1mm height 0.1mm width \the\marginparwidth}%
\hbox to
0mm{\hss$\rightarrow$\hspace*{-0.5mm}}$\\\relax\raggedright #1}}}

\makeatletter

\@addtoreset{equation}{section}  
\makeatother

\begin{document}
\title[]{Pseudomodes for non-self-adjoint Dirac operators}
\author{David Krej\v{c}i\v{r}\'{i}k}
\address[David Krej\v{c}i\v{r}\'{i}k]{Department of Mathematics, Faculty of Nuclear Sciences and Physical Engineering,
Czech Technical University in Prague, Trojanova 13, 12000 Prague 2, Czech Republic}
\email{david.krejcirik@fjfi.cvut.cz}

\author{Tho Nguyen Duc}
\address[Tho Nguyen Duc]{Department of Mathematics, Faculty of Nuclear Sciences and Physical Engineering,
Czech Technical University in Prague, Trojanova 13, 12000 Prague 2, Czech Republic}
\email{nguyed16@fjfi.cvut.cz}


\begin{abstract}
Depending on the behaviour of 
the complex-valued electromagnetic potential 
in the neighbourhood of infinity, 
pseudomodes of one-dimensional Dirac operators
corresponding  to large pseudoeigenvalues are constructed.
This is a first systematic approach which goes beyond the standard semi-classical setting. Furthermore, this approach results in substantial
progress in achieving optimal conditions and conclusions as
well as in covering a wide class of previously inaccessible
potentials, including superexponential ones.
\end{abstract}
\maketitle

\section{Introduction}
%
\subsection{Motivations}
 Spectral theory of self-adjoint operators  
has exhibited an enormous development 
since the discovery of quantum mechanics 
at the beginning of the last century
and it can be regarded as well understood in many respects by now.
Recent years have brought new motivations
for considering non-self-adjoint operators, too,
notably due to the unconventional concept 
of representing physical observables
by operators which are merely similar to
self-adjoint ones \cite{Bagarello-book}. 
This brand-new construct, 
nicknamed \emph{quasi-self-adjoint} quantum mechanics,
remained overlooked for almost hundred years
and has led to challenging mathematical problems 
which cannot be handled by standard tools
and the theory is by far not complete.

It has been accepted by mathematicians as well as physicists
that an appropriate characteristic 
which conveniently describes the pathological 
properties of non-self-adjoint operators is the notion 
of pseudospectra
\cite{LM05,Davies_2007,Helffer}.
Given a positive number~$\eps$,
the $\eps$-\emph{pseudospectrum} $\sigma_\eps(H)$ 
of any operator~$H$ in a complex Hilbert space
is defined as its spectrum $\sigma(H)$ enriched by 
those complex points~$\lambda$ 
(called \emph{pseudoeigenvalues})
for which there exists a vector $\Psi \in \Dom(H)$
(called \emph{pseudoeigenvector} or \emph{pseudomode}
or \emph{quasimode})
such that 
\begin{equation}\label{Mo Eq Pseudo Eigen}
\Vert (H-\lambda)\Psi \Vert 
< \varepsilon \, \Vert \Psi\Vert\,.
\end{equation}
This notion is trivial for self-adjoint
(or, more generally, normal) operators,
because then $\sigma_\eps(H)$ merely coincides  
with the $\eps$-tubular neighbourhood of the spectrum.
If~$H$ is non-normal, however, 
the pseudospectrum $\sigma_\eps(H)$ 
can contain points which lie outside
(in fact, possibly ``very far'' from)
the spectrum $\sigma(H)$.
It turns out that it is the pseudospectrum 
which determines the decay of the semigroup 
generated by~$H$ as well as  
the behaviour of the spectrum of~$H$ under small perturbations.

The usefulness of pseudospectra in 
quasi-self-adjoint quantum mechanics
was pointed out by Siegl and one of the present authors in~\cite{SK}.
Based on the semiclassical construction of pseudomodes
in Davies' pioneering work \cite{Da99}
(see also~\cite{Zworski_2001} and~\cite{NJM04}),
we proved an abrupt lack of quasi-self-adjointness 
for the prominent \emph{imaginary cubic oscillator} of
\cite{Bender-Boettcher_1998},
which stayed at the advent of the so-called
$\mathcal{PT}$-\emph{symmetric} quantum mechanics.
Many other Schr\"odinger operators 
with complex-valued potentials
were included in the subsequent works 
\cite{Hitrik-Sjostrand-Viola_2013,
Henry,Henry_2014,Novak_2015,KSTV,HK,KS19,Arnaiz}. 

In this series of works,
the paper~\cite{KS19} of Siegl and one of the present authors
is exceptional in that it develops 
a \emph{direct construction of large-energy pseudomodes}
(\emph{i.e.}\ those corresponding 
in~\eqref{Mo Eq Pseudo Eigen}
to $|\lambda|\to+\infty$ with $\eps_\lambda \to 0$), 
which does not require
the passage through semiclassical Schr\"odinger operators.
In fact, the semiclassical setting follows 
as a special consequence of~\cite{KS19}.
Moreover, the newly developed, general approach of~\cite{KS19} 
enables one to cover previously inaccessible potentials
such as the exponential and discontinuous ones.
What is more, the technique is applicable to other models 
such as the damped wave equation~\cite{AS20}.
 
The objective of the present paper is to extend 
the method developed in~\cite{KS19} 
to \emph{relativistic quantum mechanics} 
by considering Dirac instead of Schr\"odinger operators. 
This extension is both 
mathematically challenging and physically interesting, 
because the Dirac equation is not scalar 
and the external electromagnetic perturbations 
are allowed to be fundamentally matrix-valued.
We also remark that the present model is additionally relevant 
in the context of \emph{graphene} materials.
What is more, we substantially generalise the method of~\cite{KS19}  
on a technical level, which enables us to cover
previously inaccessible potentials,
including superexponential ones.
Non-self-adjoint Dirac operators 
have attracted a lot of attention recently
\cite{CLT14,C14,Cuenin_2017,Enblom_2018,CP18,
FK9,CFK,D'Ancona-Fanelli-Schiavone},
however, we are not aware of any result related
to the construction of pseudomodes in $\lambda$-dependent WKB form.

\subsection{The model and main results}\label{Section Dirac op}
Following the classical reference~\cite{Th92},
we introduce the one-dimensional 
\emph{free Dirac operator} by 
\[ H_0 := -i\sigma_{1}\frac{\dd }{\dd x}  + m \sigma_{3} \,,\qquad \sigma_1:= \begin{pmatrix}
0 & 1\\
1 & 0
\end{pmatrix}, \qquad \sigma_3:= \begin{pmatrix}
1 & 0\\
0 & -1
\end{pmatrix},\]
with mass $m\geq 0$.
We think of $H_0$ as an operator acting in the Hilbert space $\mathcal{H}$ which is a direct sum of two $L^2(\R)$ spaces,
\[ \mathcal{H} := L^2(\R) \oplus L^2(\R)= \left\{ \begin{pmatrix}
u_1\\u_2
\end{pmatrix}: \ u_1, u_2 \in  L^2(\R)\right\}, \]
equipped with the inner product
\begin{equation}\label{inner}
  \langle u,v\rangle 
  := \int_{\R} \sum_{j=1}^{2} u_{j}(x) \overline{v_{j}(x)}\, \dd x
  \,.
\end{equation}
Alternatively, we identify~$\mathcal{H}$
with $L^2(\R)^2$ or $L^2(\R) \otimes \C^2$.
It is well known that~$H_0$ is self-adjoint 
if $\Dom(H_0) := H^1(\R)^2$ and one has
\begin{equation}\label{spectrum}
 \sigma(H_0)=(-\infty,-m]\cup [m,+\infty)
 \,.
\end{equation}

The free operator~$H_0$ is perturbed 
by a complex matrix-valued potential
$V:\R \to \C^{2\times 2}$. We write 
\begin{equation*}
V(x):= \begin{pmatrix}
V_{11}(x) && V_{12}(x)\\
V_{21}(x) && V_{22}(x)
\end{pmatrix},
\end{equation*}
where $x \in \R$, and assume that 
$V\in L_{\mathrm{loc}}^2(\R)\otimes \C^{2 \times 2}$, 
meaning that all the components $V_{ij}$ with $i,j \in \{1,2\}$
are complex-valued (scalar) functions
belonging to $ L_{\mathrm{loc}}^2(\R)$. 
In the case of real-valued potentials, the special scenario $V_{11}=V_{22}$ and $V_{12}=V_{21}=0$
(respectively, $V_{11}=V_{22}=0$ and $V_{12}=V_{21}$)
corresponds to purely electric
(respectively, purely magnetic) fields.
The special case $V_{11}=-V_{22}$ and $V_{12}=V_{21}=0$
is known as the scalar potential. 
We keep the same terminology in the general, complex-valued case.
 
The perturbed operator~$H_V$ is introduced 
as the maximal extension of the operator sum $H_0 + V$,
where we denote by the same symbol~$V$    
the maximal operator of multiplication 
by the generated function~$V$.
More specifically, 
\begin{equation}
\begin{aligned}\label{Eq Def H_V}
H_{V}f &:=\begin{pmatrix}
(-i\partial_{x}+V_{12})f_{2}+ (V_{11}+m)f_{1}\\
(-i\partial_{x}+V_{21})f_{1}+ (V_{22}-m)f_{2}
\end{pmatrix},\\
\mathcal{D}(H_{V}) &:=\left\{f:=(f_1,f_2)\in \mathcal{H} : H_{V} f \in \mathcal{H} \right\}.
\end{aligned}
\end{equation}
The local integrability conditions imposed 
on the coefficients of~$V$ ensure that
all the actions of~$H_{V}$ in~\eqref{Eq Def H_V} 
are well defined in the sense of distributions.
By straightforward arguments, 
it follows that~$H_{V}$ is a closed operator. 
However, the closedness of $H_{V}$ is inessential 
for our construction of pseudomodes. 
In fact, many of the constructed pseudomodes 
belong to $C_0^\infty(\R)^2$,
so the majority of our results apply also to any 
(possibly non-closed or trivial) extension 
of the sum $H_0 + V$ initially defined on $C_0^\infty(\R)^2$.

The main purpose of this work is to build 
a $\lambda$-dependent family $\Psi_{\lambda}$ such that
\begin{equation}\label{Mo Eq Pseudo Eigen 2}
\Vert \left( H_{V}-\lambda \right) \Psi_{\lambda} \Vert \leq o(1) \Vert \Psi_{\lambda} \Vert \qquad \text{as } \lambda \to \infty 
\ \text{ in } \ \Omega \subset \C
\,,
\end{equation}
where~$\|\cdot\|$ is the norm associated with~\eqref{inner}.
The principal tool for this construction is the (J)WKB analysis
(also known as the Liouville--Green approximation). 
In addition to investigating the rate of the decay 
in~\eqref{Mo Eq Pseudo Eigen 2}, 
we also address the question of describing the shape 
of the complex region~$\Omega$ depending on~$V$.
The method used in Section~\ref{Section complex} gives us a way to sketch the region~$\Omega$ which seems rather optimal, 
even for low regular potential. It is interesting to observe that the domain $\Omega$ does not depend on the regularity of the potentials when the imaginary parts of $V_{11}$ and $V_{22}$ grow slowly at $+\infty$ such as logarithmic one (see Example \ref{Example Log}) or root-type $x^{\gamma}$ with $\gamma\in(0,1)$ (see Example \ref{Example Pol}).

As a foretaste of our main theorems 
without going into technical details, 
we present here a very special example 
of~\eqref{Mo Eq Pseudo Eigen 2}
for real~$\lambda$ and purely electric perturbations. 
The following theorem is a particular consequence 
of Theorem~\ref{Theorem 1} and Remark~\ref{Remark sign} below. From now on, $\mathrm{Re}\, u$ and $\mathrm{Im}\, u$ denote, respectively, the real part and the imaginary part of a function $u:\R \to \C$.
\begin{theorem}
Let $v\in W_{\mathrm{loc}}^{2,\infty}(\R)$ satisfy
\begin{equation}\label{asymptotes}
\limsup_{x\to -\infty} \mathrm{Im} \, v(x)<0,\qquad \liminf_{x\to -\infty} \mathrm{Im}\, v(x)>0,
\end{equation}
and set $V_{11}:= v =: V_{22}$ 
and $V_{12}:= 0 =: V_{21}$. 
Assume further that there exist continuous function $f_{\pm}:I^{\pm} \to (0,+\infty)$ 
\[ \vert f_{\pm}(x) \vert =\mathcal{O}\left( \int_{0}^{x} \mathrm{Im}\,v(t)\, \dd t \right) \qquad \text{ as }x\to+\infty,\]
and, for all $n\in \{1,2\}$,
\[  \vert v^{(n)}(x) \vert = \mathcal{O}\left(f_{\pm}(x)^{n} \vert v(x)\vert\right)\qquad \text{ as }x\to+\infty.\]
Then, there exists a $\lambda$-dependent family $\left(\psi_{\lambda}\right)\subset \mathcal{D}(H_{V})\setminus \{ 0\}$ such that
\begin{equation}\label{both}
\frac{\Vert \left(H_{V} -\lambda\right) \Psi_{\lambda}\Vert }{\Vert \Psi_{\lambda} \Vert} =o(1)\qquad \text{as } \lambda  \to \pm\infty.
\end{equation}
\end{theorem}

Although the matrix structure of the potential is rather simple, 
the assumptions of the above theorem allow us to touch a very large class of potentials. For example, the following functions for $\vert x \vert\geq 1$ (the middle part of the functions for $\vert x \vert\leq 1$ can be adjusted such that we have $v\in  W_{\mathrm{loc}}^{2,\infty}(\R)$) are covered:  
\begin{enumerate}[i)]
\item 
Polynomial-like functions
$v(x):= \vert x\vert^{\alpha}+ i \,\mathrm{sgn}(x) \vert x \vert^{\gamma}$ with $\alpha\in \R, \gamma\geq 0$. 
\\
Here we choose $f_{\pm}(x):=\vert x \vert^{-1}$.
\item 
Exponential functions $v(x):=e^{\vert x \vert^{\alpha}}+i \,\mathrm{sgn}(x) e^{\vert x \vert^{\gamma}}$ with $\alpha\in \R, \gamma\geq 0$. 
\\
Here we choose $f_{\pm}(x):=\vert x \vert^{\max \{\alpha,\gamma\}-1}$.
\item \label{Exam superexponential}
Superexponential functions $v(x):=i \,\mathrm{sgn}(x) e^{e^x}$. 
\\
Here we choose $f_{\pm}(x):=e^{x}$.
\end{enumerate}

\begin{remark}\label{Rem.normal}
Condition~\eqref{asymptotes} ensures
that~$H_V$ is ``significantly non-normal''.
Indeed, if $(H_{V})^{*}$ is the formal adjoint of~$H_V$, 
\emph{i.e.} $(H_{V})^{*}=H_{V^{*}}$,  
then it is straightforward to verify (at least algebraically)
that the normality $H_{V}(H_{V})^{*}=(H_{V})^{*} H_{V}$ 
holds if and only if the following identities
\begin{equation*}
\left\{ 
\begin{aligned}
&\mathrm{Im} \, V_{11}=\mathrm{Im} \, V_{22} = \textup{constant},\\
&\mathrm{Re} \, V_{12}=\mathrm{Re} \, V_{21},\\
&\vert \mathrm{Im} \, V_{12} \vert = \vert \mathrm{Im} \, V_{21} \vert,\\
& \mathrm{Im} \, V_{12}+ \mathrm{Im} \, V_{21}=\textup{constant},\\
&\left( \mathrm{Im} \, V_{12}+ \mathrm{Im} \, V_{21}\right)\left(\mathrm{Re}\, V_{22}-\mathrm{Re}\, V_{11}-2m \right)=0,
\end{aligned}
\right.
\end{equation*}
hold simultaneously on $\R$.
For matrix-valued potentials~$V$ of more general structures,
the asymptotics of the sum of the imaginary parts
of the diagonal components~$V_{11}$ and~$V_{22}$   
should be considered instead of the imaginary part of~$v$
in~\eqref{asymptotes},
see condition~\eqref{Assump G Diagonal 1} below.
\end{remark}

\subsection{Comparison between Schr\"{o}dinger and Dirac pseudomodes}
The Dirac setting is richer in that the perturbation~$V$
is a matrix-valued function, while it is just a scalar
potential in the Schr\"{o}dinger case.
Let us make a brief comparison of the present results 
with the Schr\"{o}dinger situation considered in~\cite{KS19}
(see Assumption~\ref{Assump General} below in this article 
and \cite[Ass.~I]{KS19}):
\begin{enumerate}[a)]
\item
Because of the unboundedness of the spectrum of~$H_0$
both from below and from above, see~\eqref{spectrum},
it is not surprising that we are able to construct 
pseudomodes for $\lambda \to \pm\infty$, while
just the limit $\lambda \to +\infty$
is relevant in the Schr\"{o}dinger case.
\item
The regularity of the potentials has a direct influence on the decay rates of the problem \eqref{Mo Eq Pseudo Eigen 2} for both the Schr\"{o}dinger and Dirac cases. The more regular the potential is, the stronger the rate of decay in \eqref{Mo Eq Pseudo Eigen 2} is obtained. 
\item
A version of the ``non-normality condition'' 
is imposed in the Schr\"{o}dinger case as well, 
see~\cite[Cond.~(3.1)]{KS19}. 
However, as explained in Remark~\ref{Rem.normal},
an alternative condition needs to be imposed 
for more general structures of the matrix-valued potentials~$V$ 
in the Dirac case. 
\item
The assumption that the real part of the potential is controlled 
by its imaginary part in \cite[Cond.~(3.3)]{KS19} 
can be completely ignored in the Dirac case. 
This salient feature due to the Dirac structure
is explained in Remark \ref{Remark Re Psi} below. 
\item
Moreover, the class of functions 
whose derivatives are controlled by the functions 
is significantly extended in our assumptions.
For example, the superexponential function $e^{e^x}$ 
of example~\ref{Exam superexponential} above does satisfy our assumption,
while it is not covered by \cite[Cond.~(3.2)]{KS19}.
This more general result of the present paper
is technically due to the freedom in the choice
of the function~$f_{\pm}$ above, 
while it is fixed to the canonical choice $f_{\pm}(x):=\vert x \vert^{\nu}$ with some $\nu\in \R$ in~\cite{KS19}.
\end{enumerate}

\subsection{Handy notations}
Here we summarise some special notations 
which will appear regularly in the paper: 
\begin{enumerate}[1)]
\item $\N_{k}$, with a non-negative integer $k$,  
is the set of integers starting from $k$;
\item $\R_{-}:=(-\infty,0)$ and $\R_{+}:=(0,+\infty)$;
\item $f^n$ and $f^{(n)}$ denotes respectively
the power $n$ and the $n$-th derivative of a function $f:\R \to \C$ with $n\in \N_{0}$;
\item 
We use the same symbol $\Vert \cdot \Vert$
for $L^2$-norms of both scalar- and vector-valued functions;
\item For two real-valued functions $a$ and $b$, we write
$a\lesssim b$ (respectively, $a \gtrsim b$)  
if there exists a constant $C>0$, independent of $\lambda$ and $x$ (or any other relevant parameter), such that
$a\leq C b$  (respectively, $a\geq Cb$);
\item $a \approx b$ if $a\lesssim b$ and $a \gtrsim b$;
\item $[[m,n]]:=[m,n]\cap \Z$ for all $m,n\in \R$;
\item $\vert j \vert := j_1 + j_2$ for all $j=(j_1,j_2)\in \N_{0}^2$.
\end{enumerate}

\subsection{Structure of the paper}
The paper is organised as follows. 
In Section~\ref{Section WKB}, 
we present a general scheme of constructing 
a pseudomode satisfying~\eqref{Mo Eq Pseudo Eigen 2}
for the Dirac operator by the WKB method. 
This scheme is applied to real~$\lambda$'s 
in Section~\ref{Section Real},
while more general complex curves are allowed 
in Section~\ref{Section complex}.
Many illustrative examples are considered 
at the end of each of the two sections.  

\section{WKB construction}\label{Section WKB}

\subsection{Warming-up}
Let us start the scheme of constructing the pseudomode 
of the Dirac operator~\eqref{Eq Def H_V}
satisfying~\eqref{Mo Eq Pseudo Eigen 2}
by searching it in the form
\begin{equation}\label{Eq k1k2}
\Psi_{\lambda} := 
\begin{pmatrix}
k_1 u_{\lambda}\\
k_2 v_{\lambda}
\end{pmatrix},
\qquad
\begin{aligned}
k_1(x)&:=\exp{\left(-i\int_{0}^{x} V_{21}(\tau) \, \dd \tau\right)},
\\
\qquad k_2(x) &:= \exp{\left(-i\int_{0}^{x} V_{12}(\tau) \, \dd \tau\right)},
\end{aligned}
\end{equation}
where $u_{\lambda}, v_{\lambda}$ depending on $\lambda$ will be determined later.  
This structure of pseudomode allows us to pull out the off-diagonal terms of potential $V$ by the following step
\begin{align*}
(H_{V}-\lambda) \Psi_{\lambda} 
&=\begin{pmatrix}
\left(-ik_2\partial_{x} -ik_2^{(1)}+V_{12}k_2 \right)v_{\lambda} + \left(V_{11}+m-\lambda \right) k_1 u_{\lambda} \\
\left(-ik_1\partial_{x} -ik_1^{(1)}+V_{21}k_1 \right)u_{\lambda} + \left(V_{22}-m-\lambda \right)k_2 v_{\lambda}
\end{pmatrix}\\
&=\begin{pmatrix}
k_2\left(-i\partial_{x}\right)v_{\lambda} + k_1\left(V_{11}+m-\lambda \right)  u_{\lambda} \\
k_1 \left(-i\partial_{x} \right)u_{\lambda} + k_2\left(V_{22}-m-\lambda \right) v_{\lambda}
\end{pmatrix}.
\end{align*}

By letting one of the two components of $(H_{V}-\lambda) \Psi_{\lambda}$ be zero, for example the second one, we can compute $v_{\lambda}$ by $u_{\lambda}$:
\[ v_{\lambda} = \frac{k_1}{k_2}\frac{(-i\partial_{x})u_{\lambda}}{\lambda+m-V_{22}}.\]
Here we assume that $\lambda+m-V_{22}\neq 0$ on the support of the pseudomode; this will be ensured by the condition \eqref{Eq Cond Re} later. 
In this way, we have relaxed the problem to finding $u_{\lambda}$ such that 
\begin{align*}
\frac{\Vert (H_{V}-\lambda)\Psi_{\lambda} \Vert}{\Vert \Psi_{\lambda} \Vert}=\frac{\Vert \mathscr{L}_{\lambda,V} u_{\lambda}\Vert_{L^2}}{\sqrt{\Vert k_1 u_{\lambda} \Vert_{L^2}^2 + \left\Vert k_1\frac{\left(-i\partial_{x}\right)u_{\lambda} }{\lambda+m-V_{22}} \right\Vert_{L^2}^2}} = o(1) \qquad \text{as } \lambda\to \infty.
\end{align*}
Here the differential expression $\mathscr{L}_{\lambda,V}$ is defined by
\begin{equation*}
\begin{aligned}
\mathscr{L}_{\lambda,V} := k_2\left(-i\partial_{x} \right)K_{\lambda} \left(-i\partial_{x}  \right)-k_1\left(\lambda-m-V_{11} \right)
\end{aligned},
\qquad
K_{\lambda}:= \frac{k_1}{k_2} \frac{1}{\lambda+m-V_{22}}.
\end{equation*} 

\begin{remark}
In \cite{AS20}, in order to construct the pseudomode for a damped wave system, the authors transferred the problem into finding the pseudomode of a quadratic operator having the Schr\"{o}dinger form. We, meanwhile, will establish the WKB construction directly for the Sturm--Liouville-like operator $\mathscr{L}_{\lambda,V}$ without converting it to the Schr\"{o}dinger operator.
\end{remark}

Let us consider a sufficiently regular and complex-valued function $P:\R \to \C$ which will be determined later in the WKB process. We consider the formal conjugated operator 
\begin{align*}
\mathscr{L}_{\lambda,V}^{P} := \ & e^{P}\mathscr{L}_{\lambda,V} e^{-P}
\\
= \ & \frac{k_1}{\lambda+m-V_{22}}\left[ -\partial_{x}^2 +\left( 2 P^{(1)}-\frac{K_{\lambda}^{(1)}}{K_{\lambda}}\right)\partial_{x} + P^{(2)}+\frac{K_{\lambda}^{(1)}}{K_{\lambda}} P^{(1)}-(P^{(1)})^2-V_{\lambda}\right]\,,
\end{align*}
where
\begin{equation}\label{Eq V lambda}
V_{\lambda}:= (\lambda-m-V_{11})(\lambda+m-V_{22}).
\end{equation}
Let us denote
\begin{equation}\label{Eq remainder}
R_{\lambda} :=  P^{(2)}+\frac{K_{\lambda}^{(1)}}{K_{\lambda}} P^{(1)}-(P^{(1)})^2-V_{\lambda},
\end{equation}
which will play the role of the remainder in the WKB analysis.

We consider $u_{\lambda}$ in the form
$u_{\lambda} := \xi e^{-P},$
where $\xi$ is a cut-off function whose support 
is allowed to depend on~$\lambda$; it will be determined later in Section \ref{Section: cutoff}. Hence, the action of $\mathscr{L}_{\lambda,V}$ on $u_{\lambda}$ can be expressed as 
\begin{equation}\label{Eq L u}
\begin{aligned}
\mathscr{L}_{\lambda,V} u_{\lambda} 
=& \ e^{-P} \mathscr{L}_{\lambda,V}^{P}\xi\\
=& \ -\frac{k_1e^{-P}}{\lambda +m-V_{22}}\xi^{(2)} + \frac{k_1e^{-P}}{\lambda +m-V_{22}} \left( 2P^{(1)}-\frac{K_{\lambda}^{(1)}}{K_{\lambda}}\right) \xi^{(1)} + \frac{k_1e^{-P}R_{\lambda}}{\lambda+m-V_{22}} \xi \,.
\end{aligned}
\end{equation}
The WKB strategy is as follows. 
For each $n\in \N_{0}$, we look for the phase $P$ in the form
\begin{equation}\label{Eq Pn}
P_{\lambda,n} (x)= \sum_{k=-1}^{n-1} \lambda^{-k} \psi_{k}(x),
\end{equation}
where functions $\left(\psi_{k}\right)_{k\in[[-1,n-1]]}$ are to be determined by solving ordinary differential equations (ODEs); the number $n$ will be chosen later depending on the maximal possible order derivative of $V$. After that, we show that the exponential decay of $\psi_{-1}$ allows the norm of the first two terms in \eqref{Eq L u} to decay exponentially according to $\lambda$ (Proposition \ref{Lemma L2 norm}) and the norm of the final term to decrease with the rate power of $\lambda^{-1}$ (Theorem \ref{Theorem 1}).

Starting with $n=0$ and putting $P_0 = \lambda \psi_{-1}$ into \eqref{Eq remainder}, we obtain
\begin{equation}\label{Eq remainder 0}
R_{\lambda,0}
:=\lambda\left(\psi_{-1}^{(2)}+\frac{K_{\lambda}^{(1)}}{K_{\lambda}}  \psi_{-1}^{(1)}\right)-\lambda^2(\psi_{-1}^{(1)})^2-V_{\lambda}
=\lambda\left(\psi_{-1}^{(2)}+\frac{K_{\lambda}^{(1)}}{K_{\lambda}}  \psi_{-1}^{(1)}\right).
\end{equation}
Here the second equality follows by solving the eikonal equation
$ -\lambda^2(\psi_{-1}^{(1)})^2-V_{\lambda} =0, $;
in this way, the second order of $\lambda$ in $R_{\lambda,0}$ is removed 
and the final order of~$\lambda$  has been reduced. 

We can do the same trick for any $n\in \N_{1}$.
Replacing~$P$ by~$P_{\lambda,n}$ in \eqref{Eq remainder}, 
we obtain
$$
\left( \sum_{k=-1}^{n-1} \lambda^{-k} \psi_{k}^{(2)}\right)+\frac{K_{\lambda}^{(1)}}{K_{\lambda}} \left( \sum_{k=-1}^{n-1} \lambda^{-k} \psi_{k}^{(1)}\right)-\left( \sum_{k=-1}^{n-1} \lambda^{-k} \psi_{k}^{(1)}\right)^2 - V_{\lambda}
= \sum_{\ell=-2}^{n-2} \lambda^{-\ell} \phi_{\ell} + \sum_{\ell=-1}^{n-2} \lambda^{-(n+\ell)} \phi_{n+\ell}.
$$
Here the functions $\phi_{\ell}$ with $\ell\in [[-2,2(n-1)]]$ are naturally defined by grouping together the terms attached with the same order of $\lambda$, with the exception of $V_{\lambda}$ which we include in the leading order term. In detail, the first $n+1$ functions $\phi_{\ell}$ are expressed by
\begin{equation*}
\begin{aligned}
&\lambda^2:  &-(\psi_{-1}^{(1)})^2-\frac{V_{\lambda}}{\lambda^2}  &= :\phi_{-2},\\
&\lambda^{1}:  &\psi_{-1}^{(2)}+\frac{K_{\lambda}^{(1)}}{K_{\lambda}}\psi_{-1}^{(1)}-2\psi_{-1}^{(1)}\psi_{0}^{(1)}&=:\phi_{-1},\\
&\vdots &\\
&\lambda^{-\ell}: &\psi_{\ell}^{(2)}+\frac{K_{\lambda}^{(1)}}{K_{\lambda}} \psi_{\ell}^{(1)} - \sum_{j=-1}^{\ell+1}\psi_{j}^{(1)} \psi_{\ell-j}^{(1)}&=:\phi_{\ell},\\
&\vdots &\\
&\lambda^{-(n-2)}: &\psi_{n-2}^{(2)}+\frac{K_{\lambda}^{(1)}}{K_{\lambda}} \psi_{n-2}^{(1)} - \sum_{j=-1}^{n-1}\psi_{j}^{(1)} \psi_{n-2-j}^{(1)}&=:\phi_{n-2},
\end{aligned}
\end{equation*}
and the last $n$ functions $\phi_{\ell}$ are
\begin{equation}\label{Eq phi n}
\begin{aligned}
&\lambda^{-(n-1)}: &\psi_{n-1}^{(2)}+\frac{K_{\lambda}^{(1)}}{K_{\lambda}} \psi_{n-1}^{(1)} - \sum_{j=0}^{n-1}\psi_{j}^{(1)} \psi_{n-1-j}^{(1)}&=:\phi_{n-1},\\
&\lambda^{-n}: &- \sum_{j=1}^{n-1}\psi_{j}^{(1)} \psi_{n-j}^{(1)}&=:\phi_{n},\\
&\vdots &\\
&\lambda^{-(n+\ell)}: &- \sum_{j=\ell+1}^{n-1}\psi_{j}^{(1)} \psi_{n+\ell-j}^{(1)}&=:\phi_{n+\ell},\qquad \text{for } \ell\in [[0,n-2]] \text{ if } n\geq 2,\\
&\vdots &\\
&\lambda^{-2(n-1)}: & -\left(\psi_{n-1}^{(1)}\right)^2&=:\phi_{2(n-1)}.
\end{aligned}
\end{equation}
Requiring $\phi_{\ell}=0$ for all $\ell\in [[-2,n-2]]$, we obtain $(n+1)$ ODEs which can be solved explicitly to find all $\{\psi_{k}\}_{k\in [[-1,n-1]]}$ by a recursion formula
\begin{equation}\label{Eq Recursion}
\begin{aligned}
\psi_{-1}^{(1)} &= \pm i\lambda^{-1} V_{\lambda}^{1/2},\\
\psi_{\ell+1}^{(1)} &= \frac{1}{2\psi_{-1}^{(1)}}\left(\psi_{\ell}^{(2)} + \frac{K_{\lambda}^{(1)}}{K_{\lambda}} \psi_{\ell}^{(1)} -  \sum_{j=0}^{\ell}\psi_{j}^{(1)} \psi_{\ell-j}^{(1)}\right).
\end{aligned}
\end{equation}
After solving these ODEs, the WKB remainder is
\begin{equation}\label{Eq remainder n}
R_{\lambda,n} := \sum_{\ell=-1}^{n-2} \lambda^{-(n+\ell)} \phi_{n+\ell},\qquad n\in \N_{1}\,.
\end{equation}
\begin{remark}\label{Remark WKB}
Let us make some comments at this stage.
\begin{enumerate}[i)]
\item The choice of the sign in the definition of $\psi_{-1}^{(1)}$ will be determined by the sign of the sum $\mathrm{Im}\, V_{11} +\mathrm{Im}\, V_{22}$ at infinity and the sign of $\lambda$ (see Remark \ref{Remark sign}). 
\item Since $V_{\lambda}$ is a complex-valued function, the square root appearing in \eqref{Eq Recursion} is considered as the principal branch of the square root which is defined as, for $z\in \C\setminus (-\infty,0]$,
\begin{equation*}
\sqrt{z} = \frac{1}{\sqrt{2}} \left(\vert z \vert+ \mathrm{Re} \, z \right)^{1/2} + i\frac{1}{\sqrt{2}} \frac{\mathrm{Im}\, z}{\left(\vert z \vert+ \mathrm{Re} \, z \right)^{1/2}}.
\end{equation*}
This principal branch of square root is holomorphic on $\C\setminus (-\infty,0]$. For that reason, the range of $V_{\lambda}$ needs to stay away from $(-\infty,0]$ such that the continuity (and/or the differentiability) of $V_{\lambda}^{1/2}$ is deduced by the continuity (and/or the differentiability) of $V_{11}$ and $V_{22}$. By writing $\lambda:= \alpha+i\beta$ ($\alpha,\beta\in\R$), this will be ensured if, by simple argument of product of two complex numbers,
\begin{equation}\label{Eq Cond Re}
 (\alpha-m- \mathrm{Re}\,V_{11}(x)) (\alpha+m-\mathrm{Re}\,V_{22}(x))>0,
\end{equation}
for all $x\in \R$. If $\mathrm{Re}\, V_{11}$ and $\mathrm{Re}\, V_{22}$ are bounded, by considering $\alpha$ very large, \eqref{Eq Cond Re} is always satisfied. Otherwise, we need to employ the support of $\xi$ so as to make \eqref{Eq Cond Re} happen. 
\end{enumerate}
\end{remark}
\begin{remark}[Beyond semiclassical]
It is a common knowledge that the limit of ``large energies''
in quantum mechanics is related to the ``semiclassical limit''.
While we indeed consider the spectral parameter~$\lambda$ 
in~\eqref{Mo Eq Pseudo Eigen 2} diverging in the complex plane 
and employ WKB analysis standardly used for semiclassical regimes too
(see, \emph{e.g.}, \cite{MJ99} for the Schr\"{o}dinger operator 
or \cite{BNRVN20} for the magnetic Laplacian),
there are some important novelties in our approach
 that we list here:
\begin{enumerate}[i)]
\item Our spectral parameter $\lambda\in \C$ not only plays the scaling role as semi-classical parameter $h\in \R_{+}$, but also indicates the direction in which the large complex number will belong to the pseudospectrum. 
\item  In the original WKB method applied to the spectral problems in the semi-classical regime, the solutions of the eikonal equation and the transports equations are independent of the semi-classical parameter. On the other hand, our solutions depend on the parameter~$\lambda$ and all estimates established for $(\psi_{k})_{k\in[[-1,n]]}$ are necessary to be uniform in $\lambda$.
This makes the present analysis considerably more demanding.
However, both the WKB strategies share the same scheme that the eikonal solution plays a dominant part in deciding the decay of the main problem (the problem \eqref{Mo Eq Pseudo Eigen 2} in this case).
\item While the semi-classical quasimodes always localise, the supports of our pseudomodes can be extended in some cases. Furthermore, the cut-off functions are occasionally needless in our WKB construction.
\item While semi-classical works deal with smooth potentials, our framework can cover the potentials with low regularity (possibly discontinuous).
\end{enumerate}
In summary, the present work goes beyond 
standard semiclassical settings.
What is more, our approach is more robust in the sense that
semiclassical results can be deduced as a consequence of it
(\emph{cf.}~\cite[Ex.~5.4]{KS19}), but not vice versa
(without the important developments mentioned above).
\end{remark}
\subsection{Structure of solutions of the transport equations and the WKB remainder}
From now on, we assume that we are dealing with the plus sign in the formula of $\psi_{-1}$ in \eqref{Eq Recursion}, unless otherwise stated. Let us list some first solutions of the first transport equations to see which structure they are equipped with:
\begin{align*}
\psi_{-1}^{(1)} &= \frac{iV_{\lambda}^{1/2}}{\lambda},\\
\psi_{0}^{(1)} &=  \frac{1}{4}\frac{V_{\lambda}^{(1)}}{V_{\lambda}}+ \frac{1}{2} \frac{K_{\lambda}^{(1)}}{K_{\lambda}},\\
\psi_{1}^{(1)}&=\frac{-i\lambda}{8V_{\lambda}^{1/2}} \left( \frac{V_{\lambda}^{(2)}}{V_{\lambda}}-\frac{5}{4} \frac{(V_{\lambda}^{(1)})^2}{V_{\lambda}^2}+2\frac{K_{\lambda}^{(2)}}{K_{\lambda}}-\frac{(K_{\lambda}^{(1)})^2}{K_{\lambda}^2} \right),\\
\psi_{2}^{(1)}&= \frac{-\lambda^2}{16V_{\lambda}} \left[ \frac{V_{\lambda}^{(3)}}{V_{\lambda}}-\frac{9}{2}\frac{V_{\lambda}^{(1)}V_{\lambda}^{(2)}}{V_{\lambda}^2}+\frac{15}{4}\frac{(V_{\lambda}^{(1)})^3}{V_{\lambda}^3} -\frac{V_{\lambda}^{(1)}}{V_{\lambda}}\left(2\frac{K_{\lambda}^{(2)}}{K_{\lambda}}- \frac{(K_{\lambda}^{(1)})^2}{K_{\lambda}^2}\right)\right.\\
&\hspace{2 cm} \left.+2\frac{K_{\lambda}^{(3)}}{K_{\lambda}} - 4 \frac{K_{\lambda}^{(1)}K_{\lambda}^{(2)}}{K_{\lambda}^2}-2\frac{(K_{\lambda}^{(1)})^{3}}{K_{\lambda}^3}\right].
\end{align*}
\textbf{The remainder when we solve up to $\psi_{-1}$:}
\begin{align*}
R_{\lambda,0}=  \frac{iV_{\lambda}^{(1)}}{2V_{\lambda}^{1/2}} + \frac{iK_{\lambda}^{(1)}}{K_{\lambda}} V_{\lambda}^{1/2}.
\end{align*}
\textbf{The remainder when we solve up to $\psi_0$:}
\begin{align*}
R_{\lambda,1}= \frac{V_{\lambda}^{(2)}}{4V_{\lambda}} - \frac{5}{16} \frac{(V_{\lambda}^{(1)})^{2}}{V_{\lambda}^{2}} +\frac{K_{\lambda}^{(2)}}{2K_{\lambda}}-\frac{(K_{\lambda}^{(1)})^2}{4K_{\lambda}^2}.
\end{align*}
\textbf{The remainder when we solve up to $\psi_1$:}
\begin{align*}
R_{\lambda,2} = & \frac{-i}{8V_{\lambda}^{1/2}}\left[ \frac{V_{\lambda}^{(3)}}{V_{\lambda}}-\frac{9}{2}\frac{V_{\lambda}^{(1)}V_{\lambda}^{(2)}}{V_{\lambda}^2}+\frac{15}{4}\frac{(V_{\lambda}^{(1)})^3}{V_{\lambda}^3} -\frac{V_{\lambda}^{(1)}}{V_{\lambda}}\left(2\frac{K_{\lambda}^{(2)}}{K_{\lambda}}-\frac{(K_{\lambda}^{(1)})^2}{K_{\lambda}^2} \right) +2\frac{K_{\lambda}^{(3)}}{K_{\lambda}}\right.\\
&\hspace{1.2cm} \left.- 4 \frac{K_{\lambda}^{(1)}K_{\lambda}^{(2)}}{K_{\lambda}^2}-2\frac{(K_{\lambda}^{(1)})^{3}}{K_{\lambda}^3}\right]+\frac{1}{64V_{\lambda}}\left(\frac{V_{\lambda}^{(2)}}{V_{\lambda}} - \frac{5}{4}\frac{(V_{\lambda}^{(1)})^2}{V_{\lambda}^2} +2\frac{K_{\lambda}^{(2)}}{K_{\lambda}}-\frac{(K_{\lambda}^{(1)})^2}{K_{\lambda}^2}\right)^2.
\end{align*}

For $n\in \N_{1}$, since the formulae of solutions $(\psi_{k})_{k\in[[0,n-1]]}$ are obtained from the recursion steps \eqref{Eq Recursion}, if we want to write the formula of $\psi_{k}$, the formulae of all $\psi_{\ell}$ with $\ell\leq k-1$ need to be explicitly given. It could be a challenging effort to find out the exact formulae for the transport solutions. However, the good news is that these solutions can be estimated without knowing their exact expression, instead a common structure of them is required. This is the content of the following lemma, but first, some notations should be introduced.

\begin{notation}
Let $f,g$ be two functions which are assumed to be sufficiently regular so that all appearing derivatives of them exist. For $j=(j_1,j_2)\in \N_{0}^{2}$, $r=(r_1,r_2)\in \N_{0}^{2}$ and $s=(s_1,s_2)\in \N_{1}^{2}$, we employ the following notations
\begin{equation}\label{Nota Derivative}
D_{j}^{r,s}(f,g):=\left\{\displaystyle \sum_{\alpha \in \mathcal{I}_{j}^{r,s}} c_{\alpha} (f^{(1)})^{\alpha_1} \ldots(f^{(s_1)})^{\alpha_{s_1}} (g^{(1)})^{\alpha_{s_1+1}}\ldots(g^{(s_2)})^{\alpha_{s_1+s_2}} : c_{\alpha}\in \C\right\},
\end{equation}
where
\begin{align}\label{Notation Ijr}
\mathcal{I}_{j}^{r,s} &:= \left\{ \alpha \in \N_{0}^{s_1+s_2}: \sum_{p=1}^{s_1} \alpha_{p} = j_1,\, \sum_{p=1}^{s_2} \alpha_{s_1+p} = j_2 ;\sum_{p=1}^{s_1} p \alpha_{p} = r_1,\,\sum_{p=1}^{s_2} p \alpha_{s_1+p} = r_2\right\}\,.
\end{align}
When $\mathcal{I}_{j}^{r,s} =\emptyset$, we make a convention that $D_{j}^{r,s}(f,g)=\{0\}$. Thus, if $j_{i}=0$ and $r_{i}\geq 1$ for some $i\in \{1,2\}$, then $D_{j}^{r,s}(f,g)=\{0\}$.
\end{notation}

\begin{lemma}\label{Lemma psi k'}
Let $n\in \N_{0}$, assume that $V_{11},V_{22}\in W_{\mathrm{loc}}^{n+1,2}(\R)$ and $V_{12}, V_{21}\in W_{\mathrm{loc}}^{n,2}(\R)$ and $\lambda \in \C$ are such that \eqref{Eq Cond Re} is satisfied. Let $\{\psi_{k}^{(1)}\}_{k\in [[-1,n-1]]}$ be a family determined by the formula \eqref{Eq Recursion}. Then their first order derivatives are of the form
\begin{equation*}
\psi_{k}^{(1)} = \frac{\lambda^{k}}{V_{\lambda}^{k/2}} \sum_{\vert j \vert=0}^{k+1}\sum_{\vert r\vert =k+1} \frac{d^{r,r-j+(1,1)}_{j}(V_{\lambda},K_{\lambda})}{V_{\lambda}^{j_1}K_{\lambda}^{j_2}},
\end{equation*}
where $\{r,j\}\subset \N_{0}^2$ and $d^{r,r-j+(1,1)}_{j}(V_{\lambda},K_{\lambda})\in D_{j}^{r,r-j+(1,1)}(V_{\lambda},K_{\lambda})$.
\end{lemma}

For each $k\geq 0$, the maximal possible order  derivative of $V_{11},V_{22}$ in $\psi_{k}'$ is $k+1$ and the maximal possible order derivative of $V_{12}$ and $V_{21}$ in $\psi_{k}^{(1)}$ is $k$. Indeed, notice that, from the definition of $V_{\lambda}$ and $K_{\lambda}$, the levels of the derivatives of $V_{\lambda}$ and $K_{\lambda}$ are equal to the levels of the derivatives of $V_{11}$ and $V_{22}$ while larger than the levels of the derivatives of $V_{12}$ and $V_{21}$ by one order. On the other hand, for all $i\in\{1,2\}$ with $j_i \geq 1$ and $r_{i}\leq k+1$,  we get
\[ \max \{ r_{i}-j_{i}+1\} \leq k+1.\]

The remainders are controlled by the next lemma.

\begin{lemma}\label{Lemma Rn}
Let $n\in \N_{0}$, assume that $V_{11},V_{22}\in W_{\mathrm{loc}}^{n+1,2}(\R)$ and $V_{12}, V_{21}\in W_{\mathrm{loc}}^{n,2}(\R)$ and $\lambda \in \C$ are such that \eqref{Eq Cond Re} is satisfied. For $n=0$, let $R_{\lambda,0}$ as in \eqref{Eq remainder 0}. For $n\geq 1$, let $\{\psi_{k}'\}_{k\in [[-1,n-1]]}$ be a family determined by the formula \eqref{Eq Recursion}, $\{ \phi_{k}\}_{k\in [[n-1,2(n-1)]]}$ as in \eqref{Eq phi n} and and $R_{\lambda,n}$ as in \eqref{Eq remainder n}. Then the maximal possible order  derivative of $V_{11},V_{22}$ is $n+1$ and of $V_{12}, V_{21}$ is $n$ in $R_{\lambda,n}$ and
\begin{equation}\label{Eq Estimate Rn}
\begin{aligned}
&\text{for } n=0:&&\vert R_{\lambda, 0} \vert \lesssim \frac{\vert V_{\lambda}^{(1)} \vert}{\vert V_{\lambda} \vert^{1/2}} + \frac{\vert K_{\lambda}^{(1)} \vert}{\vert K_{\lambda} \vert} \vert V_{\lambda} \vert^{1/2},\\
&\text{for } n\geq 1: &&\vert R_{\lambda, n} \vert \lesssim \sum_{\ell=-1}^{n-2} \frac{1}{\vert V_{\lambda} \vert^{(n+\ell)/2}} \sum_{\vert j \vert=1}^{n+\ell+2} \sum_{\vert r \vert=n+\ell+2}\frac{\vert d_{j}^{r,(n+1,n+1)}(V_{\lambda},K_{\lambda})\vert}{\vert V_{\lambda}\vert^{j_1} \vert K_{\lambda}\vert^{j_2}},
\end{aligned}
\end{equation}
where $\{r,j\}\subset \N_{0}^2$ and $d^{r,(n+1,n+1)}_{j}(V_{\lambda},K_{\lambda})\in D_{j}^{r,(n+1,n+1)}(V_{\lambda},K_{\lambda})$.
\end{lemma}
These lemmata can be proved, by induction, in the same manner as in Appendix of \cite{KS19}. Therefore we omit the proofs here.
\begin{remark}
At the end of this section, we want to show that choosing the shape of the pseudomode for the Dirac operator also plays important role technically. From the beginning, if we choose the basic form 
$\Psi_{\lambda}=\begin{psmallmatrix}
u_{\lambda}\\ v_{\lambda}
\end{psmallmatrix}$ and insert it to the eigenvalue equation $(H_{V}-\lambda)\Psi_{\lambda}=0$, we will have to deal WKB with the electromagnetic-like Schr\"{o}dinger operator
\[ \widetilde{\mathscr{L}}_{\lambda,V} =\left(-i\partial_{x}+V_{12}\right)\frac{1}{\lambda+m-V_{22}}\left(-i\partial_{x}+V_{21}\right)-\left(\lambda-m-V_{11}\right).\]
Then its formal conjugated operator is described as follows
\begin{align*}
\widetilde{\mathscr{L}}_{\lambda,V}^{P} &:= e^{P}\mathscr{L}_{\lambda,V} e^{-P}\\
&= \frac{k_1}{\lambda+m-V_{22}}\left[ -\partial_{x}^2 +\left( 2 P^{(1)}-\frac{V_{22}^{(1)}}{\lambda+m-V_{22}}-i(V_{12}+V_{21})\right)\partial_{x} \right.\\
&\hspace{3cm}\left. + P^{(2)}+\left(\frac{V_{22}^{(1)}}{\lambda+m-V_{22}}+i(V_{12}+V_{21})\right) P^{(1)}-(P^{(1)})^2-\widetilde{V}_{\lambda}\right]\,,
\end{align*}
where 
\[\widetilde{V}_{\lambda}= (\lambda-m-V_{11})(\lambda+m-V_{22})-V_{12}V_{21}+iV_{21}^{(1)}+iV_{21}\frac{V_{22}^{(1)}}{\lambda+m-V_{22}}.\]
We see that this form of $\widetilde{V}_{\lambda}$ is very complicated to consider its square-root. Furthermore, by solving some firsts transport equations, we recognize that the sum $(V_{12}+V_{21})$ attached with $P^{(1)}$ will destroy the structure of solutions of transport equations. These difficulties will make our WKB analysis unusable. Therefore, multiplying $u_{\lambda}$ and $v_{\lambda}$  with respectively, $k_{1}$ and $k_{2}$ not only gauges out $V_{12}$ and $V_{21}$, but also allows this method to be workable.
\end{remark} 
\section{Pseudomodes for $\lambda \to \pm \infty$}\label{Section Real}
Let us recall here the picture of the Schr\"{o}dinger operators to compare and outline the direction for the Dirac operators simultaneously. It is well known that the spectrum of 
the free Schr\"{o}dinger operator 
(\emph{i.e.}\ the Laplacian in $L^2(\R)$ 
with domain being the Sobolev space $H^2(\R)$) 
is given by the set $[0,+\infty)$. In \cite{KS19}, when the pseudoeigenvalue $\lambda$ is real, the pseudomode of the Schr\"{o}dinger operator with the complex-valued potential are constructed successfully when $\lambda$ is positive and very large. 
Now, as mentioned in Section \ref{Section Dirac op}, 
the spectrum of the free Dirac operator is a set which is symmetric through the origin, see~\eqref{spectrum}. 
Therefore, this evokes that the construction of the pseudomode for the Dirac operator~$H_{V}$ for the positive and negative~$\lambda$'s can be established.  

This expectation is also supported by looking at the structure of the WKB construction in both cases, especially the solution of eikonal equation which depends on the square root of $V_{\lambda}$:
\begin{equation*}
\begin{aligned}
&\text{\textbf{In the Schr\"{o}dinger case:}} &&V_{\lambda}^{\text{Schr\"{o}dinger}}:=\lambda-V^{\text{Schr\"{o}dinger}}.\\
&\text{\textbf{In the Dirac case:}} &&V_{\lambda}^{\text{Dirac}}:= (\lambda-m-V_{11})(\lambda+m-V_{22}).
\end{aligned}
\end{equation*}
Here, $V^{\text{Schr\"{o}dinger}}$ denotes the scalar potential in the Schr\"{o}dinger operator. Assume that the real part of~$V^{\text{Schr\"{o}dinger}}$ 
(respectively, real parts of $V_{11}$ and $V_{22}$) 
in the Schr\"{o}dinger  
(respectively, Dirac) case is 
(respectively, are) bounded. 
Then the principal branch of the square root is well-defined only when $\lambda\to + \infty$ in the Schr\"{o}dinger case, 
while it is also able to be valid when $\lambda\to - \infty$ in the Dirac case.

However, in this section, unless otherwise stated, we always assume that $\lambda$ is positive. The case of negative~$\lambda$ can be considered analogously (see Remark~\ref{Remark sign} below).

\subsection{General shapes of the potentials}
Let us denote
\begin{equation*}
F(x) := \int_{0}^{x} \left(
\mathrm{Im}\, V_{11}(t) + \mathrm{Im}\, V_{22}(t) 
\right)
\, \dd t.
\end{equation*}

Our main hypothesis reads as follows.
\begin{assumpA}\label{Assump General}
Let $N\in \N_{0}$, assume that $V_{11},V_{22}\in W^{N+1,\infty}_{\mathrm{loc}}(\R)$, $V_{12},V_{21}\in W^{N,\infty}_{\mathrm{loc}}(\R)$ and there exist $a_{\pm}>0$, by denoting $I^{\pm}:= \{x\in \R_{\pm}: \vert x \vert >a_{\pm}\}$, such that 
\begin{enumerate}[1)]
\item  the sum of diagonal terms of $V$ has a different asymptotic behaviour at $\pm \infty$:
\begin{equation}\label{Assump G Diagonal 1}
\begin{aligned}
&\mathrm{Im}\, V_{11}+\mathrm{Im}\, V_{22}(x)\lesssim -1 , \qquad &&\forall x\in I^{-},\\
&\mathrm{Im}\, V_{11}+\mathrm{Im}\, V_{22}(x) \gtrsim 1,\qquad &&\forall x\in I^{+},
\end{aligned}
\end{equation}
and there exist $\mu_{\pm}\in (0,1]$ such that
\begin{equation}\label{Assump G Diagonal 2}
\vert\mathrm{Im} V_{11}(x)+\mathrm{Im} V_{22}(x)\vert \geq \mu_{\pm} \left( \vert \mathrm{Im} V_{11}(x)\vert+\vert \mathrm{Im} V_{22}(x)\vert\right),\qquad \forall x\in I^{\pm};
\end{equation}
\item the primitive of the sum of off-diagonal terms $\mathcal{U}:= \mathrm{Im}\, V_{12} + \mathrm{Im}\, V_{21}$  is controlled by~$F$ at $\pm \infty$: there exist $\varepsilon_{\pm}\in \left(0, \frac{\mu_{\pm}}{2}\right)$ such that
\begin{equation}\label{Assump G Off-Diagonal 1}
\int_{0}^{x} \mathcal{U}(t) \, \dd t \leq 2\varepsilon_{\pm} F(x), \qquad \forall x\in I^{\pm};
\end{equation}
\item there exist continuous functions $f_{\pm}:I^{\pm}\to\R_{+}$ such that
\begin{equation}\label{Eq function f}
 f_{\pm}(x) \lesssim F(x), \qquad \forall x\in I^{\pm},
\end{equation}
and, for all $i\in \{1,2\}$,
\begin{align}
&\forall n \in [[1,N+1]],\qquad &&\vert V_{ii}^{(n)}(x) \vert\lesssim  f_{\pm}(x)^{n} \left\vert  V_{ii}(x)\right\vert , \qquad &&\forall x \in I^{\pm},\label{Assump G Der Vii}\\
& \forall n \in [[0,N]],\qquad &&\vert (V_{21}-V_{12})^{(n)}(x) \vert\lesssim f_{\pm}(x)^{n+1}, \qquad  &&\forall x \in I^{\pm}. \label{Assump G Der Vij}
\end{align}
\end{enumerate}
\end{assumpA}

Notice that the first condition~\eqref{Assump G Diagonal 1}
implies that
\begin{equation}\label{Eq F}
F(x) \gtrsim \vert x \vert, \qquad \forall \vert x \vert \gtrsim 1.
\end{equation}
Next lines gather some comments on Assumption \ref{Assump General}. Let us recall the expression of $\lambda\psi_{-1}^{(1)}$ when $\lambda\in \R$:
\begin{align}\label{Eq Re psi-1}
\mathrm{Re}\,\left( \lambda \psi_{-1}^{(1)}(x)\right)
&= \frac{1}{\sqrt{2}} \frac{\mathrm{Im}\, V_{11} (\lambda +m-\mathrm{Re} \,V_{22})+\mathrm{Im}\, V_{22} (\lambda -m-\mathrm{Re} \,V_{11})}{\sqrt{\vert V_{\lambda}\vert+ \mathrm{Re} \,V_{\lambda}}}.
\end{align}
We will see later that the shape of the pseudomode depends a lot on $\psi_{-1}$ and the sign of $\mathrm{Im}\, V_{11}+\mathrm{Im}\, V_{22}$ (which is attached with very large $\lambda$) will decide the sign for the decay of the pseudomode. The larger the sum is, the faster the pseudomode decreases at infinity (see the proof of Proposition \ref{Lemma L2 norm}). Furthermore, by looking at Remark~\ref{Rem.normal},  
the assumption \eqref{Assump G Diagonal 1} also ensures that the operator defined in~\eqref{Eq Def H_V} is \enquote{significantly non-self-adjoint}.

From the conditions \eqref{Assump G Diagonal 2}, we deduce the similarity of the sum of absolute values and the absolute value of the sum of $\mathrm{Im}\,V_{11}$ and $\mathrm{Im}\,V_{22}$ in the neighbourhood of infinity:
\begin{equation}\label{Eq Sim of Abso}
\vert\mathrm{Im} V_{11}(x)+\mathrm{Im} V_{22}(x)\vert \approx \vert \mathrm{Im} V_{11}(x)\vert+\vert \mathrm{Im} V_{22}(x)\vert,\qquad \forall x \in I^{\pm}.
\end{equation}
If $\mathrm{Im}\, V_{11}$, $\mathrm{Im} \,V_{22}$ have the same signs at $+ \infty$ (or $-\infty$), the condition \eqref{Assump G Diagonal 2} is obviously satisfied. Thus, these conditions guarantee that the opposite signs of $\mathrm{Im}\, V_{11}$ and $\mathrm{Im}\, V_{22}$ does not spoil the decay of the quasimode.
As for the condition~\eqref{Assump G Off-Diagonal 1}, it is easy to find $V_{12}$ and $V_{21}$ that can verify this. Indeed, since $F(x)$ is positive at $\pm \infty$, the class of all functions $V_{12}$ and $V_{21}$ such that 
$\mathcal{U}(x) \geq 0$ for all $x \lesssim -1$
and $\mathcal{U}(x) \leq 0$ for all $x \gtrsim 1$
will fulfil \eqref{Assump G Off-Diagonal 1} completely.

The conditions \eqref{Assump G Diagonal 1}, \eqref{Assump G Diagonal 2} and \eqref{Assump G Off-Diagonal 1} of Assumption \ref{Assump General} combine together to ensure the exponential decay of all the terms attached with $\xi^{(1)}$ and $\xi^{(2)}$ in \eqref{Eq L u}. 

The two conditions~\eqref{Assump G Der Vii} and \eqref{Assump G Der Vij} help us to control any wild behaviour of the derivatives of diagonal and off-diagonal terms which will be appearing in the formula of $\psi_{k}^{(1)}$ for $k\geq 0$. 

\subsection{Shapes of the cut-off functions}\label{Section: cutoff}
The role of the cut-off functions in the construction of pseudomodes is very important. Not all functions which are created from the WKB method would become the pseudomodes for the operator, since most of them do not belong to the domain of the operator. Therefore, the cut-off functions are added to complete this task. Furthermore, as discussed in Remark \ref{Remark WKB}, when $V_{11}$ and $V_{22}$ are differentiable, in order to make the $V_{\lambda}^{1/2}$ well-defined (\emph{i.e.}\ non-multi-valued) and differentiable, the condition \eqref{Eq Cond Re} need to be satisfied. It is obvious that \eqref{Eq Cond Re} will be broken if $\mathrm{Re}\, V_{11}$ or $\mathrm{Re}\, V_{22}$ is not bounded. Thus, it is necessary to employ a suitable cut-off function whose support allows \eqref{Eq Cond Re} to occur. 

Let us denote by $\xi: \R \to [0,1]$ the cut-off function satisfying the following properties
\begin{equation}\label{Eq cutoff}
\left\{
\begin{aligned}
&\xi \in C^{\infty}(\R),\\
&\xi(x)=1, \qquad \forall x \in (-\delta^{-}_{\lambda}+\Delta^{-}_{\lambda}, \delta^{+}_{\lambda}-\Delta^{+}_{\lambda}),\\
&\xi(x)=0, \qquad \forall x \in \R\setminus (-\delta^{-}_{\lambda},\delta^{+}_{\lambda}),
\end{aligned}\right.
\end{equation}
where $\delta^{\pm}_{\lambda}$ and $\Delta^{\pm}_{\lambda} <\delta^{\pm}_{\lambda}$ are $\lambda$-dependent positive numbers which will be determined later. Notice that the cut-off $\xi$ can be selected in such a way that
\begin{equation}\label{Eq cutoff inequa}
\Vert \xi^{(j)} \Vert_{L^{\infty}(\R_{\pm})} \lesssim \left(\Delta^{\pm}_{\lambda}\right)^{-j},\qquad j\in \{1,2\}.
\end{equation}
To simplify the notation, we define the following sets
$$
\begin{aligned}
&J_{\lambda}^{-}:=(-\delta^{-}_{\lambda},0],\qquad &&J_{\lambda}^{+}:=[0,\delta^{+}_{\lambda}), \\
&\widetilde{J}_{\lambda}^{-}:=(-\delta^{-}_{\lambda}+\Delta^{-}_{\lambda},0],\qquad &&\widetilde{J}_{\lambda}^{+}:=[0,\delta^{+}_{\lambda}-\Delta^{+}_{\lambda}),
\end{aligned}
\qquad \text{and } \qquad J_{\lambda}:= J_{\lambda}^{-}\cup J_{\lambda}^{+}.\
$$

The next lemma is set up to define the boundary of the cut-off functions.
\begin{lemma}\label{Lemma delta}
Let $a>0$ and let $g:[a,+\infty) \to [0,+\infty)$ be a continuous function and let $\lambda$ be a positive number, we define 
\begin{equation}\label{Eq G delta}
\delta(\lambda) : = \inf\left\{ x\geq a: g(x)=\lambda\right\} .
\end{equation}
Then $\delta(\lambda)$ can be infinite $(\inf \emptyset = +\infty)$, however, when $g$ is unbounded at $+\infty$ and for all sufficiently large $\lambda>0$, the number $\delta(\lambda)$ is finite and
\begin{equation}\label{Eq G Property delta 1}
\lim_{\lambda\to +\infty} \delta(\lambda)=+\infty.
\end{equation}
Furthermore, if $\lambda>g(a)$ then
\begin{equation}\label{Eq G Property delta 2}
 g(x)\leq \lambda,
 \qquad \forall x\in [a,\delta(\lambda)].
\end{equation}
\end{lemma}
\begin{proof}
When $g$ is unbounded at $+\infty$ and $\lambda>\min_{x\geq a} g(x)$, the number $\delta(\lambda)$ is finite. Given arbitrary $M>a$, we consider $ \lambda\geq \max_{[a,M]} g(x)+1$, then $\delta(\lambda)\geq M$, thus the unboundedness of $\delta(\lambda)$ is checked. In order to prove \eqref{Eq G Property delta 2} under the assumption that $\lambda>g(a)$, we assume opposite that there exists $x_{0}\in [a,\delta(\lambda)]$ such that $g(x_0)>\lambda$, then by the intermediate value theorem, there exists $\tilde{x}_{0}\in (a,x_{0})$ such that $g(\tilde{x}_{0})=\lambda$. This implies that $\tilde{x}_{0}\geq \delta(\lambda)$ which is a contradiction.
\end{proof}
By using Lemma \ref{Lemma delta}, we introduce the boundary of the cut-off functions
\begin{equation}\label{Eq delta Def}
\delta^{\pm}_{\lambda} : = \inf\left\{ x\geq a_{\pm}: g_{\pm}(x)=\lambda\right\} 
\end{equation}
through defining functions $g_{\pm}:[a_{\pm},+\infty) \to [0,+\infty)$ as follows
\begin{equation}\label{Eq G g 1}
g_{\pm}( x):=\max \left\{ \frac{1}{\eta} \vert \mathrm{Re} \,V_{11}+m \vert,\,\frac{1}{\eta} \vert \mathrm{Re} \,V_{22}-m \vert,\,\frac{1}{\eta}\vert \mathrm{Im}\, V_{11}- \mathrm{Im}\, V_{22}\vert,\,f_{\pm}^{\frac{2}{1-\varepsilon_{1}}}\right\}(\pm x).
\end{equation}
Here, $\varepsilon_{1}$, $\eta$ are fixed numbers such that
$0<\varepsilon_{1}<1$ 
and $0<\eta < \min\{\mu_{-},\mu_{+} \}$,
in which $\eta$ will be chosen small enough later in Lemma \ref{Lemma Decay solution}. 
\begin{remark}
The continuity of $g_{\pm}$ will be given by the continuities of $V_{11}$, $V_{22}$ (since they belong to $W^{1,\infty}_{\mathrm{loc}}(\R)$) and of $f_{\pm}$. Note that, when $g_{+}$ is bounded at $+\infty$, \emph{i.e.}, all the functions $ \mathrm{Re}\, V_{11}, \mathrm{Re}\, V_{22}, \vert \mathrm{Im}\, V_{11} - \mathrm{Im}\, V_{22} \vert$ and $f_{+}$ are bounded at $+\infty$, we have $\delta^{+}_{\lambda}=+\infty$ for all sufficiently large $\lambda>0$. In this case, we want to say that $\xi$ is constant on the positive side, \emph{i.e.} $\xi(x)=1$ for all $x\geq 0$. This remark is also the same for $g_{-}$ for the negative axis. In other words, sometimes we may not need the cut-off functions to localize the pseudomode.
\end{remark}
When $\delta_{\lambda}^{\pm}$ is finite, we define
\begin{equation}\label{Eq G Delta}
\Delta^{\pm}_{\lambda}:=\frac{1}{\delta^{\pm}_{\lambda}}.
\end{equation}
\begin{proposition}
There exists $\lambda_{0}>0$ such that for all $\lambda>\lambda_{0}$ and for all $x\in J_{\lambda}$, we have
\begin{enumerate}[i)]
\item
\begin{equation}\label{Rem Re V1122} 
\begin{aligned}
&(1-\eta)\lambda\leq  \lambda-m-\mathrm{Re}\, V_{11}(x)  \leq (1+\eta)\lambda,\\
&(1-\eta)\lambda\leq  \lambda+m-\mathrm{Re}\, V_{22}(x)  \leq (1+\eta)\lambda,\\
&\vert \mathrm{Im}\, V_{11}(x)- \mathrm{Im}\, V_{22}(x)\vert \leq \eta\lambda;
\end{aligned}
\end{equation}
\item
\begin{equation}\label{Rem V11 approx V22} 
\begin{aligned}
\vert \lambda-m-V_{11}(x) \vert \approx\vert \lambda+m-V_{22}(x) \vert.
\end{aligned}
\end{equation}
\end{enumerate}
\end{proposition} 
\begin{proof}
In case $\mathrm{Re}\, V_{11}$, $\mathrm{Re}\, V_{22}$ and $\mathrm{Im}\, V_{11}- \mathrm{Im}\, V_{22}$ are bounded at infinity, it is easy to check the above estimates. Now we assume that the unboundedness of  $\mathrm{Re}\, V_{11}$ or  $\mathrm{Re}\, V_{22}$ or $\mathrm{Im}\, V_{11}- \mathrm{Im}\, V_{22}$ at $+\infty$ occurs. 
The case of unboundness at the negative infinity is analogous. 
It follows from the estimate \eqref{Eq G Property delta 2} that, 
for all $x\in J_{\lambda}^{+}$,
\begin{equation*}
\left\{ 
\begin{aligned}
&\vert \mathrm{Re}\, V_{11} (x)+m \vert\leq \eta \lambda,\\
&\vert \mathrm{Re}\, V_{22}(x) -m \vert\leq \eta \lambda,\\
&\vert \mathrm{Im}\, V_{11}(x)- \mathrm{Im}\, V_{22}(x)\vert \leq \eta\lambda.
\end{aligned}
\right.
\end{equation*}
Consequently, the three estimates in \eqref{Rem Re V1122} follow. 
From them, we deduce that
\begin{align*}
\frac{\vert \lambda-m-V_{11}(x) \vert}{\vert \lambda+m-V_{22}(x) \vert} &\lesssim \frac{\vert\lambda-m- \mathrm{Re}\,V_{11}(x) \vert+\vert\mathrm{Im}\, V_{11}\vert}{\vert \lambda+m-\mathrm{Re}\,V_{22}(x)\vert+ \vert\mathrm{Im}\, V_{22}\vert}\\
&\leq \frac{\vert\lambda-m- \mathrm{Re}\,V_{11}(x) \vert+\vert\mathrm{Im}\, V_{11}-\mathrm{Im}\, V_{22}\vert+\vert\mathrm{Im}\, V_{22}\vert}{\vert \lambda+m-\mathrm{Re}\,V_{22}(x)\vert+ \vert\mathrm{Im}\, V_{22}\vert}\\
&\leq \frac{(1+2\eta)\lambda+\vert\mathrm{Im}\, V_{22}\vert}{(1-\eta)\lambda+ \vert\mathrm{Im}\, V_{22}\vert}\leq \frac{1+2\eta}{1-\eta}.
\end{align*}
Thus, 
$
\vert \lambda-m-V_{11}(x) \vert \lesssim \vert \lambda+m-V_{22}(x) \vert.
$
The other direction is proved analogously, 
therefore the second estimate \eqref{Rem V11 approx V22} is verified.
\end{proof}
\subsection{Auxiliary steps}\label{Subsection Auxi Steps}
The next lemma shows us that $V_{\lambda}$ and $K_{\lambda}$ 
inherit the properties of~$V$ in \eqref{Assump G Der Vii} and \eqref{Assump G Der Vij}.

\begin{lemma}\label{Lemma VK}
Let $N\in \N_{0}$, assume that $V_{11},V_{22}\in W^{N+1,\infty}_{\mathrm{loc}}(\R)$ and $V_{12},V_{21}\in W^{N,\infty}_{\mathrm{loc}}(\R)$ satisfy the assumptions \eqref{Assump G Der Vii} and \eqref{Assump G Der Vij}. There exists $\lambda_{0}>0$ such that, for all $\lambda>\lambda_{0}$ and for all $\ell\in  [[1,N+1]]$, we have
\begin{align*}
&\text{on } I^{\pm}\cap J_{\lambda}^{\pm},\qquad
&&\begin{aligned}
&\vert V_{\lambda}^{(\ell)}\vert \lesssim \vert V_{\lambda} \vert \frac{ f_{\pm}^{\ell}\max\{\vert V_{11}\vert,\vert V_{22}\vert\}}{\vert \lambda+m-V_{22} \vert},
&\vert K_{\lambda}^{(\ell)} \vert \lesssim \vert K_{\lambda} \vert  f_{\pm}^{\ell},
\end{aligned}\\
&\text{and}\\
&\text{on }[-a_{-},a_{+}],\qquad 
&&\begin{aligned}
&\vert V_{\lambda}^{(\ell)} \vert \lesssim \frac{\vert V_{\lambda} \vert}{\vert \lambda+m-V_{22}\vert},
&\vert K_{\lambda}^{(\ell)} \vert \lesssim \vert K_{\lambda}\vert.
\end{aligned}
\end{align*}
Furthermore, if $V_{12}=V_{21}$, we have
\begin{equation*}
\begin{aligned}
&\text{on } I^{\pm}\cap J_{\lambda}^{\pm},
\qquad &&\vert K_{\lambda}^{(\ell)} \vert \lesssim \vert K_{\lambda} \vert \frac{ f_{\pm}^{\ell}\vert V_{22}\vert}{\vert \lambda+m-V_{22} \vert}, \\
&\text{on } [-a_{-},a_{+}], 
\qquad &&\vert K_{\lambda}^{(\ell)} \vert \lesssim  \frac{\vert K_{\lambda} \vert}{\vert \lambda+m-V_{22}\vert}.
\end{aligned}
\end{equation*}
\end{lemma}
\begin{proof}
We can choose $\lambda_{0}>0$ satisfying $\delta_{\lambda}^{\pm}> a_{\pm}$ for all $\lambda>\lambda_{0}$, thanks to \eqref{Eq G Property delta 1}. Then, $I^{\pm}\cap J_{\lambda}^{\pm} \neq \emptyset$.
From the formula of $V_{\lambda}$, the general Leibniz rule for the $\ell$-th derivative of the product yields that
\begin{align*}
(V_{\lambda})^{(\ell)} &= \sum_{k=0}^{\ell} \begin{pmatrix}
\ell\\
k
\end{pmatrix}(\lambda-m-V_{11})^{(k)} (\lambda+m-V_{22})^{(\ell-k)}\\ 
&=-(\lambda-m-V_{11})V_{22}^{(\ell)}+\sum_{k=1}^{\ell-1} \begin{pmatrix}
\ell\\
k
\end{pmatrix} V_{11}^{(k)}V_{22}^{(\ell-k)} -(\lambda+m-V_{22})V_{11}^{(\ell)}\,.
\end{align*}
From the assumption \eqref{Assump G Der Vii}, we obtain the estimate on $I^{\pm}\cap J_{\lambda}^{\pm}$,
\begin{equation}\label{Eq Estimate V}
\begin{aligned}
\frac{\left\vert V_{\lambda}^{(\ell)}\right\vert}{\left\vert V_{\lambda}\right\vert} &\lesssim \frac{ f_{\pm}^{\ell} \vert V_{22}\vert}{\vert \lambda+m-V_{22}\vert}+  \sum_{k=1}^{\ell-1} \frac{ f_{\pm}^{\ell}\vert V_{11}\vert \vert V_{22} \vert}{\vert \lambda-m-V_{11} \vert \vert \lambda+m-V_{22} \vert}+\frac{f_{\pm}^{\ell}\vert V_{11}\vert}{\vert \lambda-m-V_{11}\vert}\\
&\lesssim \frac{ f_{\pm}^{\ell}\max\{\vert V_{11}\vert,\vert V_{22}\vert\}}{\vert \lambda+m-V_{22} \vert}\,,
\end{aligned}
\end{equation}
where in the last step, we used \eqref{Rem V11 approx V22} and the fact that (with some large $\lambda_{0}$), 
\begin{align*}
\frac{ \vert V_{11}\vert }{\vert \lambda-m-V_{11} \vert}\lesssim \frac{\vert\mathrm{Re}\,V_{11}\vert+\vert\mathrm{Im}\, V_{11}\vert}{\vert\lambda-m-\mathrm{Re}\, V_{11}\vert+\vert\mathrm{Im} V_{11}\vert}\leq \frac{\eta\lambda+m+\vert\mathrm{Im}\, V_{11}\vert}{(1-\eta)\lambda+\vert\mathrm{Im}\, V_{11}\vert} \lesssim 1.
\end{align*}
Next, we prove the estimate for $K_{\lambda}$. Let us recall that
\[ K_{\lambda} =  \frac{1}{\lambda+m-V_{22}}e^{u},\qquad \text{ with } \displaystyle u(x):=-i\int_{0}^{x}(V_{21}-V_{12}) (\tau) \,\dd \tau.\]
The Leibniz rule also leads us to
\begin{align*}
(K_{\lambda})^{(\ell)} = \sum_{k=0}^{\ell} \begin{pmatrix}
\ell\\
k
\end{pmatrix} \left(\frac{1}{\lambda+m-V_{22}} \right)^{(k)}\left(e^{u}\right)^{(\ell-k)}.
\end{align*}
Using Fa\`{a} di Bruno's formula for the derivative of a composition of two functions (see \cite{S80}), we get
\begin{equation*}
\left\{
\begin{aligned}
\left(\frac{1}{\lambda+m-V_{22}} \right)^{(k)} &= \frac{1}{\lambda+m-V_{22}}\sum_{1\alpha_1+2\alpha_2+\cdots+k \alpha_{k}=k} \frac{k!}{\alpha_1!\alpha_2!\ldots\alpha_{k}!} \prod_{j=1}^{k} \left(\frac{V_{22}^{(j)}}{j!(\lambda+m-V_{22})}\right)^{\alpha_{j}},\\
\left(e^{u}\right)^{(\ell-k)} &= e^{u} \sum_{1\beta_1+2\beta_2+\cdots+(\ell-k) \beta_{\ell-k}=\ell-k} \frac{(\ell-k)!}{\beta_1!\beta_2!\ldots\beta_{\ell-k}!}\prod_{j=1}^{\ell-k} \left(\frac{u^{(j)}}{j!}\right)^{\beta_{j}},
\end{aligned}
\right.
\end{equation*}
where  $(\alpha_{j})_{1\leq j \leq k}$ and $(\beta_{j})_{1\leq j\leq \ell-k}$ are non-negative integers.

From the assumption \eqref{Assump G Der Vii} for~$V_{22}$ and \eqref{Assump G Der Vij}, we obtain the estimate on $I^{\pm} \cap J_{\lambda}^{\pm}$,
\begin{equation}\label{Eq Estimate K}
\begin{aligned}
\vert K_{\lambda}^{(\ell)} \vert 
&\lesssim \sum_{k=0}^{\ell} \vert K_{\lambda} \vert\left(\sum_{1\alpha_1+2\alpha_2+\cdots+\ell\alpha_{k}=k} \frac{k!}{\alpha_1!\alpha_2!\ldots\alpha_{k}!} \prod_{j=1}^{k} \left(\frac{\vert V_{22}\vert f_{\pm}^{j}}{j!\vert\lambda+m-V_{22}\vert}\right)^{\alpha_{j}} \right)\\
&\hspace{2 cm} \times \left( \sum_{1\beta_1+2\beta_2+\cdots+(\ell-k) \beta_{\ell-k}=\ell-k} \frac{(\ell-k)!}{\beta_1!\beta_2!\ldots\beta_{\ell-k}!}\prod_{j=1}^{\ell-k} \left(\frac{ f_{\pm}^{j}}{j!}\right)^{\beta_{j}}\right)\\
&\lesssim  \vert K_{\lambda} \vert  f_{\pm}^{\ell}.
\end{aligned}
\end{equation}

If $V_{12}=V_{21}$, then $u=0$ and we have, on $I^{\pm} \cap J_{\lambda}^{\pm}$ and for all $\ell\geq 1$,
\begin{align*}
\vert K_{\lambda}^{(\ell)} \vert &= \left\vert \frac{1}{\lambda+m-V_{22}}\sum_{1\alpha_1+2\alpha_2+\cdots+\ell \alpha_{\ell}=\ell} \frac{\ell!}{\alpha_1!\alpha_2!\ldots\alpha_{\ell}!} \prod_{j=1}^{\ell} \left(\frac{V_{22}^{(j)}}{j!(\lambda+m-V_{22})}\right)^{\alpha_{j}}\right\vert \\
&\lesssim \vert K_{\lambda}\vert \sum_{1\alpha_1+2\alpha_2+\cdots+\ell \alpha_{\ell}=\ell} \prod_{j=1}^{\ell} \left(\frac{\vert V_{22}\vert f_{\pm}^{j}}{\vert\lambda+m-V_{22}\vert}\right)^{\alpha_{j}}\\
&\lesssim \vert K_{\lambda}\vert f_{\pm}^{\ell}  \frac{\vert V_{22}\vert}{\vert \lambda+m-V_{22}\vert}.
\end{align*}
The last step is to the bound $\frac{\vert V_{22}\vert}{\vert \lambda+m-V_{22}\vert} \lesssim 1$, and the fact that $ \sum_{j=1}^{\ell} \alpha_{j} \geq 1$.

All the estimates for $x\in[-a_{-},a_{+}]$ hold thanks to the boundedness of the appearing derivatives of $V$ on a compact set.
\end{proof}
We use the next lemma to gather all the real parts of the diagonal terms to one group and their imaginary parts to the other group. This allows us to estimate the denominator of $ \mathrm{Re} \, (\lambda\psi_{-1}^{(1)})$ in an easier way. Furthermore, it also tells us that the case $\mathrm{Im}\, V_{11} = \mathrm{Im}\, V_{22}$ is very special.
\begin{lemma}
On $J_{\lambda}$, we have the following inequalities
\begin{equation}\label{Eq bound denominator}
\begin{aligned}
&\sqrt{\vert V_{\lambda} \vert + \mathrm{Re}\, V_{\lambda}} \geq \sqrt{2} \sqrt{(\lambda-m-\mathrm{Re}\, V_{11})(\lambda+m-\mathrm{Re}\, V_{22})},\\
&\sqrt{\vert V_{\lambda} \vert + \mathrm{Re}\, V_{\lambda}}\leq \frac{1}{\sqrt{2}}\sqrt{(\mathrm{Im} \, V_{11}-\mathrm{Im} \, V_{22})^2+(2\lambda-\mathrm{Re}\, V_{11}-\mathrm{Re}\, V_{22})^2}.
\end{aligned}
\end{equation}
\end{lemma}
\begin{proof}
Using the Cauchy--Schwarz inequality, we have
\begin{align*}
\vert V_{\lambda} \vert + \mathrm{Re}\, V_{\lambda}
&= \sqrt{[(\lambda-m-\mathrm{Re}\, V_{11})^2+ (\mathrm{Im}\, V_{11})^2][(\lambda+m-\mathrm{Re}\, V_{22})^2+ (\mathrm{Im}\, V_{22})^2]}\\
& \qquad +(\lambda+m-\mathrm{Re}\, V_{22})(\lambda-m-\mathrm{Re}\, V_{11})-(\mathrm{Im}\, V_{11})(\mathrm{Im}\, V_{22})\\
&\geq \left\vert(\lambda+m-\mathrm{Re}\, V_{22})(\lambda-m-\mathrm{Re}\, V_{11})+(\mathrm{Im}\, V_{11})(\mathrm{Im}\, V_{22}) \right\vert\\
& \qquad +(\lambda+m-\mathrm{Re}\, V_{22})(\lambda-m-\mathrm{Re}\, V_{11})-(\mathrm{Im}\, V_{11})(\mathrm{Im}\, V_{22})\\
&\geq 2(\lambda+m-\mathrm{Re}\, V_{22})(\lambda-m-\mathrm{Re}\, V_{11}).
\end{align*}
By an elementary inequality, 
the modulus of $V_{\lambda}$ can be bounded from above as follows:
\begin{align*}
\vert V_{\lambda} \vert &= \sqrt{[(\lambda-m-\mathrm{Re}\, V_{11})^2+ (\mathrm{Im}\, V_{11})^2][(\lambda+m-\mathrm{Re}\, V_{22})^2+ (\mathrm{Im}\, V_{22})^2]}\\
&\leq \frac{1}{2} \left((\lambda-m-\mathrm{Re}\, V_{11})^2+ (\mathrm{Im}\, V_{11})^2+(\lambda+m-\mathrm{Re}\, V_{22})^2+ (\mathrm{Im}\, V_{22})^2 \right).
\end{align*}
From this we can deduce successively that
\begin{equation*}
\vert V_{\lambda}\vert + \mathrm{Re}\, V_{\lambda} \leq \frac{1}{2} \left( \left(\mathrm{Im}\, V_{11}-\mathrm{Im}\, V_{22}\right)^2+\left(2\lambda-\mathrm{Re}\, V_{11}-\mathrm{Re}\, V_{22}\right)^2 \right).
\end{equation*}
\vspace{-6ex} \\
\end{proof}

\begin{lemma}\label{Lemma WKB sol}
Let Assumption \ref{Assump General} hold for some $N\in \N_{0}$. Let $n\in [[0,N]]$ and $\{\psi_{k}^{(1)}\}_{k\in[[-1,n-1]]}$ be determined by \eqref{Eq Recursion} with the plus sign in the formula of $\psi_{-1}^{(1)}$. There exists $\lambda_{0}>0$ such that, for all $\lambda>\lambda_{0}$,
\begin{equation}\label{Eq Re psi-1'}
\begin{aligned}
&\text{on } I^{+}\cap J_{\lambda}^{+}, \qquad &&\mathrm{Re} \,(\lambda \psi_{-1}^{(1)} )\geq  \frac{\mu_{+}-\eta}{\sqrt{\eta^2+(2+2\eta)^2}} \left(\mathrm{Im}\, V_{11} +  \mathrm{Im}\, V_{22} \right),\\
&\text{on } I^{-}\cap J_{\lambda}^{-}, \qquad && \mathrm{Re} \,(\lambda \psi_{-1}^{(1)} )\leq \frac{\mu_{-}-\eta}{\sqrt{\eta^2+(2+2\eta)^2}} \left(\mathrm{Im}\, V_{11} +  \mathrm{Im}\, V_{22} \right),\\
&\text{on } J_{\lambda}, \qquad && \vert \mathrm{Re} \,(\lambda\psi_{-1}^{(1)}) \vert\lesssim  \vert  \mathrm{Im} V_{11}\vert +  \vert\mathrm{Im} V_{22} \vert,
\end{aligned}
\end{equation}
and for all $k\in [[0,n-1]]$,
\begin{equation}\label{Eq psi k'}
\begin{aligned}
&\text{on } I^{\pm}\cap J_{\lambda}^{\pm}, \qquad &&\left\vert \lambda^{-k} \psi_{k}^{(1)} \right\vert \lesssim \frac{f_{\pm}^{k+1}}{\lambda^{k}},\\
&\text{on }[-a_{-},a_{+}], \qquad &&\left\vert \lambda^{-k} \psi_{k}^{(1)} \right\vert \lesssim \frac{1}{\lambda^{k}}.
\end{aligned}
\end{equation}
\end{lemma}
\begin{proof}
Firstly, we prove the lemma for the first two estimates in \eqref{Eq Re psi-1'}. By looking at the formula of $ \mathrm{Re}\, (\lambda\psi_{-1}^{(1)})$ in \eqref{Eq Re psi-1} and using assumptions \eqref{Assump G Diagonal 1} and \eqref{Assump G Diagonal 2}, we see that the numerator of $\mathrm{Re}\, (\lambda\psi_{-1}^{(1)})$ has opposite signs at $-\infty$ and $+\infty$. Namely, employing the remark~\eqref{Rem Re V1122} 
and recalling that $\eta<\mu_{\pm}$, 
we have the following estimate on $I^{+}\cap J_{\lambda}^{+}$:
\begin{align*}
\lefteqn{\mathrm{Im} \, V_{11} (\lambda+m-\mathrm{Re}\, V_{22})+\mathrm{Im} \, V_{22} (\lambda-m-\mathrm{Re}\, V_{11})}&\\
&= \ \lambda ( \mathrm{Im} \, V_{11} +\mathrm{Im} \, V_{22} )+ \mathrm{Im} \, V_{11} (m-\mathrm{Re}\, V_{22})- \mathrm{Im} \, V_{22} (m+\mathrm{Re}\, V_{11})
\\
&\geq  \ \lambda \mu_{+} (\vert \mathrm{Im} \, V_{11} \vert + \vert\mathrm{Im} \, V_{22} \vert)-\eta \lambda (\vert \mathrm{Im} \, V_{11}\vert +\vert\mathrm{Im} \, V_{22} \vert)\\
&\geq \ \lambda(\mu_{+}-\eta)(\mathrm{Im} \, V_{11} +\mathrm{Im} \, V_{22} );
\end{align*}
and similarly, on $I^{-}\cap J_{\lambda}^{-}$:
\[ \mathrm{Im} \, V_{11} (\lambda+m-\mathrm{Re}\, V_{22})+ \mathrm{Im} \, V_{22} (\lambda-m-\mathrm{Re}\, V_{11}) \leq \lambda(\mu_{-}-\eta)\left( \mathrm{Im} \, V_{11} +  \mathrm{Im} \, V_{22} \right).\]
Next, it follows from the upper bound in \eqref{Eq bound denominator} for the denominator of $ \mathrm{Re}\,(\lambda\psi_{-1}^{(1)})$ in \eqref{Eq Re psi-1} that, on $I^{+}\cap J_{\lambda}^{+}$,
\begin{align*}
\mathrm{Re}\,( \lambda\psi_{-1}^{(1)})&\geq \frac{\lambda(\mu_{+}-\eta)( \mathrm{Im}\, V_{11}+ \mathrm{Im}\, V_{22})}{\left((\mathrm{Im}\, V_{11}-\mathrm{Im}\, V_{22})^2+(2\lambda-\mathrm{Re}\, V_{11}- \mathrm{Re}\,V_{22})^2 \right)^{1/2}}\\
 &\geq \frac{\mu_{+}-\eta}{\sqrt{\eta^2+(2+2\eta)^2}} \left( \mathrm{Im}\, V_{11}+\mathrm{Im}\, V_{22} \right).
\end{align*}
In the last step of the above expression, we used the estimates \eqref{Rem Re V1122}. On $I^{-}\cap J_{\lambda}^{-}$, we do it in the same manner.

Secondly, the final estimate in \eqref{Eq Re psi-1'} is obtained by using the lower bound in \eqref{Eq bound denominator} for the denominator of $ \mathrm{Re}\,(\lambda \psi_{-1}^{(1)})$ on $J_{\lambda}$:
\begin{align*}
\left\vert  \mathrm{Re}\,(\lambda \psi_{-1}^{(1)})\right\vert \lesssim \frac{\vert \mathrm{Im} V_{11}(\lambda+m- \mathrm{Re}\, V_{22} )\vert+ \vert \mathrm{Im} V_{22}(\lambda-m- \mathrm{Re}\, V_{11} )\vert}{\vert (\lambda+m-\mathrm{Re}\, V_{22})(\lambda-m-\mathrm{Re}\, V_{11}) \vert^{1/2}}\nonumber \lesssim  \vert  \mathrm{Im} V_{11} \vert + \vert  \mathrm{Im} V_{22} \vert .
\end{align*}

Finally, let us prove the first estimate in \eqref{Eq psi k'} while the second one can be considered in a similar way. For $x\in I^{\pm} \cap J_{\lambda}^{\pm}$, from Lemma \ref{Lemma VK}, we have
$ \vert V_{\lambda}^{(\ell)}(x) \vert \lesssim \vert V_{\lambda}(x) \vert f_{\pm}^{\ell}$ and $\vert K_{\lambda}^{(\ell)}(x) \vert \lesssim \vert K_{\lambda}(x) \vert f_{\pm}^{\ell}.$
Here we used the fact implied by \eqref{Rem V11 approx V22} that, for all $x\in J_{\lambda}$,
\begin{equation}\label{Eq max}
\frac{\max\{\vert V_{11}(x) \vert, \vert V_{22}(x) \vert \}}{\vert \lambda+m-V_{22}\vert} \leq \frac{\vert V_{11}(x) \vert+ \vert V_{22}(x) \vert }{\vert \lambda+m-V_{22}\vert}\lesssim  \frac{\vert V_{11}(x) \vert}{\vert \lambda+m-V_{11}\vert}+\frac{\vert V_{22}(x) \vert}{\vert \lambda-m-V_{22}\vert}\lesssim 1.
\end{equation}
Thus, applying this to control each element $d_{j}^{r,r-j+(1,1)}(V_{\lambda},K_{\lambda})$ on $I^{\pm} \cap J_{\lambda}^{\pm}$:
\begin{align*}
\lefteqn{
\displaystyle\vert d_{j}^{r,r-j+(1,1)}(V_{\lambda},K_{\lambda}) \vert
}
\\
&\lesssim \sum_{\alpha \in \mathcal{I}_{j}^{r,r-j+(1,1)}}\left\vert V_{\lambda}^{(1)}\right\vert^{\alpha_1}\ldots\left\vert V_{\lambda}^{(r_{1}-j_1+1)}\right\vert^{\alpha_{r_{1}-j_1+1}} \left\vert K_{\lambda}^{(1)} \right\vert^{\alpha_{r_{1}-j_1+2}}\ldots \left\vert K_{\lambda}^{(r_2-j_2+1)}\right\vert^{\alpha_{\vert r \vert-\vert j \vert +2}}\\
&\lesssim \sum_{\alpha \in \mathcal{I}_{j}^{r,r-j+(1,1)}}\left\vert V_{\lambda}\right\vert^{\displaystyle\sum_{p=1}^{r_{1}-j_1+1} \alpha_{p}}  \left\vert K_{\lambda} \right\vert^{\displaystyle\sum_{p=1}^{r_{2}-j_2+1} \alpha_{r_{1}-j_1+1+p}}  f_{\pm}^{\displaystyle\sum_{p=1}^{r_{1}-j_1+1} p\alpha_{p}+ \sum_{p=1}^{r_{2}-j_2+1} p\alpha_{r_{1}-j_1+1+p}}\\
&\lesssim \left\vert V_{\lambda}\right\vert^{j_1}  \left\vert K_{\lambda} \right\vert^{j_2}  f_{\pm}^{\vert r \vert},
\end{align*}
in which we borrowed the definition of the set $\mathcal{I}_{j}^{r}$ in \eqref{Notation Ijr}. The estimate in \eqref{Eq psi k'} for $x\in I^{\pm}\cap J_{\lambda}^{\pm}$ follows from the formula of $\psi_{k}^{(1)}$ in Lemma \ref{Lemma psi k'}. 
\end{proof}

\begin{remark}\label{Remark Re Psi}
From the estimates \eqref{Eq Re psi-1'} and \eqref{Eq Sim of Abso}, 
it follows that, for all $x\in I^{\pm}\cap J_{\lambda}^{\pm}$,
\[ \text{\textbf{In the Dirac case:} } \quad
\mathrm{Re} \,(\lambda\psi_{-1}^{(1)} (x))\approx  \mathrm{Im} V_{11}(x)  +  \mathrm{Im} V_{22}(x) .\]
The sign of the sum $\mathrm{Im} V_{11}(x)  +  \mathrm{Im} V_{22}(x)$ decides the sign of $ \mathrm{Re}\, (\lambda\psi_{-1}^{(1)})$ in the neighbourhood of infinity. 

This is to be compared with the Schr\"{o}dinger case 
\cite[Lem.~3.4]{KS19} where the sign of $\mathrm{Im} \, V$ 
(with scalar~$V$ now)
plays this role, more precisely
\[ \text{\textbf{In the Schr\"{o}dinger case:} } \quad
\mathrm{Re} \,(\lambda\psi_{-1}^{(1)} (x))\approx 
 \lambda^{-\frac{1}{2}} \,  \mathrm{Im}\, V.\]
In this case, when $\lambda$ is considered to be large, $\mathrm{Im}\, V$ needs to be proportional to (and larger than) $\lambda^{\frac{1}{2}}$ near $\delta_{\pm}$ such that $\mathrm{Re} \,(\lambda\psi_{-1}^{(1)} (x))$ is also large. This was handled in~\cite{KS19}
thanks to the definition of $\delta_{\pm}$ which is in terms of $\mathrm{Im}\, V$. Then $\mathrm{Re}\, V$ needs to be controlled by $\mathrm{Im}\, V$ such that $\mathrm{Re}\, V$ can also be bounded by~$\lambda$,
whence the extra condition \cite[Cond.~(3.3)]{KS19}.

However, this extra work can been relaxed in the Dirac case thanks to the above form of $\mathrm{Re} \,(\lambda\psi_{-1}^{(1)})$. Technically, this can be explained by the product structure of $V_{\lambda}^{\text{Dirac}}$ which allows $\lambda$ to show up in $\mathrm{Im} \, V_{\lambda}^{\text{Dirac}}$ and therefore it cancels $\lambda$ appearing in the denominator of $ \mathrm{Re}\, (\lambda\psi_{-1}^{(1)})$ asymptotically.

Furthermore, the case $\mathrm{Im}\, V_{11} =\mathrm{Im}\, V_{22}$ needs to be taken into account. For example, if this happens on $[0,+\infty)$ (obviously, $\mu_{+}=1$ will be chosen in this situation), then the first estimate in \eqref{Eq Re psi-1'} 
can be taken strictly as follows:
\begin{align*}
\mathrm{Re}\,\left( \lambda \psi_{-1}^{(1)}(x)\right)= \frac{1}{\sqrt{2}} \frac{\mathrm{Im}\, V_{11} (2\lambda -\mathrm{Re} \,V_{22} -\mathrm{Re} \,V_{11})}{\sqrt{\vert V_{\lambda}\vert+ \mathrm{Re} \,V_{\lambda}}}\geq \frac{\mathrm{Im}\, V_{11} (2\lambda -\mathrm{Re} \,V_{22} -\mathrm{Re} \,V_{11})}{ 2\lambda -\mathrm{Re} \,V_{22} -\mathrm{Re} \,V_{11}}= \mathrm{Im}\, V_{11}.
\end{align*}
Meanwhile, the constant which turns up at \eqref{Eq Re psi-1'}, $\frac{2(\mu_{+}-\eta)}{\sqrt{\eta^2+(2+2\eta)^2}}$, is strictly smaller than and close to $1$ when $\eta$ is chosen small enough. However, it does not matter because this constant will be attached with $1-o(1)$ as $\lambda \to +\infty$ when we deal with it in the next lemma.
\end{remark}

With the derivatives of $\psi_{k}^{(1)}$ given in \eqref{Eq Recursion}, we can determine the primitives $\psi_{k}$ uniquely by choosing the initial data 
$\psi_{k}(0)=0$, $\forall k\in [[-1,n-1]].$
\begin{lemma}\label{Lemma Decay solution}
Let Assumption \ref{Assump General} hold for some $N\in \N_{1}$. Let $n\in [[1,N]]$ and $\{\psi_{k}^{(1)}\}_{k\in[[-1,n-1]]}$ be determined by \eqref{Eq Recursion} with the plus sign in the formula of $\psi_{-1}^{(1)}$. Let $P_{\lambda, n}$ defined as in ~\eqref{Eq Pn}. There exist $c_{1}>0$, $c_{2}>0$ and $\lambda_{0}>0$ such that, for all $\lambda>\lambda_{0}$ and for all $x\in J_{\lambda}$,
\begin{equation*}
\exp\left( -c_1F(x)  \right)\exp\left(\frac{1}{2} \int_{0}^{x} \mathcal{U}(t)\, \dd t\right) \lesssim \left\vert k_1(x) \exp(-P_{\lambda,n}(x))\right\vert\lesssim \exp\left( -c_{2}F(x)\right).
\end{equation*}
Furthermore, if $V_{12}=V_{21}$, the statement is also true as $n=0$, \emph{i.e.} $N=0$ is allowed.
\end{lemma}
\begin{proof}
Let us recall that
\[ P_{\lambda,n}(x) = \int_{0}^{x} \psi_{0}^{(1)}(t)\, \dd t+ \sum_{\substack{k\neq 0\\
k=-1}}^{n-1}  \int_{0}^{x} \lambda^{-k}\psi_{k}^{(1)}(t) \, \dd t.\]
From the formula of $\psi_{0}^{(1)}$, we observe that
\begin{align*}
\left \vert \exp \left( - \int_{0}^{x} \psi_{0}^{(1)}(t) \dd t\right)  \right\vert&= \frac{\vert V_{\lambda}(0) \vert^{1/4}}{\vert V_{\lambda}(x) \vert^{1/4}}\frac{\vert K_{\lambda}(0)\vert^{1/2}}{\vert K_{\lambda}(x)\vert^{1/2}}.
\end{align*}
Then, it follows from the definition of the functions $V_{\lambda}$, $K_{\lambda}$ and estimates in \eqref{Rem V11 approx V22} that
\begin{align*}
&\left\vert k_1(x) \exp(-P_{\lambda,n}(x))\right\vert \\
&=\vert k_1(x) \vert  \frac{\vert V_{\lambda}(0) \vert^{1/4}}{\vert V_{\lambda}(x) \vert^{1/4}}\frac{\vert K_{\lambda}(0)\vert^{1/2}}{\vert K_{\lambda}(x)\vert^{1/2}} \exp\left(- \sum_{\substack{k=-1 \\k\neq 0}}^{n-1}  \int_{0}^{x} \mathrm{Re} \,(\lambda^{-k}\psi_{k}^{(1)}(t))\, \dd t\right)\\
&= \vert k_1(x) k_2(x) \vert^{1/2}\frac{\vert \lambda-m-V_{11}(0) \vert^{1/4}}{\vert \lambda+m-V_{22}(0)\vert^{1/4}}\frac{\vert \lambda+m-V_{22}(x) \vert^{1/4}}{\vert \lambda-m+V_{11}(x)\vert^{1/4}}\exp\left( -\sum_{\substack{k=-1 \\k\neq 0}}^{n-1}   \int_{0}^{x}\mathrm{Re} \,(\lambda^{-k} \psi_{k}^{(1)}(t))\, \dd t\right)\\
&\approx\exp\left(-\int_{0}^{x} \mathrm{Re} \,(\lambda \psi_{-1}^{(1)}(t))\, \dd t+ \frac{1}{2}\int_{0}^{x} \mathcal{U}(t)\, \dd t- \sum_{k=1}^{n-1}  \int_{0}^{x} \mathrm{Re} \,(\lambda^{-k} \psi_{k}^{(1)}(t))\, \dd t\right)\,.
\end{align*}
Thanks to the estimate \eqref{Eq Re psi-1'} and \eqref{Eq psi k'}, $\mathcal{U}\in L_{\mathrm{loc}}^{\infty}(\R)$, we have the uniform bound, for all $x\in [0,a_{+}]$,
 \[ \left\vert \int_{0}^{x} \mathrm{Re} \,(\lambda\psi_{-1}^{(1)}(t))\, \dd t \right\vert + \left\vert \frac{1}{2}\int_{0}^{x} \mathcal{U}(t)\, \dd t \right\vert+ \left\vert \sum_{k=1}^{n-1}  \int_{0}^{x} \mathrm{Re} \,(\lambda^{-k}\psi_{k}^{(1)}(t))\, \dd t \right\vert \lesssim 1.\]
By the estimate \eqref{Eq Re psi-1'} and \eqref{Eq psi k'} again, there exists a constant $M>0$ such that for all $x\in I^{+}\cap J_{\lambda}^{+}$,
\begin{align*}
&-\int_{0}^{x} \mathrm{Re} \,(\lambda \psi_{-1}^{(1)}(t))\, \dd t+ \frac{1}{2}\int_{0}^{x} \mathcal{U}(t)\, \dd t- \sum_{k=1}^{n-1}  \int_{0}^{x} \mathrm{Re} \,(\lambda^{-k}\psi_{k}^{(1)}(t))\, \dd t\\
\leq & -\frac{\mu_{+}-\eta}{\sqrt{\eta^2+(2+2\eta)^2}} F(x) + \frac{1}{2}\int_{0}^{x} \mathcal{U}(t)\, \dd t- \sum_{k=1}^{n-1}  \int_{a_{+}}^{x} \mathrm{Re} \,(\lambda^{-k}\psi_{k}^{(1)}(t))\, \dd t+M\\
= &-\frac{\mu_{+}-\eta}{\sqrt{\eta^2+(2+2\eta)^2}} F(x) \left(1- \frac{\frac{1}{2}\int_{0}^{x} \mathcal{U}(t)\, \dd t}{\frac{\mu_{+}-\eta}{\sqrt{\eta^2+(2+2\eta)^2}} F(x)}  + \frac{\displaystyle\sum_{k=1}^{n-1}  \int_{a_{+}}^{x} \mathrm{Re} \,(\lambda^{-k}\psi_{k}^{(1)}(t))\, \dd t}{\frac{\mu_{+}-\eta}{\sqrt{\eta^2+(2+2\eta)^2}} F(x)}\right)+M.
\end{align*}
To estimate the term with $\mathcal{U}$, we use the condition \eqref{Assump G Off-Diagonal 1}. Let $\varepsilon_{+}\in \left(0, \frac{\mu_{+}}{2}\right)$ be the number given in the condition \eqref{Assump G Off-Diagonal 1}. 
We choose $\eta$ very small such that
\[ \frac{\mu_{+}-\eta}{\sqrt{\eta^2+(2+2\eta)^2}} > \varepsilon_{+}.\]
We deduce that, for all $x\in I^{+}\cap J_{\lambda}^{+}$,
\[\frac{\displaystyle\frac{1}{2}\int_{0}^{x} \mathcal{U}(t)\, \dd t}{\frac{\displaystyle\mu_{+}-\eta}{\sqrt{\eta^2+(2+2\eta)^2}}F(x) } \leq \frac{\varepsilon_{+}}{\displaystyle\frac{\mu_{+}-\eta}{\sqrt{\eta^2+(2+2\eta)^2}} }<1 .\]
To bound the terms with $\psi_{k}$ for $k \in [[1,n-1]]$, we apply \eqref{Eq psi k'} for all $t\in I^{+}\cap J_{\lambda}^{+}$
and obtain
\begin{equation}\label{Eq sum estimate}
\left\vert \sum_{k=1}^{n-1}  \mathrm{Re}\, (\lambda^{-k}\psi_{k}^{(1)}(t)) \right\vert \lesssim \sum_{k=1}^{n-1} \frac{ f_{+}(t)^{k+1}}{\lambda^{k}}\lesssim  \left\{ \begin{aligned}
& \lambda^{-1} &&\quad \text{if } f_{+} \text{ is bounded at }+\infty,\\
& \lambda^{-\varepsilon_{1}} &&\quad \text{if } f_{+} \text{ is unbounded at }+\infty.
\end{aligned}\right.
\end{equation}
Indeed, the case $f_{+}$ is bounded at $+\infty$ is obvious. In contrast, we employ the property \eqref{Eq G Property delta 2} of $\delta^{+}_{\lambda}$ and the definition of function $g$ in \eqref{Eq G g 1}, for all $t \in I^{+}\cap J_{\lambda}^{+}$,
\begin{align*}
\frac{ f_{+}(t)^{k+1}}{\lambda^{k}}  \leq \frac{\lambda^{\frac{k+1}{2}(1-\varepsilon_{1})}}{\lambda^{k}},
\end{align*}
and notice that $\frac{k+1}{2}(1-\varepsilon_{1})-k\leq -\varepsilon_{1}$ for all $k\geq 1$. By employing \eqref{Eq F}, we have
\begin{equation}\label{Eq sum psi k}
\frac{\left\vert \displaystyle\sum_{k=1}^{n-1}  \int_{a_{+}}^{x} \mathrm{Re} \,(\lambda^{-k}\psi_{k}^{(1)}(t))\, \dd t \right\vert}{ F(x)} = o(1)\qquad \text{as } \lambda \to +\infty.
\end{equation}
Hence, the second inequality in the statement of this lemma is proved. The first inequality is obtained easily by the final estimate in \eqref{Eq Re psi-1'}, the observation \eqref{Eq Sim of Abso} and the selected sign of the sum $\mathrm{Im}\, V_{11}+\mathrm{Im}\, V_{22}$ in \eqref{Assump G Diagonal 1} for all $x\in I^{+}\cap J_{\lambda}^{+}$,
$
\mathrm{Re}\, (\lambda\psi_{-1}^{(1)})\lesssim \mathrm{Im}\, V_{11}+\mathrm{Im}\, V_{22}.
$
Finally, combining this with \eqref{Eq sum psi k}, we obtain the result.

When $n=0$, there is no presence of $\psi_{0}^{(1)}$ in $P_{\lambda,n}$ and thus the integral $\int_{0}^{x} \frac{\mathcal{U}(t)}{2}\, \dd t$ will not come out in the above estimates, but the integral $\int_{0}^{x} \mathrm{Im}\, V_{21}(t)\, \dd t$ appears instead. However, if $V_{12}=V_{21}$, we can perform the proof as we have done above.
\end{proof}
The next proposition reveals that the terms attached with the derivatives of the cut-off function $\xi$ decay exponentially as $\lambda\to+ \infty$ at a rate controlled by the function $F$.
\begin{proposition}\label{Lemma L2 norm}
Let Assumption \ref{Assump General} hold for some $N\in \N_{1}$. Let $n\in [[1,N]]$, $\{\psi_{k}^{(1)}\}_{k\in[[-1,n-1]]}$ be determined by \eqref{Eq Recursion} with the plus sign in the formula of $\psi_{-1}^{(1)}$ and $P_{\lambda,n}$ defined as in \eqref{Eq Pn}. Let $\xi$ be given in \eqref{Eq cutoff} whose $\delta^{\pm}_{\lambda}$, $g_{\pm}$ and $\Delta^{\pm}_{\lambda}$ is identified by \eqref{Eq delta Def}, \eqref{Eq G g 1} and \eqref{Eq G Delta}. Let us denote
\begin{equation}\label{Eq L^2 xi''}
\kappa(\lambda):= \frac{\left\Vert \frac{k_1\exp(-P_{\lambda,n})}{\lambda+m-V_{22}}  \xi^{(2)} \right\Vert+\left\Vert \frac{k_1\exp(-P_{\lambda,n})}{\lambda+m-V_{22}} \left( 2P_{\lambda,n}^{(1)}-\frac{K_{\lambda}^{(1)}}{K_{\lambda}}\right) \xi^{(1)} \right\Vert}{\sqrt{\Vert k_1 \exp(-P_{\lambda,n})\xi \Vert^2+ \left\Vert \frac{k_1 (-i\partial_{x})(\exp(-P_{\lambda,n})\xi)}{\lambda+m-V_{22}}  \right\Vert^2}}.
\end{equation}
Then $\kappa(\lambda) = o(1)$ as $\lambda\to +\infty$. More precisely, there exists $\lambda_{0}>0$ such that, for all $\lambda>\lambda_{0}$,
$$\kappa(\lambda)=\kappa_{-}(\lambda)+\kappa_{+}(\lambda)$$
where (with some $d_1,d_2>0$)
\begin{equation*}
\kappa_{+}(\lambda) = \left\{ \begin{aligned}
&0, && \text{if }g_{+} \text{ is bounded at} +\infty,\\
&\mathcal{O}\left(\exp(-d_1F(d_2\delta^{+}_{\lambda}))\right), && \text{otherwise},
\end{aligned}\right.
\end{equation*}
and
\begin{equation*}
\kappa_{-}(\lambda) = \left\{ \begin{aligned}
&0, && \text{if }g_{-} \text{ is bounded at} -\infty,\\
&\mathcal{O}\left(\exp(-d_1F(- d_2\delta^{-}_{\lambda}))\right), && \text{otherwise}.
\end{aligned}\right.
\end{equation*}
Furthermore, if $V_{12}=V_{21}$, 
the statement is also true for $n=0$, \emph{i.e.} $N=0$ is allowed.
\end{proposition}
\begin{proof}
First of all, we want to show that the denominator in \eqref{Eq L^2 xi''} is bounded from below by a constant not depending on $\lambda$. Thanks to  Lemma \ref{Lemma Decay solution} and the boundedness of $V_{ij}$ for $i,j\in\{1,2\}$ on $[0,1]$, one has
\begin{align*}
\int_{\R} \left\vert k_{1} \exp(-P_{\lambda,n}) \right\vert^2 \, \dd x\gtrsim \int_{0}^{1} \exp\left(-2c_{1} F(x)+\int_{0}^{x} \mathcal{U}(t) \, \dd t \right)\, \dd x\gtrsim 1.
\end{align*}
Now, we try to control two terms attached with $\xi^{(1)}$ and $\xi^{(2)}$ in the numerator of \eqref{Eq L^2 xi''}. 
Obviously, the case of $g_{+}$ being bounded at $+\infty$ is trivial, since $\xi(x)=1$ on $[0,+\infty)$. The negative case is the same. We just need to care about the remaining situations in which $\xi$ is a \enquote{true} cut-off function. The main idea is to employ the exponential decay in order to limit the growth of polynomials on the support of $\xi^{(1)}$ and $\xi^{(2)}$. 
In detail, applying Lemma~\ref{Lemma Decay solution}
on the support of $\xi^{(2)}$ and \eqref{Eq cutoff inequa}, \eqref{Eq G Delta}, we obtain (with some $c_{3}, c_{4}>0$)
\begin{align*}
&\quad\left\Vert \frac{k_1}{\lambda+m-V_{22}} \exp(-P_{\lambda,n}) \xi^{(2)} \right\Vert^2\\
& \lesssim \left(\delta_{\lambda}^{-}\right)^{4} \int_{J_{\lambda}^{-}\setminus \widetilde{J}_{\lambda}^{-}}  \exp\left(-2c_{2}F(x) \right) \dd x +  \left(\delta_{\lambda}^{+}\right)^{4}\int_{J_{\lambda}^{+}\setminus \widetilde{J}_{\lambda}^{+}}  \exp\left(-2c_{2}F(x) \right)  \dd x\\
&\lesssim \left(\delta_{\lambda}^{-}\right)^{3} \exp\left(-2c_{2}F\left(-\delta^{-}_{\lambda}+ \Delta_{\lambda}^{-}\right) \right) +\left(\delta_{\lambda}^{+}\right)^{3}\exp\left(-2c_{2} F\left(\delta^{+}_{\lambda}- \Delta_{\lambda}^{+}\right) \right)\\
& \lesssim \exp\left(-c_{3} F\left(-c_{4}\delta_{-}\right) \right) +\exp\left(-c_{3} F\left(c_{4}\delta_{+}\right) \right).
\end{align*}
In the second inequality, we used the fact that 
$F(x)$ increases as $\vert x \vert$ increases. Whereas the observation \eqref{Eq F} is employed in the third inequality. 
The term associated with $\xi^{(1)}$ is estimated in the same way. Just notice that, from the similarity in \eqref{Rem V11 approx V22}, Lemma \ref{Lemma VK}, Lemma \ref{Lemma WKB sol} and estimates $f_{\pm}$ as in \eqref{Eq sum estimate} for all $x\in J_{\lambda}^{\pm}\setminus \widetilde{J}_{\lambda}^{\pm}$, we have (with some $c_5>0$)
\begin{align*}
&\left\vert\left( 2P_{\lambda,n}^{(1)}(x)-\frac{K_{\lambda}^{(1)}(x)}{K_{\lambda}(x)}\right) \frac{k_1(x)\exp(-P_{\lambda,n}(x))}{\lambda+m-V_{22}(x)}\right\vert\\
&\leq \left( \vert V_{\lambda}(x) \vert^{1/2}+ \frac{\vert K_{\lambda}^{(1)}(x)\vert}{\vert K_{\lambda}(x)\vert}+\sum_{k=0}^{n-1}\vert \lambda^{-k}\psi_{k}^{(1)}\vert \right)\frac{\exp\left( -c_{2}F(x) \right) }{ \vert \lambda+m-V_{22}\vert} \\
&\lesssim \left( 1+  \sum_{k=0}^{n-1} \lambda^{-(k+1)} f_{\pm}(x)^{k+1}\right)\exp\left( -c_{2}F(x) \right) \\
&\lesssim \exp\left( -c_{5}F(x) \right).
\end{align*}
Thus, the desired claim follows.
\end{proof}

\subsection{Main results}
Now, we can state our main theorems and their consequences.
 
The following theorem says that if $V_{11}, V_{22}$ at least belong to $W^{2,\infty}_{\mathrm{loc}}(\R)$ and $V_{12}, V_{21}$ at least belong to $W^{1,\infty}_{\mathrm{loc}}(\R)$ and satisfy Assumption \ref{Assump General}, our WKB solution will become the quasimode for the problem \eqref{Mo Eq Pseudo Eigen 2}. When the potential $V$ is symmetric, the sufficient conditions for the involving spaces of $V_{11},V_{22}$ are released to $W^{1,\infty}_{\mathrm{loc}}(\R)$ and of $V_{12}, V_{21}$ are released to $L^{\infty}_{\mathrm{loc}}(\R)$ in some cases. Furthermore, the rate of decay of the estimate when $V$ is symmetric is better in some situations.
\begin{theorem}\label{Theorem 1}
Let Assumption \ref{Assump General} hold for some $N\in \N_{1}$ and set $n=N$. Let $\{\psi_{k}^{(1)}\}_{k\in[[-1,n-1]]}$ be determined by \eqref{Eq Recursion}  with the plus sign in the formula of $\psi_{-1}^{(1)}$ and $P_{\lambda,n}$ defined as in \eqref{Eq Pn}. Let us define
\[
\Psi_{\lambda,n} := \begin{pmatrix}
k_1 u_{\lambda,n}\\
k_2 v_{\lambda,n}
\end{pmatrix},
\qquad \mbox{where} \qquad
u_{\lambda,n} := \xi \exp(-P_{\lambda,n}),\quad v_{\lambda,n}:= \frac{k_1}{k_2} \frac{\partial_{x} u_{\lambda,n}}{\lambda+m-V_{22}},
\]
$k_1,k_2$ are functions as in \eqref{Eq k1k2}
and $\xi$ is given in \eqref{Eq cutoff} whose $\delta^{\pm}_{\lambda}$, $\Delta^{\pm}_{\lambda}$ are identified by \eqref{Eq delta Def} and \eqref{Eq G Delta}. Then, 
\begin{equation*}
\frac{\Vert (H_{V}- \lambda)\Psi_{\lambda,n} \Vert}{\Vert \Psi_{\lambda,n} \Vert} = o(1) , \qquad \lambda \to +\infty.
\end{equation*}
More precisely,
\begin{equation*}
\frac{\Vert (H_{V}- \lambda)\Psi_{\lambda,n} \Vert}{\Vert \Psi_{\lambda,n} \Vert} \leq \kappa(\lambda)+\sigma^{(n)},
\end{equation*}
where $\kappa$ is as in \eqref{Eq L^2 xi''} and $\sigma^{(n)}=\sigma_{-}^{(n)}+\sigma_{+}^{(n)}$ with
\[\sigma_{\pm}^{(n)}= \left\{ \begin{aligned}
& \mathcal{O}(\lambda^{-n}),\quad && f_{\pm} \text{ is bounded at }\pm \infty,\\
& \mathcal{O}(\lambda^{-\frac{(1+\varepsilon_{1})}{2}(n+1)+1}),\quad &&f_{\pm} \text{ is unbounded at }\pm \infty.
\end{aligned}\right.\]
\end{theorem}
\begin{proof}
Before going to the proof, it is necessary to check that~$\Psi_{\lambda,n}$ belongs to the domain of~$H_{V}$. With the choice of $\Psi_{\lambda,n}$ in the statement of the theorem, we have the relation between $H_{V}$ and $\mathscr{L}_{\lambda,V}$ as follows:
$ (H_{V}-\lambda)\Psi_{\lambda,n} = \begin{psmallmatrix}
\mathscr{L}_{\lambda,V} u_{\lambda,n} \\0
\end{psmallmatrix}. 
$
Thus, $\Psi_{\lambda,n} \in \mathcal{D}(H_{V})$ if and only if
\begin{equation*}
\begin{aligned}
k_1 u_{\lambda,n} \in L^2(\R), \qquad
\frac{k_1 \left(-i\partial_{x}\right) u_{\lambda,n}}{\lambda+m-V_{22}} \in L^2(\R), \qquad
\mathscr{L}_{\lambda,V} u_{\lambda,n} \in L^2(\R).
\end{aligned} 
\end{equation*}
Obviously, this happens if $\xi$ is a \enquote{true} cut-off. The thing that makes us worry, for example, is the case $\delta^{+}_{\lambda}=+\infty$ when $g_{+}$ is bounded at $+\infty$.
From the observation \eqref{Eq F} combined with Lemma \ref{Lemma Decay solution} and mimicking the proof of Proposition \ref{Lemma L2 norm}, it yields that, for  sufficiently large $x>0$ (with some $c_3>0$),
$$
  \left\vert k_1(x) u_{\lambda,n}(x) \right\vert\lesssim \exp(-c_3 x),
  \qquad
  \left\vert \frac{k_1(x) (-i\partial_{x})u_{\lambda,n}(x)}{\lambda+m-V_{22}(x)} \right\vert \lesssim  \exp(-c_3 x).
$$
From \eqref{Eq L u}, we obtain
\[ \left\vert \mathscr{L}_{\lambda,V} u_{\lambda,n} \right\vert = \left\vert\frac{R_{\lambda,n}}{\lambda+m-V_{22}} \right\vert \vert k_1 \exp(-P_{\lambda,n}) \vert.\]
We will see later that $\left\vert\frac{R_{\lambda,n}(x)}{\lambda+m-V_{22}(x)} \right\vert$ is uniformly bounded from above. Therefore, the $L^2$-integrability of $\mathscr{L}_{\lambda,V} u_{\lambda,n}$ is ensured by the $L^2$-integrability of $k_1(x) \exp\left(-P_{\lambda,n}(x)\right)$.

Now, we can come back to prove the statement of the theorem. Let us recall that we need to estimate the quantity
\begin{align*}
\frac{\Vert (H_{V}- \lambda)\Psi_{\lambda,n} \Vert}{\Vert \Psi_{\lambda,n} \Vert} &= \frac{\Vert \mathscr{L}_{\lambda,V} u_{\lambda,n}\Vert}{\sqrt{\Vert k_1 u_{\lambda,n} \Vert^2+ \left\Vert \frac{k_1}{\lambda+m-V_{22}} (-i\partial_{x})u_{\lambda,n} \right\Vert^{2}}}.
\end{align*}
Using the triangle inequality and \eqref{Eq L^2 xi''}, we obtain
\begin{align*}
 \Vert \mathscr{L}_{\lambda,V} u_{\lambda} \Vert
 \leq \ &\left\Vert \frac{k_1}{\lambda+m-V_{22}} \exp(-P_{\lambda,n}) \xi^{(2)} \right\Vert + \left\Vert \frac{k_1}{\lambda+m-V_{22}} \left( 2P_{n}^{(1)}-\frac{K_{\lambda}^{(1)}}{K_{\lambda}}\right)\exp(-P_{\lambda,n}) \xi^{(1)} \right\Vert 
 \\
&+ \left\Vert \frac{R_{\lambda,n}}{\lambda+m-V_{22}} \right\Vert_{L^{\infty}(J_{\lambda})}\left\Vert k_1 u_{\lambda,n} \right\Vert,
\end{align*}
and thus
\begin{align*}
\frac{\Vert \mathscr{L}_{\lambda,V} u_{\lambda,n}\Vert}{\sqrt{\Vert k_1 u_{\lambda,n} \Vert^2+ \left\Vert \frac{k_1}{\lambda+m-V_{22}} (-i\partial_{x})u_{\lambda,n} \right\Vert^{2}}} \leq \kappa(\lambda) + \left\Vert \frac{R_{\lambda,n}}{\lambda+m-V_{22}} \right\Vert_{L^{\infty}(J_{\lambda})}.
\end{align*}
The estimate of the remainder \eqref{Eq Estimate Rn} together 
with Lemma~\ref{Lemma VK} yield that, for all $ x\in [a_{-},a_{+}]$,
\begin{align*}
\left\vert\frac{R_{\lambda,n}(x)}{\lambda+m-V_{22}(x)} \right\vert \lesssim \lambda^{-n}.
\end{align*}

In a similar way, it turns out that, for all $x\in I^{\pm}\cap J_{\lambda}^{\pm}$,
\begin{equation*}
\left\vert\frac{R_{\lambda,n}(x)}{\lambda+m-V_{22}(x)} \right\vert \lesssim \frac{1}{\lambda}\sum_{\ell=-1}^{n-2} \frac{ f_{\pm}(x)^{n+\ell+2}}{\vert V_{\lambda}(x)\vert^{(n+\ell)/2}}\lesssim \left\{ \begin{aligned}
& \lambda^{-n}, \qquad &&f_{\pm} \text{ is bounded at }\pm \infty,\\
& \lambda^{-\frac{(1+\varepsilon_{1})}{2}(n+1)+1}, \qquad &&f_{\pm} \text{ is unbounded at }\pm \infty.\end{aligned} \right.
\end{equation*}
To prove the case $f_{\pm}$ is unbounded at $\pm \infty$, we employ the fact $f_{\pm}(x)^{\frac{2}{1-\varepsilon_{1}}}\leq \lambda$, which is a consequence of \eqref{Eq G Property delta 2} and the definition of $g$ in \eqref{Eq G g 1}.
\end{proof}

\begin{theorem}\label{Theorem 2}
Under the same assumptions 
and settings as in Theorem~\ref{Theorem 1} with $N\in \N_{0}$ and $V_{12}=V_{21}$, let $n=N$. Then, 
\begin{equation*}
\frac{\Vert (H_{V}- \lambda)\Psi_{\lambda,n} \Vert}{\Vert \Psi_{\lambda,n} \Vert} \leq \kappa(\lambda)+\tau^{(n)}(\lambda) , \qquad \lambda \to +\infty.
\end{equation*}
Here, $\kappa$ is as in \eqref{Eq L^2 xi''} and $\tau^{(n)}=\tau_{-}^{(n)}+\tau_{0}^{(n)}+\tau_{+}^{(n)}$ with $\tau_{0}=\lambda^{-(n+1)}$ and
\begin{enumerate}[i)]
\item if $V_{11}$ or $V_{22}$  is not bounded at $\pm \infty$
\begin{equation*}
\tau_{\pm}^{(n)}(\lambda)=\left\{
\begin{aligned}
&\mathcal{O}\left(\lambda^{-n}\sup_{x\in I^{\pm}\cap J_{\lambda}^{\pm}} \frac{f_{\pm}(x)^{n+1} \max\{\vert V_{11} \vert, \vert V_{22} \vert \}(x)}{\vert \lambda+m-V_{22}(x) \vert} \right),\quad &&f_{\pm} \text{ is bounded at }\pm \infty,\\
&\mathcal{O}\left(\lambda^{-\frac{(1+\varepsilon_{1})}{2}(n+1)+1}\right),\quad &&f_{\pm} \text{ is unbounded at }\pm \infty;
\end{aligned}
\right.
\end{equation*}
\item if $V_{11}$ and $V_{22}$ are bounded at $\pm \infty$
\begin{equation*}
\tau_{\pm}^{(n)}(\lambda)=\left\{ \begin{aligned}
& \mathcal{O}(\lambda^{-(n+1)}),\quad &&f_{\pm} \text{ is bounded at }\pm \infty,\\
& \mathcal{O}(\lambda^{-\frac{(1+\varepsilon_{1})}{2}(n+1)}),\quad &&f_{\pm} \text{ is unbounded at }\pm \infty.
\end{aligned}\right.
\end{equation*}
\end{enumerate}
\end{theorem}
\begin{proof}
Since $V_{12}=V_{21}$, Lemma \ref{Lemma Decay solution} and Proposition \ref{Lemma L2 norm} still hold when $n=0$. However, we can perform the estimates in Theorem \ref{Theorem 1} more strictly. Now, we assume that $V_{11}$ or $V_{22}$ is not bounded at $+\infty$. In detail, from the estimate of the remainder in \eqref{Eq Estimate Rn} together with Lemma \ref{Lemma VK} and for $n=0$, we obtain
\begin{align*}
\left\vert\frac{R_{\lambda,0}(x)}{\lambda+m-V_{22}(x)} \right\vert&\lesssim\frac{1}{\lambda}\left(\frac{\vert V_{\lambda}^{(1)}(x) \vert}{\vert V_{\lambda}(x) \vert^{1/2}} + \frac{\vert K_{\lambda}^{(1)}(x) \vert}{\vert K_{\lambda}(x) \vert} \vert V_{\lambda}(x) \vert^{1/2}\right)\\
&\lesssim \left\{ \begin{aligned}
&\frac{1}{\lambda} \qquad &&\forall x \in [0,a_{+}],\\
& \frac{f_{+}(x)\max\{\vert V_{11} \vert, \vert V_{22} \vert \}(x)}{\vert \lambda + m - V_{22}(x)\vert} \qquad &&\forall x \in I^{+}\cap J_{\lambda}^{+}.
\end{aligned}
\right.
\end{align*}
For $n\geq 1$, for all $x\in [0,a_{+}]$, we have
\begin{align*}
\left\vert\frac{R_{\lambda,n}(x)}{\lambda+m-V_{22}(x)} \right\vert \lesssim \frac{1}{\lambda^{n+1}}.
\end{align*}
While for all $x\in I^{+}\cap J_{\lambda}^{+}$,
\begin{align*}
\left\vert\frac{R_{\lambda,n}(x)}{\lambda+m-V_{22}(x)} \right\vert &\lesssim \sum_{\ell=-1}^{n-2} \frac{1}{\lambda^{n+\ell+1}} \sum_{\vert j \vert=1}^{n+\ell+2} \sum_{\vert r \vert=n+\ell+2}  f_{+}(x)^{\vert r \vert} \left( \frac{\max\{\vert V_{11} \vert, \vert V_{22} \vert \}(x)}{\vert \lambda+m-V_{22}(x)\vert}\right)^{\vert j \vert}\\
&\lesssim \frac{1}{\lambda^{n}}\frac{ f_{+}(x)^{(n+1)} \max\{\vert V_{11} \vert, \vert V_{22} \vert \}(x)}{\vert \lambda+m-V_{22}(x)\vert}\sum_{\ell=-1}^{n-2} \frac{ f_{+}(x)^{\ell+1}}{\lambda^{\ell+1}}.
\end{align*}
In the second inequality, we employed \eqref{Eq max}. Notice that, if $f_{+}$ is unbounded, from \eqref{Eq G Property delta 2} and the definition of $g$ in \eqref{Eq G g 1}, we have $ f_{+}(x) \leq \lambda^{\frac{1-\varepsilon_{1}}{2}}$ for all $x\in J_{\lambda}^{+}$ and thus for all $x\in I^{+}\cap J_{\lambda}^{+}$,
\[\sum_{\ell=-1}^{n-2} \frac{f_{+}(x)^{\ell+1}}{\lambda^{\ell+1}} \lesssim 1,\]
in all cases of $f_{+}$. From this, we obtain the estimates in the statement for all $x\geq 0$, even in the case $V_{11}$ and $V_{22}$ bounded at $\pm \infty$. The proof for $x\leq 0$ is fulfilled in the same way.
\end{proof}
\begin{remark}\label{Remark sign}
Let us make some comments about the shape of the pseudomodes in connection with the sign of $\mathrm{Im}\, V_{11}+\mathrm{Im}\, V_{22}$ and the sign of $\lambda$:
\begin{enumerate}[i)]
\item If $\lambda>0$ and the sum of the diagonal terms of $V$ changes its sign in the assumption~$\eqref{Assump G Diagonal 1}$, \emph{i.e}.
\begin{equation}\label{Assump G Diagonal 1'}
\begin{aligned}
&\mathrm{Im}\, V_{11}+\mathrm{Im}\, V_{22}\gtrsim 1 \qquad &&\text{ on }I^{-},\\
&\mathrm{Im}\, V_{11}+\mathrm{Im}\, V_{22} \lesssim -1 \qquad &&\text{ on }I^{+};
\end{aligned}
\end{equation}
then we just need to choose the minus sign in the formula of $\psi_{-1}^{(1)}$ in \eqref{Eq Recursion}. Then, we have
\begin{equation*}
\mathrm{Re}\, (\lambda \psi_{-1}^{(1)})=-\frac{1}{\sqrt{2}} \frac{\mathrm{Im} V_{11}(\lambda+m-\mathrm{Re}\, V_{22})+ \mathrm{Im}\, V_{22} (\lambda-m-\mathrm{Re}\, V_{11})}{\sqrt{\vert V_{\lambda} \vert + \mathrm{Re}\, V_{\lambda}}}.
\end{equation*}
By repeating the procedure when proving \eqref{Eq Re psi-1'}, we have
\begin{equation}\label{Eq lambda Re psi-1'}
\begin{aligned}
&\mathrm{Re} \,(\lambda\psi_{-1}^{(1)})\geq  -\frac{\mu_{+}-\eta}{\sqrt{\eta^2+(2+2\eta)^2}} \left(\mathrm{Im}\, V_{11} +  \mathrm{Im}\, V_{22} \right)\qquad && \text{ on } I^{+}\cap J_{\lambda}^{+},\\
&\mathrm{Re} \,(\lambda\psi_{-1}^{(1)})\leq -\frac{\mu_{-}-\eta}{\sqrt{\eta^2+(2+2\eta)^2}} \left(\mathrm{Im}\, V_{11} +  \mathrm{Im}\, V_{22} \right)\qquad &&\text{ on } I^{-}\cap J_{\lambda}^{-}.
\end{aligned}
\end{equation}
Therefore, the function
\[ \tilde{F}(x) := -\int_{0}^{x} \left(\mathrm{Im}\, V_{11}(t) +  \mathrm{Im}\, V_{22}(t)\right)\, \dd t \]
will play the same role as function $F$. 
Although all the other terms $\{\psi_{k}^{(1)}\}_{1\leq k \leq n-1}$ also change their sign, it does not matter because they are all estimated with absolute value. Only the sign of $\lambda \psi_{-1}^{(1)}$ is crucial. Thus, we still assume the same remaining hypotheses 
in Assumption~\ref{Assump General}, but $F$ is replaced by $\tilde{F}$ and we have the same outcomes as stating in the above theorems.
\item Let $\lambda<0$ and Assumption \ref{Assump General} hold. What we need to do is to slightly change $\lambda$ into $-\lambda$ in some places such as in Lemma \ref{Lemma delta}. We redefine
$
\delta^{\pm}_{\lambda} : = \inf\left\{ x\geq a_{\pm}: g_{\pm}(x)=-\lambda\right\} 
$
with $g_{\pm}$ being given in the same way as in \eqref{Eq G g 1}. When $(-\lambda)$ is large enough, it follows, as in \eqref{Rem Re V1122}, that for all $x\in J_{\lambda}$:
\begin{equation*}
\begin{aligned}
&(1-\eta)(-\lambda)\leq \mathrm{Re}\, V_{11}(x)+m-\lambda \leq (1+\eta)(-\lambda),\\
&(1-\eta)(-\lambda)\leq  \mathrm{Re}\, V_{22}(x)-m-\lambda\leq (1+\eta)(-\lambda),\\
&\vert \mathrm{Im}\, V_{11}(x)- \mathrm{Im}\, V_{22}(x)\vert \leq \eta(-\lambda).
\end{aligned}
\end{equation*}
In this case, we will choose the minus sign in the formula of $\psi_{-1}^{(1)}$ in \eqref{Eq Recursion} and we obtain the same results as in \eqref{Eq lambda Re psi-1'}. Then, the method of the current section still works for the pseudomode construction and the outcomes of the above theorems are analogous.
\end{enumerate}
In summary, from the two remarks above, 
our scheme suggests that the sign of the solution $\psi_{-1}^{(1)}$ of the eikonal equation should be chosen as in the following table:
\begin{table}[h]
\begin{tabular}{ |c||l||l| } 
\hline
\multirow{2}{1cm}{} &$\mathrm{Im}\, V_{11}+ \mathrm{Im}\, V_{22}\lesssim -1$ on $I_{-}$  &$\mathrm{Im}\, V_{11}+ \mathrm{Im}\, V_{22}\gtrsim 1$ on $I_{-}$\\
&$\mathrm{Im}\, V_{11}+ \mathrm{Im}\, V_{22}\gtrsim 1$ on $I_{+}$&$\mathrm{Im}\, V_{11}+ \mathrm{Im}\, V_{22}\lesssim -1$ on $I_{+}$ \\
\hhline{|=||=||=|}
$\lambda>0$ & $+$ & $-$ \\ 
\hline
$\lambda<0$ & $-$ & $+$ \\ 
\hline
\end{tabular}
\caption{The sign in the formula of $\psi_{-1}^{(1)}$ in \eqref{Eq Recursion}.}
\end{table}

\end{remark} 
\subsection{Applications}
We consider some special examples of the matrix-valued potentials $V$ which satisfy Assumption \ref{Assump General}.
\begin{example}
Let us list some smooth potentials $V$ defined on $\R$ such that $V_{12}=V_{21}=u$ and Assumption~\ref{Assump General} holds true. From that, we can apply Theorem \ref{Theorem 2}.
\begin{enumerate}
\item[1)] $V_{11}$ and $V_{22}$ are bounded at $\pm \infty$:
\begin{equation*}
V(x) :=
\begin{pmatrix}
i \frac{x}{\sqrt{x^2+1}} && u(x)\\
u(x) && 0
\end{pmatrix},
\end{equation*}
where $u$ is some smooth function on $\R$ such that, 
with $\varepsilon \in \left(0,\frac{1}{2}\right)$,
\[ \int_{0}^{x} \mathrm{Im} \, u(t) \, \dd t \leq \varepsilon F(x) =\varepsilon (\sqrt{x^2+1}-1), \qquad \forall \vert x \vert \gtrsim 1.\]
Here we choose $\mu_{\pm}=1$, $f_{\pm}(x)=\vert x\vert^{-1}$ for $\vert x \vert\gtrsim 1$. Since $g_{\pm}$ are bounded both at $-\infty$ and $+\infty$, the cut-off function is not needed for the pseudomodes construction. For all $n\in \N_{0}$, there exists $\lambda_0>0$ such that, for all $\lambda>\lambda_0$, 
\[\frac{\Vert (H_{V}- \lambda)\Psi_{\lambda,n} \Vert}{\Vert \Psi_{\lambda,n} \Vert} \lesssim \lambda^{-(n+1)}. \]
\item[2)] $V_{11}$ is bounded but $V_{22}$ is unbounded at $\pm \infty$:
\begin{equation*}
V(x) :=
\begin{pmatrix}
i \frac{x}{\sqrt{x^2+1}} && u(x)\\
u(x) && i \ln(x+\sqrt{x^2+1})
\end{pmatrix},
\end{equation*}
where $u$ is any smooth function on $\R$ such that, 
with $\varepsilon \in \left(0,\frac{1}{2}\right)$,
\[ \int_{0}^{x} \mathrm{Im} \, u(t) \, \dd t \leq \varepsilon F(x)=\varepsilon x\ln(x+\sqrt{x^2+1}), \qquad \forall \vert x \vert \gtrsim 1.\]
Here we choose $\mu_{\pm}=1$, $f_{\pm}(x)=\vert x\vert^{-1}$ for $\vert x \vert\gtrsim 1$. Following \eqref{Eq G g 1} and \eqref{Eq delta Def}, for $\lambda>0$ large enough, the boundary of the cut-off can be computed approximately as
$ \delta^{-}_{\lambda} =\delta^{+}_{\lambda}\approx \sinh(\eta \lambda) \approx e^{\eta \lambda},$
and thus with some $c>0$,
$\kappa(\lambda) = \mathcal{O}\left(\exp(-ce^{\eta \lambda}) \right).$
It implies that, for all $n\in N_{0}$, there exists $\lambda_{0}>0$ such that, for all $\lambda>\lambda_{0}$, we have
\[ \frac{\Vert (H_{V}- \lambda)\Psi_{\lambda,n} \Vert}{\Vert \Psi_{\lambda,n} \Vert} \lesssim \lambda^{-(n+1)}.\] 
\item[3)] 
$V_{11}$ is bounded at $-\infty$ but unbounded at $+\infty$ while $V_{22}$ is on the contrary:
\[ V(x) := \begin{pmatrix}
i\frac{e^x}{2} && u\\
u && -i \frac{e^{-x}}{2}
\end{pmatrix},\]
where $u$ is any smooth function on $\R$ such that, 
with $\varepsilon \in \left(0,\frac{1}{2}\right)$,
\[ \int_{0}^{x} \mathrm{Im} \, u(t) \, \dd t \leq \varepsilon F(x)=\varepsilon \cosh(x), \qquad \forall \vert x \vert \gtrsim 1.\]
In this situation, we make a choice $\mu_{\pm}=\mu$ with some $\mu\in (2\varepsilon,1)$, $f_{\pm}(x)=1$ for $\vert x\vert\gtrsim 1$. From \eqref{Eq G g 1} and \eqref{Eq delta Def} and for $\lambda>0$ large enough, we obtain
\[ \delta^{-}_{\lambda}=\delta^{+}_{\lambda} = \arcsinh(\eta \lambda) =\ln\left(\eta \lambda+\sqrt{(\eta \lambda)^2+1}\right),\]
and thus with some $c>0$,
$\kappa(\lambda) = \mathcal{O}\left(\exp(-c\lambda^{d_2}) \right).$
It implies that, for all $n\in \N_{1}$, there exists $\lambda_{0}>0$ such that, for all $\lambda>\lambda_{0}$,
\[ \frac{\Vert (H_{V}- \lambda)\Psi_{\lambda,n} \Vert}{\Vert \Psi_{\lambda,n} \Vert} \lesssim \lambda^{-n}.\]
\end{enumerate}

\end{example}
\begin{example}[Polynomial-like diagonal terms] Let us take a look at the potential $V$ satisfying Assumption \ref{Assump General} with $f_{\pm}(x)= \vert x\vert^{-1}$, $V_{12}=V_{21}=0$ (for simplicity) and
\[ \vert \mathrm{Re}\, V_{ii} \vert \approx \vert x \vert^{\alpha_{ii}},\qquad  \vert \mathrm{Im}\, V_{ii}\vert  \approx  \vert x \vert^{\beta_{ii}} ,\qquad \forall \vert x \vert \gtrsim 1,\]
with $\alpha_{ii},\beta_{ii}\in \R$, for $i\in \{1,2\}$. 
It is necessary to assume that $\max\,\{\beta_{11},\beta_{22}\}\geq 0$ such that the sum of the imaginary parts of the diagonal terms $\mathrm{Im}\, V_{11}+ \mathrm{Im}\, V_{22}$ satisfies the condition \eqref{Assump G Diagonal 1}. Theorem \ref{Theorem 1} provides us with the fact that $n=1$ (\emph{i.e.} we need $V_{11},V_{22} \in W_{\mathrm{loc}}^{2,2}(\R)$ and $V_{11},V_{22} \in W_{\mathrm{loc}}^{1,2}(\R)$) is enough to treat all kinds of potentials satisfying Assumption \ref{Assump General}. However, we would like to see what type of potential that $n=0$ can be treated and how fast the decay is when the potential is more regular. For that purposes, let us consider two cases in Theorem \ref{Theorem 2}:
\begin{enumerate}
\item[Case 1:] $\vert V_{11} \vert$ and $\vert V_{22}\vert$ are bounded at $\pm \infty$. This happens if and only if $\displaystyle \omega:=\max_{i\in \{1,2\}} \{\alpha_{ii}, \beta_{ii}\}=0$. The application of Theorem \ref{Theorem 2} yields that, for all $n\geq 0$,
\[ \frac{\Vert (H_{V}- \lambda)\Psi_{\lambda,n} \Vert}{\Vert \Psi_{\lambda,n} \Vert} = \mathcal{O}(\lambda^{-(n+1)}) \qquad \text{ as } \lambda\to+\infty.\]
Furthermore, in this case, the pseudomodes globally localise on $\R$ without being attached with cut-off functions.
\item[Case 2:] $\vert V_{11} \vert$ or $\vert V_{22}\vert$ is not bounded at $\pm \infty$ (\emph{i.e.} $\omega>0$). 
We consider two smaller cases:
\begin{enumerate}
\item[i)] $\mathrm{Re}\, V_{11}$ and $\mathrm{Re}\, V_{22}$ and $\mathrm{Im}\, V_{11} - \mathrm{Im}\, V_{22}$ are bounded at $\pm \infty$.  We do not use cut-off in this situation and for all $n\geq 0$,
\begin{align*}
\frac{\Vert (H_{V}- \lambda)\Psi_{\lambda,n} \Vert}{\Vert \Psi_{\lambda,n} \Vert} &= \left\{
\begin{aligned}
&\mathcal{O}\left(\lambda^{-(n+1)}\right), \qquad & \omega \leq n+1,\\
&\mathcal{O}\left( \lambda^{-n}\right),\qquad &\omega>n+1,
\end{aligned}
\right.
\end{align*}
as $\lambda \to +\infty$.
In the case $\omega>n+1$, we employed \eqref{Eq max}.
\item[ii)] $\mathrm{Re}\, V_{11}$ or $\mathrm{Re}\, V_{22}$ or $\mathrm{Im}\,V_{11}-\mathrm{Im}\,V_{22}$ is unbounded at $\pm\infty$. Then, the possible maximum order of $g_{\pm}$ denoted by $\tilde{\omega}$ is $\omega$, \emph{i.e.} $\tilde{\omega} \leq \omega$. Of course, $\tilde{\omega}>0$ and thus we can compute 
$\delta^{-}_{\lambda}=\delta^{+}_{\lambda}=\delta \approx \lambda^{\frac{1}{\tilde{\omega}}}.$
Applying Theorem \ref{Theorem 2} again, it results that, for all $n\geq 0$,
\begin{align*}
\frac{\Vert (H_{V}- \lambda)\Psi_{\lambda,n} \Vert}{\Vert \Psi_{\lambda,n} \Vert} &= \left\{
\begin{aligned}
&\mathcal{O}\left(\lambda^{-(n+1)}\right), \qquad & \omega \leq n+1,\\
&\mathcal{O}\left( \lambda^{-(n+1)+\frac{\omega-n-1}{\tilde{\omega}}}\right),\qquad &\omega>n+1,
\end{aligned}
\right.
\end{align*}
as $\lambda \to +\infty$.
We see that in the second case, 
$n=0$ can cover all the potentials such that
$ \omega-1<\tilde{\omega} \leq \omega.$
For example, $n=0$ can treat
\begin{enumerate}[a)]
\item $V= \begin{pmatrix}
x^2+ix && 0\\
0 && 0
\end{pmatrix}$, with $\omega=\tilde{\omega}=2$.
\item $V= \begin{pmatrix}
i(x+\vert x \vert^{\frac{1}{2}}) && 0\\
0 && ix
\end{pmatrix}$, with $\omega=1$ and $\tilde{\omega}=\frac{1}{2}$.
\end{enumerate}
\end{enumerate}
\end{enumerate}
The same thing happens as in the Schr\"{o}dinger case: 
the pseudomode with $n=1$ is sufficient to treat all polynomial-like potential (even the case $V_{12} \neq V_{21}$, see Theorem \ref{Theorem 1}). The pseudomode associated with $n=0$ suffices for potentials growing not faster than linearly. 
\end{example}

\begin{example}[Exponential potentials] 
Consider following potentials $V$ satisfying Assumption \ref{Assump General}. Since we would like to apply Theorem \ref{Theorem 2}, we will assume further that $V_{12}=V_{21}=0$ for the sake of simplicity. 
\begin{enumerate}
\item[1)] $V_{11}$ and $V_{22}$ are bounded at $\pm \infty$ with $N\geq0$:
\begin{equation*}
V(x):=\begin{pmatrix}
i\, \mathrm{sgn}(x) e^{\vert x \vert^{-\alpha_{1}}} && 0\\
0 && i\,\mathrm{sgn}(x) e^{\vert x \vert^{-\alpha_{2}}}
\end{pmatrix},
\qquad \forall  \vert x \vert \gtrsim 1,
\end{equation*}
with $\alpha_{1},\alpha_{2}\geq 0$. We choose $f_{\pm}(x)=1$ for $\vert x \vert\gtrsim 1$. The cut-off function is not needed in this case, \emph{i.e.,} $\xi\equiv 1$. From Proposition \ref{Lemma L2 norm} and Theorem \ref{Theorem 2},  we have, 
\begin{equation*}
\frac{\Vert (H_{V}- \lambda)\Psi_{\lambda,N} \Vert}{\Vert \Psi_{\lambda,N} \Vert} = \mathcal{O}\left(\lambda^{-(N+1)}\right) \qquad \text{ as } \lambda\to+\infty.
\end{equation*}
\item[2)] $V_{11}$ is unbounded at $\pm\infty$ while $V_{22}$ is bounded at $\pm\infty$ with $N\geq 1$, moreover, they oscillate on $\R$:
\begin{equation*}
V(x):=\begin{pmatrix}
\vert x \vert^{\alpha}+i\, \mathrm{sgn}(x) e^{\sin(x)} && 0\\
0 && i\, \mathrm{sgn}(x)  e^{\cos(x)}
\end{pmatrix},
\qquad \forall \vert x \vert \gtrsim 1,
\end{equation*}
with $\alpha >0$. We choose $f_{\pm}(x)=1$ for $\vert x \vert\gtrsim 1$. From \eqref{Eq delta Def} and \eqref{Eq G g 1}, we have $\delta^{-}_{\lambda}=\delta^{+}_{\lambda}=\delta\approx \lambda^{\frac{1}{\alpha}}$. Using \eqref{Eq F}  we obtain $\kappa \lesssim\exp\left(-c \lambda^{\frac{1}{\alpha}}\right)$. Theorem \ref{Theorem 2} gives us
\[\frac{\Vert (H_{V}- \lambda)\Psi_{\lambda,N} \Vert}{\Vert \Psi_{\lambda,N} \Vert} = \mathcal{O}\left(\lambda^{-N}\right)\qquad \text{ as } \lambda\to+\infty. \]
\item[3)] 
$V_{11}$ and $V_{22}$ are unbounded at $\pm\infty$ with $N\geq 1$:
\begin{equation*}
V(x):=\begin{pmatrix}
i\,\textup{sgn}(x)\, e^{\vert x \vert^{\alpha_{1}}} && 0\\
0 && i\,\textup{sgn}(x)\, e^{\vert x \vert^{\alpha_{2}}}
\end{pmatrix},
\qquad \forall \vert x \vert \gtrsim 1,
\end{equation*}
with $\alpha_{1}, \alpha_{2}>0$. We choose $f_{\pm}(x)= \vert x \vert^{\omega-1}$ where $\omega:=\max \{\alpha_{1}, \alpha_{2}\}$. 
\begin{enumerate}[i)]
\item 
If $\vert \mathrm{Im}\, V_{11}-\mathrm{Im}\, V_{22} \vert$ is bounded at $\pm \infty$, \emph{i.e.} $\alpha_{1}=\alpha_{2}=\omega$, from \eqref{Eq delta Def} and \eqref{Eq G g 1}, we consider two situations:
\begin{enumerate}[a)]
\item if $0<\omega\leq 1$, $g_{\pm}$ are bounded at $\pm \infty$. From Proposition \ref{Lemma L2 norm}, $\kappa(\lambda)=0$,
\item if $\omega>1$, $\delta^{\pm}_{\lambda} = \lambda^{\frac{1-\varepsilon_{1}}{2(\omega-1)}}$. Thanks to \eqref{Eq F}, $\kappa(\lambda)\lesssim \exp\left(-c\lambda^{\frac{1-\varepsilon_{1}}{2(\omega-1)}} \right)$ with some $c>0$.
\end{enumerate}
\item 
If $\vert \mathrm{Im}\, V_{11}-\mathrm{Im}\, V_{22} \vert$ is unbounded at $\pm \infty$, by \eqref{Eq delta Def} and \eqref{Eq G g 1} and for sufficiently large $\lambda>0$, $ \delta^{\pm}_{\lambda} \approx (\ln (\lambda))^{\frac{1}{\omega}}$. From the definition of $F$, we can obtain the estimate $F(x) \gtrsim \vert x\vert^{\omega}$ for $\vert x \vert \gtrsim 1$. 
Thus, by Proposition~\ref{Lemma L2 norm}, there exists a constant $c>0$ such that $\kappa(\lambda) \lesssim \lambda^{-c \ln \lambda}$.
\end{enumerate}
Finally, from Theorem \ref{Theorem 2}, we have
\[\frac{\Vert (H_{V}- \lambda)\Psi_{\lambda,N} \Vert}{\Vert \Psi_{\lambda,N} \Vert} = \left\{\begin{aligned}
&\mathcal{O}\left(\lambda^{-N}\right), \qquad && \omega \leq 1,\\
&\mathcal{O}\left(\lambda^{-\frac{1+\varepsilon_{1}}{2}(N+1)+1}\right), \qquad && \omega> 1,
 \end{aligned}\right.\qquad \text{ as } \lambda\to+\infty.\]

\item[4)] Superexponential functions:
\begin{equation*}
V(x):=\begin{pmatrix}
i e^{\sinh(x)} && 0\\
0 && -ie^{-\sinh(x)}
\end{pmatrix},
\qquad \forall x\in \R.
\end{equation*}
In this example, we choose $f_{\pm}(x)=\cosh(x)$ for $\vert x\vert\gtrsim 1$.  From \eqref{Eq delta Def} and \eqref{Eq G g 1}, we can compute, with some constant $C>0$, that
\[ \delta^{-}_{\lambda} =\delta^{+}_{\lambda} = \arcsinh\left( \arccosh\left(\frac{\eta \lambda}{2}\right)\right)\geq C \ln\left(\ln(\lambda)\right),\]
when $\lambda>0$ large enough. Clearly, we have the following rough estimate
\[ F(x) = \int_{0}^x 2\sinh(\sinh(t))\, \dd t \gtrsim  e^{\frac{2}{d_{2}C}\vert x \vert},\qquad \forall \vert x \vert \gtrsim 1,\]
where $d_{2}>0$ is the number
appearing in Proposition~\ref{Lemma L2 norm}.
Thus, we obtain
$ F(\pm d_{2} \delta^{\pm}_{\lambda} ) \gtrsim \ln(\lambda)^2$
and there exists a constant $c>0$ such that
$\kappa(\lambda) =\mathcal{O}\left(\lambda^{-c \ln \lambda}\right).$
From Theorem \ref{Theorem 1}, we have
\[ \frac{\Vert (H_{V}- \lambda)\Psi_{\lambda,N} \Vert}{\Vert \Psi_{\lambda,N} \Vert} =\mathcal{O}\left(\lambda^{-\frac{1+\varepsilon_{1}}{2}(N+1)+1}\right)\qquad \text{ as } \lambda\to+\infty.\]

\end{enumerate}
\end{example}

\section{Pseudomodes for large general pseudoeigenvalues}\label{Section complex}
In this section, we want to construct the pseudomode corresponding 
to a  complex pseudoeigenvalue
\[ 
  \lambda = \alpha + i \beta,
  \qquad \mbox{where} \qquad
  \alpha,\beta \in \R.
\]
The interesting part is that the shape of the pseudospectral region is also revealed in the process of the construction. If the large real part of $\lambda$ played a decisive role in the decaying estimation in the previous section, the imaginary part $\beta$ will take on this role in this section. We shall pay attention to the class of potentials whose $\mathrm{Im}\, V_{11}$ and $\mathrm{Im}\, V_{22}$ are identical, \emph{i.e.} $\mathrm{Im}\, V_{11}=\mathrm{Im}\, V_{22}=:\mathcal{V}$, and increasing on $\R_{+}$. The WKB analysis will be performed around 
a \emph{turning} point $x_{\beta}>0$ which is defined by the equation
\begin{equation}\label{Eq Turning point}
\mathcal{V}(x_{\beta})=\beta.
\end{equation}

Since the non-zero part of the pseudomode will live completely in $\R_{+}$, it will be more convenient to consider the operators on $L^{2}(\R_{+}) \oplus L^2(\R_{+})$, instead of $ L^2(\R) \oplus L^2(\R)$. The application of the results for the class of operators on $ L^2(\R) \oplus L^2(\R)$ is easily obtained by the trivial extension of the pseudomode of the operators on $L^{2}(\R_{+}) \oplus L^2(\R_{+})$.

\subsection{Allowable shapes of the potentials}
Since we do not have to bound the derivatives of components of $V$ on a fixed compact set as in Section \ref{Section Real}, the whole space of them can be enlarged to $L^{2}_{\mathrm{loc}}(\R_{+})$, instead of $L^{\infty}_{\mathrm{loc}}(\R_{+})$. In order that the turning point $x_{\beta}$ is uniquely determined, we will assume that $\mathcal{V}$ is strictly monotone for sufficiently large $x>0$.  The assumptions \eqref{Assump G Der Vii} and \eqref{Assump G Der Vij} are kept the same in this section, so that we can control the transport solutions.  
To be more specific, we make the following hypothesis.

\begin{assumpA}\label{Assump General C}
Let $N\in \N_{1}$ and $\nu\geq -1$, assume that $V_{11},V_{22} \in W^{N+1,2}_{\mathrm{loc}}\left(\R_{+}\right)$ and $V_{12},V_{21} \in W^{N,2}_{\mathrm{loc}}\left(\R_{+}\right)$ satisfy the subsequent conditions:
\begin{enumerate}[1)]
\item the imaginary parts of $V_{11}$ and $V_{22}$ are equal, \emph{i.e.} $\mathrm{Im}\, V_{11} = \mathrm{Im}\, V_{22} =: \mathcal{V}$ satisfy
\begin{equation}\label{Assump G V Rising}
\lim_{x\to+\infty} \mathcal{V}(x)=+\infty,
\end{equation}
and there exist $\varepsilon_{1}>0$ such that, for all $x \gtrsim 1$,
\begin{numcases}{\mathcal{V}^{(1)}(x)\gtrsim}
\mathcal{V}(x)^{\frac{1}{2}} x^{\frac{3}{2}\nu+\varepsilon_{1}}, &\label{Assump G Diagonal 1 C}\\
\vert \mathcal{V}^{(2)}(x) \vert x^{-\nu}, &\label{Assump G Diagonal 2 C}
\end{numcases}

\item the sum $\mathcal{U}:=\mathrm{Im}\, V_{12}+ \mathrm{Im}\, V_{21}$ is controlled above by $\mathcal{V}^{(1)}$: 
\begin{equation}\label{Assump G Off-Diagonal C}
\vert \mathcal{U}(x) \vert = o(x^{-\nu}\mathcal{V}^{(1)}(x))  \qquad \text{ as } x \to +\infty;
\end{equation}
\item the derivatives of $V_{ii}$ are controlled by $V_{ii}$, for $i\in \{1,2\}$,
\begin{equation}\label{Assump G Der Vii C}
 \forall n \in [[1,N+1]],\qquad \vert V_{ii}^{(n)}(x) \vert =\mathcal{O}\left( x^{n\nu} \left\vert  V_{ii}(x)\right\vert \right) \qquad \text{ as } x \to +\infty,
\end{equation}
and the differences between $V_{12}$ and $V_{21}$ (their derivatives) are controlled by polynomials 
\begin{equation}\label{Assump G Der Vij C}
 \forall n \in [[0,N]],\qquad \vert (V_{21}-V_{12})^{(n)}(x) \vert =\mathcal{O}(  x^{(n+1)\nu}) \qquad \text{ as } x \to + \infty.
\end{equation}
\end{enumerate}
\end{assumpA}

Comparing with Assumption~\ref{Assump General}, although there are more conditions for the imaginary parts~$\mathcal{V}$ of the diagonal terms, the class of admissible potentials is still very large. Furthermore, our assumption allows to cover functions~$\mathcal{V}$ which grow slowly at $+\infty$ such as logarithmic ones. This is interesting because the analogous hypothesis
\cite[Ass.~5.2]{KS19} for Schr\"{o}dinger operators
does not allow for this kind of functions.
\begin{remark}
We have the following helpful properties
\begin{enumerate}[i)]
\item As discussed in \cite[Sec.~3]{KS19}, when $\nu<-1$ the condition \eqref{Assump G Der Vii C} immediately implies that $V_{11}$ and $V_{22}$ are bounded. Thus the rising of $\mathcal{V}$ in \eqref{Assump G V Rising} needs to go along with the condition $\nu\geq -1$. Furthermore, when $\nu\geq -1$, we can deduce from the condition \eqref{Assump G Diagonal 2 C} that, for large enough $x>0$ and every $\vert h \vert \leq \frac{x^{-\nu}}{2}$,
\begin{equation}\label{Eq Approx}
\mathcal{V}^{(1)}(x+h) \approx \mathcal{V}^{(1)}(x).
\end{equation}
In other word, the values of $\mathcal{V}^{(1)}$ 
can be comparable up to a constant. The proof of \eqref{Eq Approx} can be found in \cite{KS19}, 
but for the reader's convenience, we recall the proof in a simpler way, for $\nu>-1$
\begin{align*}
\left\vert \ln \frac{\vert \mathcal{V}^{(1)}(x+h)\vert}{\vert \mathcal{V}^{(1)}(x)\vert} \right\vert = \left\vert \int_{x}^{x+h} \frac{\mathcal{V}^{(2)}(t)}{\mathcal{V}^{(1)}(t)}\, \dd t\right\vert&\lesssim \left\{ 
\begin{aligned}
& \int_{x}^{x+h} \vert t \vert^{\nu} \, \dd t \qquad &&h\geq 0,\\
& \int_{x+h}^{x} \vert t \vert^{\nu} \, \dd t \qquad &&h\leq 0,
\end{aligned}
\right.\\
&\leq\left\{ 
\begin{aligned}
& h (x+h)^{\nu} \qquad &&h\geq 0,\, \nu \geq 0,\\
& h x^{\nu} \qquad &&h\geq 0,\, \nu < 0,\\
& (-h) x^{\nu} \qquad &&h\leq 0,\, \nu \geq 0,\\
& (-h) (x+h)^{\nu} \qquad &&h\leq 0, \,\nu < 0,\\
\end{aligned}
\right.\\
&\lesssim 1.
\end{align*}
In the last inequality, we used the observation that, for all $\vert h \vert \leq \frac{x^{-\nu}}{2}$ and for $x>0$,
\begin{equation}\label{Compare x}
\frac{x}{2}\leq x+h \leq \frac{3x}{2}.
\end{equation}
The case $\nu=-1$ is treated similarly.
\item From the assumption \eqref{Assump G V Rising} and \eqref{Assump G Diagonal 1 C}, we deduce that
\begin{equation}\label{V and x}
\mathcal{V}(x)\gtrsim x^{\nu+2\varepsilon_{1}},\qquad \forall x\gtrsim 1.
\end{equation}
Indeed, we just need to look at the case $\nu+2\varepsilon_{1}>0$ (then $\frac{3}{2}\nu+\varepsilon_{1}+1>\nu+1\geq 0$). For some fix large $x_{0}>0$ and  for $x\geq x_{0}$, we have,
\[ \int_{x_{0}}^{x} \frac{\mathcal{V}^{(1)}(t)}{\mathcal{V}^{1/2}(t)}\, \dd t\gtrsim\int_{x_{0}}^{x} t^{\frac{3}{2}\nu+\varepsilon_{1}}\, \dd t \Rightarrow\mathcal{V}(x)^{\frac{1}{2}}- \mathcal{V}(x_{0})^{\frac{1}{2}}\gtrsim x^{\frac{3}{2}\nu+\varepsilon_{1}+1}-x_{0}^{\frac{3}{2}\nu+\varepsilon_{1}+1},\]
and thus \eqref{V and x} is obtained when $x$ is considered to be large.
\end{enumerate}

\end{remark}

The cut-off in this case is constructed such that the pseudomode lives around the turning point $x_{\beta}$. Namely, we arrange
\begin{equation}\label{Eq cutoff C}
\begin{aligned}
&\xi \in C_{0}^{\infty}(\R_{+}), \quad 0\leq \xi \leq 1,\\
&\xi(x)=1, \qquad \forall x \in (x_{\beta}-\frac{\delta_{\beta}}{2},x_{\beta}+\frac{\delta_{\beta}}{2})=:J_{\beta}',\\
&\xi(x)=0,\qquad \forall x \notin (x_{\beta}-\delta_{\beta},x_{\beta}+\delta_{\beta})=:J_{\beta},
\end{aligned}
\end{equation}
with 
\begin{equation}\label{Eq G delta C}
\delta_{\beta} := \frac{x_{\beta}^{-\nu}}{2}.
\end{equation}
We see that if $\nu<0$, the support of the pseudomode is able to be extended on $\R_{+}$ when $\beta \to +\infty$. As the WKB construction for the real pseudoeigenvalue, for each $n\in \N_{0}$ and each $\lambda \in \C$, the pseudomode has the form
\begin{equation}\label{Eq pseudomode C}
\Psi_{\lambda,n} := \begin{pmatrix}
k_1 u_{\lambda,n}\\
k_2 v_{\lambda,n}
\end{pmatrix},
\end{equation}
where 
\begin{equation*}
\begin{aligned}
&k_1(x):=\exp{\left(-i\int_{x_{\beta}}^{x} V_{21}(\tau) \, \dd \tau\right)}\qquad \text{and} \qquad k_2(x) := \exp{\left(-i\int_{x_{\beta}}^{x} V_{12}(\tau) \, \dd \tau\right)},\\
&u_{\lambda,n} := \xi \exp\left(-P_{\lambda,n} \right)\qquad \text{and} \qquad v_{\lambda,n}= \frac{k_1}{k_2} \frac{\partial_{x} u_{\lambda,n}}{\lambda+m-V_{22}},\\
&P_{\lambda,n}(x) = \sum_{k=-1}^{n-1}  \int_{x_{\beta}}^{x} \lambda^{-k}\psi_{k}^{(1)}(t) \, \dd t,
\end{aligned}
\end{equation*}
with $\xi$ given in \eqref{Eq cutoff C}, $(\psi_{k}^{(1)})_{k\in[[-1,n-1]]}$ given in \eqref{Eq Recursion}.
\subsection{Main results}
Now, we can state our main theorem in 
the setting of this section.
\begin{theorem}\label{Theorem 3 C}
Let Assumption \ref{Assump General C} holds. Assume that there exists a ($\beta$-dependent) $\alpha$ such that the following conditions hold as $\beta\to +\infty$, for all $x\in J_{\beta}$,
\begin{equation}\label{Eq a 1 C}
\begin{aligned}
 &\left(\alpha-m-\mathrm{Re}\, V_{11}(x)\right)\left(\alpha+m-\mathrm{Re}\, V_{22}(x)\right)>0,\\
 &\vert \alpha-m- \mathrm{Re}\, V_{11}(x) \vert \approx \vert \alpha \vert ,\\
 &\vert \alpha+m-\mathrm{Re}\, V_{22}(x) \vert \approx \vert \alpha\vert ,
\end{aligned}
\end{equation}
and
\begin{equation}\label{Eq a 2 C}
\left(\beta x_{\beta}^{\nu}\right)^{\frac{1}{2}} \lesssim \vert \alpha \vert.
\end{equation}
Let $\{\psi_{k}^{(1)}\}_{k\in[[-1,N-1]]}$ be determined by \eqref{Eq Recursion}  and $\Psi_{\lambda, N}$ defined as in \eqref{Eq pseudomode C}. We choose
\begin{enumerate}[i)]
\item the plus sign in the formula of $\psi_{-1}^{(1)}$ in \eqref{Eq Recursion}  if $\alpha-m-\mathrm{Re}\, V_{11}>0$,
\item the minus sign in the formula of $\psi_{-1}^{(1)}$ in \eqref{Eq Recursion}  if $\alpha-m-\mathrm{Re}\, V_{11}<0$.
\end{enumerate}
Then, for every $c \in (0,1)$, there exists $\beta_{0}>0$ such that, for all $\beta> \beta_{0}$,
\begin{equation*}
\frac{\Vert (H_{V}- \lambda)\Psi_{\lambda,N} \Vert}{\Vert \Psi_{\lambda,N} \Vert} \leq \kappa(\beta,c)+\sigma^{(N)}(\beta),
\end{equation*}
in which
\begin{itemize}
\item $\displaystyle \kappa(\beta,c):=\exp\left(-cF\left(x_{\beta},x_{\beta}-\frac{\delta_{\beta}}{2}\right) \right) +\exp\left(-cF\left(x_{\beta},x_{\beta}+\frac{\delta_{\beta}}{2}\right) \right)=o(1),$
\item $\displaystyle\sigma^{(N)}(\beta):=  \sum_{\ell=-1}^{N-2} \frac{x_{\beta}^{(N+\ell+2)\nu}}{\vert \alpha \vert^{N+\ell+1}}\left( 1+\frac{\beta}{\vert \alpha \vert}\right)^{N+\ell+2},$
\end{itemize}
with $\displaystyle F(x_{\beta},x):=\int_{x_{\beta}}^x [\mathcal{V}(t)-\beta]\, \dd t$.
\end{theorem}
The first condition in \eqref{Eq a 1 C} is exactly the condition \eqref{Eq Cond Re} that allows the regularity of pseudomodes be inherited from the regularity of the potential through the principal square root of $V_{\lambda}$. The last two conditions in \eqref{Eq a 1 C} are inspired from the first ones in \eqref{Rem Re V1122} for the real case of $\lambda$. In order to restrain the wild growth of $(\psi_{k}^{(1)})_{k\in [[0,N-1]]}$, we require \eqref{Eq a 2 C}. Through the statement of the Theorem \ref{Theorem 3 C}, it not only indicates the existence of the pseudomode, but also gives us a way to sketch the pseudospectrum around the infinity by looking for the admissible $\alpha$ (see Subsection~\ref{Application C}). The quantity $\kappa(\beta,c)$ always has an exponential decay at the rate control by the general function $F(x_{\beta},x)$. This is also an improvement upon 
\cite[Thm.~5.1]{KS19} whose rate is only controlled by some polynomial.

\subsection{Intermediate steps}
On the way to prove our main results, some useful lemmata are designed similar to lemmata used in Subsection \ref{Subsection Auxi Steps}.
\begin{lemma}\label{Lemma VK C}
Let the assumptions of Theorem \ref{Theorem 3 C} hold. There exists $\beta_{0}>0$ such that, for all $\beta> \beta_{0}$, for all $\ell \in [[1, N+1]]$ and for all $x\in J_{\beta}$,
\begin{equation*}
\begin{aligned}
&\vert V_{\lambda}^{(\ell)}(x) \vert \lesssim \vert V_{\lambda}(x) \vert \left(1+\frac{\beta}{\vert \alpha \vert} \right)^{\ell}x^{\ell\nu},\\
&\vert K_{\lambda}^{(\ell)}(x) \vert \lesssim \vert K_{\lambda}(x) \vert \left(1+\frac{\beta}{\vert \alpha \vert} \right)^{\ell}x^{\ell\nu}.
\end{aligned}
\end{equation*}
\end{lemma}
\begin{proof}
The proof of this lemma repeats again the procedure of the proof of Lemma \ref{Lemma VK} with the observation that, thanks to the condition \eqref{Eq a 1 C},
\begin{align*}
\frac{\vert V_{11}\vert}{\vert \lambda-m-V_{11} \vert} \lesssim \frac{\vert \alpha - m - \mathrm{Re}\, V_{11} \vert+ \vert \beta - \mathrm{Im}\, V_{11} \vert+ \vert \alpha \vert+ \beta+  m }{\vert \alpha - m - \mathrm{Re}\, V_{11} \vert+\vert \beta - \mathrm{Im}\, V_{11} \vert}\lesssim 1 + \frac{\beta}{\vert \alpha \vert}.
\end{align*}
Correspondingly, we also have
\[ \frac{\vert V_{22}\vert}{\vert \lambda+m-V_{22} \vert}\lesssim 1 + \frac{\beta}{\vert \alpha \vert}.\]
Applying these estimates to \eqref{Eq Estimate V} and \eqref{Eq Estimate K}, the claims follow directly. Especially, in~\eqref{Eq Estimate K}, we notice that
$ 
\sum_{j=1}^{k} \alpha_{j} \leq \sum_{j=1}^{k}j\alpha_{j}=k\leq \ell.
$
\end{proof}
\begin{lemma}\label{Lemma WKB sol C}
Let the assumptions of Theorem \ref{Theorem 3 C} hold and $\{\psi_{k}^{(1)}\}_{k\in[[-1,N-1]]}$ be determined by \eqref{Eq Recursion}. There exists $\beta_{0}>0$ such that, for all $\beta> \beta_{0}$ and for all $x\in J_{\beta}$, we have
\begin{equation*}
F(x_{\beta},x)\leq \int_{x_{\beta}}^{x} \mathrm{Re}(\lambda \psi_{-1}^{(1)}(t))\, \dd t \lesssim F(x_{\beta},x),
\end{equation*}
and for all $k\in [[0,N-1]]$,
\begin{equation*}
\left\vert \lambda^{-k} \psi_{k}^{(1)}(x)\right\vert \lesssim \left(1+\frac{\beta}{\vert \alpha \vert}\right)^{k+1} \frac{x_{\beta}^{(k+1)\nu}}{\vert \alpha \vert^{k}}.
\end{equation*}
\end{lemma}
\begin{proof}
By working as in \eqref{Eq bound denominator}, we also have the inequalities for dealing with the denominator of $\mathrm{Re} (\lambda\psi_{-1}^{(1)})$,
\begin{equation}\label{Eq bound denominator C}
\begin{aligned}
&\sqrt{\vert V_{\lambda} \vert + \mathrm{Re}\, V_{\lambda}} \geq \sqrt{2} \sqrt{(\alpha-m-\mathrm{Re}\, V_{11})(\alpha+m-\mathrm{Re}\, V_{22})},\\
&\sqrt{\vert V_{\lambda} \vert + \mathrm{Re}\, V_{\lambda}}\leq \frac{1}{\sqrt{2}}\sqrt{(\mathrm{Im} \, V_{11}-\mathrm{Im} \, V_{22})^2+(2\alpha-\mathrm{Re}\, V_{11}-\mathrm{Re}\, V_{22})^2}.
\end{aligned}
\end{equation}
Notice that, there are two cases:
\begin{enumerate}[i)]
\item If $\alpha-m-\mathrm{Re}\, V_{11}>0$ on $J_{\beta}$, then the formula of $\mathrm{Re} (\lambda\psi_{-1}^{(1)})$ has the expression
\begin{align*}
\mathrm{Re}(\lambda \psi_{-1}^{(1)}) = \frac{1}{\sqrt{2}} \frac{(\mathcal{V} -\beta)(2\alpha-\mathrm{Re}\, V_{11}-\mathrm{Re}\, V_{22})}{\sqrt{\vert V_{\lambda} \vert + \mathrm{Re}\, V_{\lambda}}}.
\end{align*}
\item If $\alpha-m-\mathrm{Re}\, V_{11}<0$ on $J_{\beta}$, then the formula of $\mathrm{Re} (\lambda\psi_{-1}^{(1)})$ has the expression
\begin{align*}
\mathrm{Re}(\lambda \psi_{-1}^{(1)}) =- \frac{1}{\sqrt{2}} \frac{(\mathcal{V} -\beta)(2\alpha-\mathrm{Re}\, V_{11}-\mathrm{Re}\, V_{22})}{\sqrt{\vert V_{\lambda} \vert + \mathrm{Re}\, V_{\lambda}}}.
\end{align*}
\end{enumerate}
By using the second estimate in \eqref{Eq bound denominator C} and observing the sign of the term $\mathcal{V}(x)-\beta$ on the left and on the right of $x_{\beta}$ on $J_{\beta}$, it implies that, for all $x\in J_{\beta}$,
\begin{align*}
\int_{x_{\beta}}^{x} \mathrm{Re}\, (\lambda \psi_{-1}^{(1)}(t)) \, \dd t \geq \int_{x_{\beta}}^{x} [\mathcal{V}(x)-\beta] \, \dd t.
\end{align*}
The rest upper bound for the integral of $\mathrm{Re} (\lambda\psi_{-1}^{(1)})$ on $J_{\beta}$ is given by the first estimate in \eqref{Eq bound denominator C} and \eqref{Eq a 1 C}. 

For each $k\geq 0$, we establish the proof for $\lambda^{-k} \psi_{k}^{(1)}$ as in the proof of Lemma~\ref{Lemma WKB sol}. Then, Lemma~\ref{Lemma VK C} implies that, for all $x\in J_{\beta}$,
\begin{align*}
\left\vert \lambda^{-k} \psi_{k}^{(1)}(x) \right\vert \lesssim  \left(1+\frac{\beta}{\vert \alpha \vert} \right)^{k+1} \sup_{x\in J_{\beta}} \frac{x^{(k+1)\nu}}{\vert V_{\lambda}(x) \vert^{k/2}} \lesssim \left(1+\frac{\beta}{\vert \alpha \vert} \right)^{k+1} \frac{x_{\beta}^{k\nu}}{\vert \alpha\vert^{k}}.
\end{align*}
Here, the last inequality is given by \eqref{Eq a 1 C} and \eqref{Compare x}.
\end{proof}
For all $x\in J_{\beta}$, by changing variable twice in integrals, we have
\begin{align*}
F(x_{\beta},x) =(x-x_{\beta})^2 \int_{0}^1 \int_{0}^1 \xi \mathcal{V}^{(1)}\left(x_{\beta}+ \tau \xi(x-x_{\beta}) \right)\, \dd \tau \dd \xi.
\end{align*}
From \eqref{Eq Approx}, it yields that, for all $x\in J_{\beta}$,
\begin{equation}\label{Eq F C}
F(x_{\beta},x) \approx \mathcal{V}^{(1)}(x_{\beta})(x-x_{\beta})^2.
\end{equation}
Using this approximation, we obtain the following lemma.

\begin{lemma}\label{Lemma Decay solution C}
Let the assumptions of Theorem \ref{Theorem 3 C} hold and let $\{\psi_{k}^{(1)}\}_{k\in[[-1,N-1]]}$ be determined by \eqref{Eq Recursion}. Then
\begin{enumerate}[i)]
\item for every $\varepsilon\in (0,1)$ and for every $\eta \in (0,1)$, there exists $\beta_{0}>0$ such that, for all $\beta> \beta_{0}$ and for all $x\in J_{\beta}\setminus \left(x_{\beta}-\eta \delta_{\beta}, x_{\beta}+\eta \delta_{\beta}\right)$,
\begin{equation*}
\left\vert k_1(x) \exp(-P_{\lambda,N}(x))\right\vert\lesssim \exp\left( -(1-\varepsilon)F(x_{\beta},x)\right);
\end{equation*}
\item there exists $C>0$ such that for every $\eta \in (0,1)$, there exists $\beta_{0}>0$ such that, for all $\beta> \beta_{0}$ and for all $x\in J_{\beta}\setminus \left(x_{\beta}-\eta \delta_{\beta}, x_{\beta}+\eta \delta_{\beta}\right)$,
\begin{equation*}
\left\vert k_1(x) \exp(-P_{\lambda,N}(x))\right\vert\gtrsim \exp\left( -C F(x_{\beta},x) \right).
\end{equation*}
\end{enumerate}
 \end{lemma}

\begin{proof}
From the equality of two imaginary parts $\mathrm{Im}\,V_{11} = \mathrm{Im}\,V_{11}$, we have
\begin{equation}\label{V11 approx V22 C}
\vert \lambda-m-V_{11}(x) \vert \approx \vert \lambda+m-V_{22}(x) \vert,\qquad \forall x \in J_{\beta}.
\end{equation}
Indeed, from the choice of $\alpha$ in \eqref{Eq a 1 C}, it turns out that
\begin{equation}
\frac{\vert \lambda-m-V_{11}(x) \vert}{\vert \lambda+m-V_{22}(x) \vert}\lesssim \frac{\vert \alpha-m-\mathrm{Re}\, V_{11} \vert+\vert \beta-\mathcal{V}\vert}{\vert \alpha+m-\mathrm{Re}\, V_{22} \vert+\vert \beta-\mathcal{V}\vert} \lesssim \frac{\mathcal{O}(\vert \alpha \vert)+\vert \beta-\mathcal{V}\vert}{\vert \alpha \vert +\vert \beta-\mathcal{V}\vert}\lesssim 1,
\end{equation}
and the other direction is similar. Therefore, as in the proof of Lemma \eqref{Lemma Decay solution}, we obtain the following approximation on $J_{\beta}$:
\begin{equation*}
\left\vert k_1(x) \exp(-P_{\lambda,N}(x))\right\vert \approx \exp\left(\frac{1}{2}\int_{x_{\beta}}^{x} \mathcal{U}(t)\, \dd t- \sum_{\substack{k=-1\\k\neq 0}}^{N-1}  \int_{x_{\beta}}^{x} \mathrm{Re} \,(\lambda^{-k}\psi_{k}^{(1)}(t))\, \dd t\right).
\end{equation*}

Let $\eta \in (0,1)$. Using Lemma~\ref{Lemma WKB sol C} and condition~\eqref{Assump G Off-Diagonal C}, one obtains, when $x_{\beta}$ large enough and for all $x\in J_{\beta}\setminus \left(x_{\beta}-\eta \delta_{\beta}, x_{\beta}+\eta \delta_{\beta}\right)$,
\begin{align*}
\left\vert\frac{ \int_{x_{\beta}}^{x} \mathcal{U}(t)\, \dd t }{\int_{x_{\beta}}^{x}\mathrm{Re}\,\left(\lambda \psi_{-1}^{(1)}(t)\right) \, \dd t} \right\vert \lesssim \frac{ \vert x-x_{\beta}\vert o(x_{\beta}^{-\nu}\mathcal{V}^{(1)}(x_{\beta}))}{\vert x-x_{\beta}\vert^2 \mathcal{V}^{(1)}(x_{\beta})} \lesssim \frac{1}{\eta}\frac{o(x_{\beta}^{-\nu}\mathcal{V}^{(1)}(x_{\beta}))}{  x_{\beta}^{-\nu}\mathcal{V}^{(1)}(x_{\beta})}= o(1).
\end{align*}
By employing Lemma \ref{Lemma WKB sol C} for $k\in [[1,N-1]]$ and for all $x\in J_{\beta}\setminus \left(x_{\beta}-\eta \delta_{\beta}, x_{\beta}+\eta \delta_{\beta}\right)$, we have
\begin{equation*}
\begin{aligned}
\frac{ \displaystyle \int_{x_{\beta}}^{x} \left\vert \lambda^{-k}\psi_{k}^{(1)}(t)\right\vert \, \dd t }{\displaystyle \int_{x_{\beta}}^{x} \mathrm{Re}\,(\lambda\psi_{-1}^{(1)}(t)) \, \dd t}  &\lesssim \frac{\displaystyle \vert x-x_{\beta} \vert\left(1+\frac{\beta}{\vert \alpha \vert}\right)^{k+1} \frac{x_{\beta}^{(k+1)\nu}}{\vert \alpha \vert^{k}}}{\vert x-x_{\beta} \vert^2 \mathcal{V}^{(1)}(x_{\beta})}\lesssim \frac{\displaystyle \left(1+\frac{\beta}{\vert \alpha \vert}\right)^{k+1}x_{\beta}^{(k+2)\nu}}{\displaystyle \eta \mathcal{V}^{(1)}(x_{\beta})\vert \alpha \vert^{k}}\\
&\lesssim \left\{
\begin{aligned}
& \frac{1}{\eta}\left(\frac{x_{\beta}^{\nu}}{\beta} \right)^{k+\frac{1}{2}}\frac{\beta^{\frac{1}{2}}x_{\beta}^{\frac{3}{2}\nu}}{\mathcal{V}^{(1)}(x_{\beta})},\qquad &&\text{ if } \vert \alpha \vert> \beta,\\
& \frac{1}{\eta}\left( \frac{\beta x_{\beta}^{\nu}}{\vert \alpha \vert^2}\right)^{k+\frac{1}{2}}  \frac{\beta^{\frac{1}{2}}x_{\beta}^{\frac{3}{2}\nu}}{\mathcal{V}^{(1)}(x_{\beta})},\qquad &&\text{ if } \vert \alpha \vert \leq \beta,
\end{aligned}
\right.\\
&\lesssim \frac{1}{\eta}x_{\beta}^{-\varepsilon}=o(1).
\end{aligned}
\end{equation*}
Here, in case $\vert \alpha \vert\leq \beta $, we have used \eqref{Eq a 2 C} and in other case, we have used \eqref{V and x}. The conclusion of the lemma is obviously deduced from Lemma~\ref{Lemma WKB sol C}.
\end{proof}
\subsection{Proofs of the main results}
\begin{proof}[Proof of Theorem \ref{Theorem 3 C}]
Fix $c\in (0,1)$ and consider $\kappa(\beta,c)$ defined in the statement of the theorem. We start the proof by showing that
\begin{equation}\label{Eq L^2 xi C}
\frac{\left\Vert \frac{k_1 \exp(-P_{\lambda,N})}{\lambda+m-V_{22}}  \xi^{(2)} \right\Vert_{L^2(\R_{+})} +\left\Vert \frac{k_1 \exp(-P_{\lambda,N})}{\lambda+m-V_{22}} \left( 2P_{\lambda,N}^{(1)}-\frac{K_{\lambda}^{(1)}}{K_{\lambda}}\right) \xi^{(1)} \right\Vert_{L^2(\R_{+})}}{\sqrt{\Vert k_1 \exp(-P_{\lambda,N})\xi \Vert_{L^2(\R_{+})}^2+ \left\Vert \frac{k_1(-i\partial_{x})\left(\exp(-P_{\lambda,N}) \xi\right)}{\lambda+m-V_{22}} \right\Vert_{L^2(\R_{+})}^{2}}}\leq \kappa(\beta,c).
\end{equation}
By choosing $\varepsilon$ in Lemma \ref{Lemma Decay solution C} sufficiently small such that $1-\varepsilon> c$, we can write
$1-\varepsilon = c+2\tilde{c},$
for some $\tilde{c}>0$.
The plan is to use the upper bound of $k_{1} \exp\left(-P_{\lambda,N}\right)$ in Lemma \ref{Lemma Decay solution C} to control the terms in the numerator of \eqref{Eq L^2 xi C} and employ the lower bound of $k_{1} \exp\left(-P_{\lambda,N}\right)$ for the denominator. We start with the term attached with $\xi^{(2)}$:
\begin{align*}
\lefteqn{\left\Vert \frac{k_1 \exp(-P_{\lambda,N})}{\lambda+m-V_{22}}  \xi^{(2)} \right\Vert_{L^2(\R_{+})}^2} & \\
&\lesssim \frac{\delta_{\beta}^{-4}}{\vert \alpha \vert^2} \left( \int_{x_{\beta}-\delta_{\beta}}^{x_{\beta}-\frac{\delta_{\beta}}{2}}  \exp\left(-2(1-\varepsilon)F(x_{\beta},x) \right) \dd x +  \int_{x_{\beta}+\frac{\delta_{\beta}}{2} }^{x_{\beta}+\delta_{\beta}}  \exp\left(-2(1-\varepsilon)F(x_{\beta},x) \right)  \dd x\right)\\
&\lesssim x_{\beta}^{2\nu}\left(\exp\left(-2(1-\varepsilon)F\left(x_{\beta},x_{\beta}-\frac{\delta_{\beta}}{2}\right) \right) +\exp\left(-2(1-\varepsilon)F\left(x_{\beta},x_{\beta}+\frac{\delta_{\beta}}{2}\right) \right) \right)\\
& \lesssim \exp\left(-2(c+\tilde{c}) F\left(x_{\beta},x_{\beta}-\frac{\delta_{\beta}}{2}\right) \right) +\exp\left(-2(c+\tilde{c})F\left(x_{\beta},x_{\beta}+\frac{\delta_{\beta}}{2}\right) \right).
\end{align*}
In the second inequality, we used \eqref{Eq a 2 C} and the fact that $F(x_{\beta},x)$ is increasing as~$x$ goes far from~$x_{\beta}$. Furthermore, all appearing polynomial terms will be restrained, with some positive constants $C_{1}, C_{2}>0$, by
\begin{equation}\label{Eq Exponential C}
\begin{aligned}
\max &\left\{\exp\left(-\tilde{c}F\left(x_{\beta},x_{\beta}-\frac{\delta_{\beta}}{2}\right) \right),\exp\left(-\tilde{c}F\left(x_{\beta},x_{\beta}+\frac{\delta_{\beta}}{2}\right) \right)\right\}\\
\leq &\ \exp\left(-C_{1} \mathcal{V}^{(1)}(x_{\beta})x_{\beta}^{-2\nu} \right)\\
\leq &\ \exp\left( - C_{2}x_{\beta}^{2\varepsilon_{1}}\right),
\end{aligned}
\end{equation}
which follows directly from \eqref{Eq F C}, the definition of $\delta_{\beta}$ in \eqref{Eq G delta C} and the condition \eqref{Assump G Diagonal 1 C} and \eqref{V and x}. Thus, we have
\begin{equation}\label{xi'}
\begin{aligned}
\left\Vert \frac{k_1\exp(-P_{\lambda,N})}{\lambda+m-V_{22}}  \xi^{(2)} \right\Vert_{L^2(\R_{+})}\lesssim  &\ \exp\left(-(c+\tilde{c})F\left(x_{\beta},x_{\beta}-\frac{\delta_{\beta}}{2}\right) \right) \\
& + \exp\left(-(c+\tilde{c})F\left(x_{\beta},x_{\beta}+\frac{\delta_{\beta}}{2}\right) \right).
\end{aligned}
\end{equation}
The second term in the numerator is bounded in the same manner. In detail, we look at the expression
\[ 2P_{\lambda, N}^{(1)}-\frac{K_{\lambda}^{(1)}}{K_{\lambda}} = 2i V_{\lambda}^{1/2}+2\sum_{k=0}^{N-1} \lambda^{-k}\psi_{k}^{(1)}-\frac{K_{\lambda}^{(1)}}{K_{\lambda}}.\]
For all $k\in [[0,N-1]]$ and for all $x\in J_{\beta}\setminus J_{\beta}'$, thanks to \eqref{Eq a 2 C} and \eqref{V and x}, we have
\begin{align*}
\frac{\vert \lambda^{-k} \psi_{k}^{(1)}(x)\vert}{\vert V_{\lambda}(x)\vert^{\frac{1}{2}}}\lesssim \left(1+\frac{\beta}{\vert \alpha \vert}\right)^{k+1}\frac{x_{\beta}^{(k+1)\nu}}{\vert \alpha \vert^{k+1}}\lesssim \left\{
\begin{aligned}
&\left( \frac{ x_{\beta}^{\nu}}{\beta}\right)^{k+1},\qquad &&\text{ if } \vert \alpha \vert >\beta,\\
&\left( \frac{\beta x_{\beta}^{\nu}}{\vert \alpha \vert^2}\right)^{k+1},\qquad &&\text{ if } \vert \alpha \vert \leq \beta.
\end{aligned}
\right. \qquad \lesssim 1.
\end{align*}
The term related to $K_{\lambda}$ is estimated as same as $\psi_{0}^{(1)}$. Then, from \eqref{V11 approx V22 C}, we have
\[ \frac{1}{\vert \lambda+m-V_{22}\vert }\left\vert 2P_{\lambda, N}^{(1)}-\frac{K_{\lambda}^{(1)}}{K_{\lambda}}\right\vert\lesssim \frac{\vert V_{\lambda}(x)\vert^{\frac{1}{2}}}{\vert \lambda+m-V_{22}\vert }\lesssim 1.\]
Therefore, we also obtain the same estimate as \eqref{xi'} for the term attached with $\xi^{(1)}$ in \eqref{Eq L^2 xi C}. We consider $\eta>0$ in Lemma \ref{Lemma Decay solution C} such that $2\eta<\frac{1}{2}$, then it leads to
\begin{align*}
\displaystyle\left\Vert k_1 \exp(-P_{\lambda,N})\xi \right\Vert_{L^2(\R_{+})}^2 \gtrsim \  &\int_{x_{\beta}+\eta \delta_{\beta}}^{x_{\beta}+2\eta \delta_{\beta}} \exp\left(-CF(x_{\beta},x)\right) \, \dd x \\
\gtrsim \ &\eta x_{\beta}^{-\nu}\exp\left(-2CF\left(x_{\beta},x_{\beta}+2\eta\delta_{\beta}\right) \right).
\end{align*}
Using \eqref{Eq F C} and choosing $\eta$ small enough, we have
\begin{align*}
(c+\tilde{c})F\left(x_{\beta},x_{\beta}-\frac{\delta_{\beta}}{2}\right)- CF\left(x_{\beta},x_{\beta}+2\eta\delta_{\beta}\right)
 \geq & \left( 1- \mathcal{O}( \eta^2 )\right)(c+\tilde{c})F\left(x_{\beta},x_{\beta}-\frac{\delta_{\beta}}{2}\right)
\end{align*}
and similarly
\[ (c+\tilde{c})F\left(x_{\beta},x_{\beta}+\frac{\delta_{\beta}}{2}\right)- CF\left(x_{\beta},x_{\beta}+2\eta\delta_{\beta}\right)\geq (1-\mathcal{O}( \eta^2 ))(c+\tilde{c})F\left(x_{\beta},x_{\beta}+\frac{\delta_{\beta}}{2}\right).\]
Then, \eqref{Eq L^2 xi C} follows by choosing $\eta$ sufficiently small and \eqref{Eq Exponential C}. 

Finally, we estimate the remainder \eqref{Eq Estimate Rn} 
by using Lemma~\ref{Lemma VK C}.
\end{proof}


\subsection{Applications}\label{Application C}
Let us list here some examples which are direct consequences of 
Theorem~\ref{Theorem 3 C}. We will see that the shape of the pseudospectrum depends not only on the type of the potentials, but also on their regularity.

\begin{example}\label{Example Log}
First of all, we want to consider a kind of logarithmic potential on $\R_{+}$:
\begin{equation}\label{Eq Logarithm C}
 V(x) := \begin{pmatrix}
i \ln(x) && u(x)\\
u(x)  && i \ln(x)
\end{pmatrix},
\end{equation}
where $u\in W_{\mathrm{loc}}^{N,2}(\R_{+})$, with $N\in\N_{1}$, 
is such that $\vert u(x) \vert= o(1)$ as $x\to +\infty$. Then all conditions of Assumption~\ref{Assump General C} are satisfied with $\nu=-1$ and any $\varepsilon_{1}\in \left(0,\frac{1}{2}\right)$. Given $\beta>0$, then $x_{\beta}>0$ is determined by the relation
$x_{\beta}=e^{\beta}.$
About the quantity $\kappa(\beta,c)$ for  $c\in(0,1)$, on account of the estimate \eqref{Eq Exponential C}, it decays in a superexponential way independent of the choice of $\alpha$, with some constant $C>0$,
$\kappa(\beta,c) = \mathcal{O}\left(\exp\left(-C e^{\beta}\right)\right)$.
Since $\mathrm{Re}\, V_{11}=\mathrm{Re}\, V_{22}=0$, the condition \eqref{Eq a 1 C} of Theorem~\ref{Theorem 3 C} are clearly satisfied if and only if 
\begin{equation}\label{Ex1 a 0}
\vert \alpha \vert>m.
\end{equation}
While \eqref{Eq a 2 C} is assured if and only if
\begin{equation}\label{Ex1 a 1}
\vert \alpha\vert \gtrsim \beta^{\frac{1}{2}} \exp\left(-\frac{\beta}{2}\right).
\end{equation}
Therefore, we consider two cases:
\begin{enumerate}[a)]
\item If $m=0$, by the choice \eqref{Ex1 a 1},  there exist a number $\beta_{0}>0$ and a family $\left(\Psi_{\lambda,N}\right)_{\lambda\in \Omega}$ whose supports are contained in $\R_{+}$ such that
\[ \frac{\Vert (H_{V}- \lambda)\Psi_{\lambda,N} \Vert}{\Vert \Psi_{\lambda,N} \Vert} = \left\{
\begin{aligned}
&\mathcal{O}\left(\beta^{-N}\exp(-(N+1)\beta)\right), \qquad &&\text{ if } \vert \alpha \vert>\beta,\\
&\mathcal{O}\left(\beta^{\frac{1}{2}} \exp\left(-\frac{\beta}{2}\right)\right), \qquad &&\text{ if } \vert \alpha \vert\leq \beta,
\end{aligned}
\right.\]
where
\begin{equation}\label{Eq Log Omega C 1}
\Omega:= \left\{\alpha+i\beta \in \C : \beta> \beta_{0} \text{ and } \vert \alpha \vert\gtrsim \beta^{\frac{1}{2}} \exp\left(-\frac{\beta}{2}\right)\right\}.
\end{equation}
\item If $m>0$,  by the choice \eqref{Ex1 a 0},  there exist a number $\beta_{0}>0$ and a family $\left(\Psi_{\lambda,N}\right)_{\lambda\in \Omega}$ whose supports are contained in $\R_{+}$ such that
\[ \frac{\Vert (H_{V}- \lambda)\Psi_{\lambda,N} \Vert}{\Vert \Psi_{\lambda,N} \Vert} = \left\{
\begin{aligned}
&\mathcal{O}\left(\beta^{-N}\exp(-(N+1)\beta)\right), \qquad &&\text{ if } \vert \alpha \vert>\beta,\\
&\mathcal{O}\left(\beta^{2N} \exp\left(-(N+1)\beta\right)\right), \qquad &&\text{ if } \vert \alpha \vert\leq \beta,
\end{aligned}
\right.\]
where
\begin{equation}\label{Eq Log Omega C 2}
\Omega:= \left\{\alpha+i\beta \in \C : \beta> \beta_{0} \text{ and } \vert \alpha \vert> m \right\}.
\end{equation}
\end{enumerate}

From the definition of $\Omega$, we see that the pseudospectral region contains even points which stay very close to the line $\alpha=0$ when $m=0$ and when $\beta$ large enough (see Figure~\ref{fig:Pic1 C}).

\begin{figure}[h]
 \centering
\begin{subfigure}[c]{0.35\textwidth}
\centering
\includegraphics[width=1\textwidth]{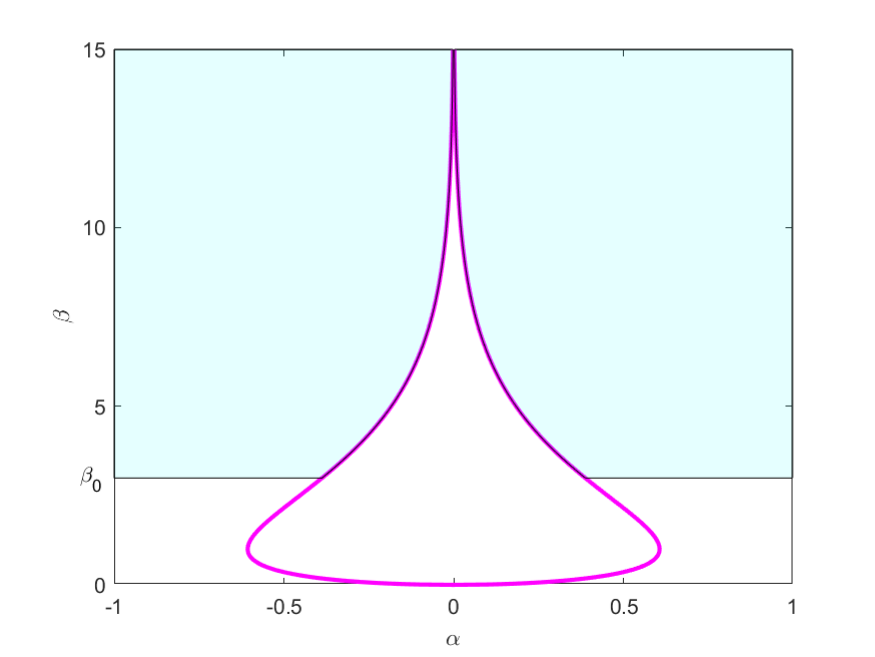}
\caption{$m=0$.}
\end{subfigure}
~
\begin{subfigure}[c]{0.35\textwidth}
 \centering
\includegraphics[width=1\textwidth]{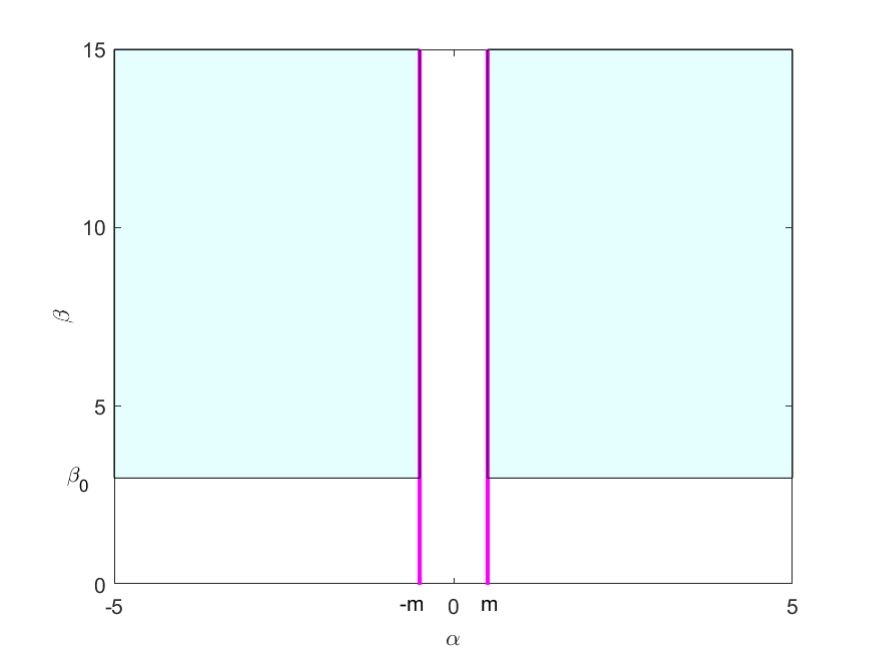}
\caption{$m>0$.}
\end{subfigure}
\captionsetup{singlelinecheck=off}
\caption[foo bar]{Illustration of the shapes of $\Omega$ (in cyan color) with the logarithmic potential $V$ given in \eqref{Eq Logarithm C} in two cases:
\begin{enumerate}[(A)]
  \item $m=0$: The \enquote{wine decanter} curve is the graph of $\vert \alpha \vert = \beta^{\frac{1}{2}}\exp\left(\frac{-\beta}{2}\right)$,
  \item $m>0$: The vertical lines are the graphs of $\vert \alpha \vert= m$.
\end{enumerate}
}\label{fig:Pic1 C}
\end{figure}
\end{example}

\begin{example}\label{Example Pol}
Next, we want to study the polynomial-like potential on $\R_{+}$ in the following form
\begin{equation}\label{Eq Polynomial C}
V(x) = \begin{pmatrix}
i x^{\gamma} && v(x)\\
v(x)  && i x^{\gamma}
\end{pmatrix},
\end{equation}
where $\gamma>0$, $v\in W_{\mathrm{loc}}^{N,2}(\R_{+})$ with $N\in \N_{1}$ such that $\vert v(x) \vert= o(x^{\gamma})$ for $x\to +\infty$. Then all the conditions of Assumption \ref{Assump General C} are satisfied with $\nu=-1$ and any $\varepsilon_{1}\in\left(0,\frac{\gamma+1}{2}\right)$. Given $\beta>0$, then $x_{\beta}>0$ is determined, see \eqref{Eq Turning point}, by
$x_{\beta}=\beta^{\frac{1}{\gamma}}.$
From the estimate \eqref{Eq Exponential C}, for any $c\in(0,1)$, the quantity $\kappa(\beta,c)$  has an exponentially decay, with some $C>0$,
\begin{align*}
\kappa(\beta,c) = \mathcal{O}\left(\exp\left(-C \beta^{\frac{\gamma+1}{\gamma}}\right)\right).
\end{align*}
Since $\mathrm{Re}\, V_{11}=\mathrm{Re}\, V_{22}=0$, the constraints in \eqref{Eq a 1 C} imposed on $\alpha$ are satisfied if and only if
\begin{equation}\label{Ex2 a 1}
\vert \alpha \vert >m.
\end{equation}
We can compute directly the left-hand side of \eqref{Eq a 2 C} as a function of $\beta$:
$ \left(\beta x_{\beta}^{\nu}\right)^{\frac{1}{2}}=\beta^{\frac{1}{2}\frac{\gamma-1}{\gamma}}.$
We consider two cases:
\begin{enumerate}[i)]
\item If $\gamma\geq 1$, we may take $\alpha$ as (with $\varepsilon>0$)
\begin{equation*}
\vert \alpha \vert\gtrsim  \beta^{\frac{1}{2}\frac{\gamma-1}{\gamma}+\varepsilon}.
\end{equation*}
Next, we are concerned about how small $\varepsilon$ can be chosen such that we have the decay of $\sigma^{N}(\beta)$. 
For $\beta>0$ large enough, we have 
\begin{equation*}
\sigma^{(N)}(\beta) \lesssim \left\{
\begin{aligned}
&\beta^{-\left( 1+\frac{1}{\gamma}\right)N-\frac{1}{\gamma}}, \qquad &\text{if } \vert \alpha \vert>\beta,\\
&\beta^{-\varepsilon(2N+1)+ \frac{1}{2}\frac{\gamma-1}{\gamma}}, \qquad &\text{if } \vert \alpha \vert\leq \beta.
\end{aligned}
\right.
\end{equation*}
In order to have a decay for $\sigma^{(N)}(\beta)$ as $\beta\to+\infty$, we choose 
\begin{equation*}
\varepsilon=\frac{1}{4N+2}\frac{\gamma-1}{\gamma}+\eta,\qquad \text{ with } \eta>0.
\end{equation*}
In summary, for any $ \eta >0$, there exist $\beta_{0}>0$ and a family $(\Psi_{\lambda,N})_{\lambda\in \Omega}$ such that
\begin{equation*}
\frac{\Vert (H_{V}- \lambda)\Psi_{\lambda,N} \Vert}{\Vert \Psi_{\lambda,N} \Vert} =\left\{
\begin{aligned}
&\mathcal{O}\left(\beta^{-\left( 1+\frac{1}{\gamma}\right)N-\frac{1}{\gamma}}\right), \qquad &\text{if } \vert \alpha \vert>\beta,\\
&\mathcal{O}\left(\beta^{-\eta(2N+1)}\right), \qquad &\text{if } \vert \alpha \vert\leq \beta,
\end{aligned}
\right.
\end{equation*}
has the decay at the polynomial rate and the pseudospectral region $\Omega$ is defined as
\begin{equation}\label{Eq Pol Omega}
\Omega:=\left\{\alpha+i\beta \in \C: \beta> \beta_{0} \text{ and } \vert \alpha \vert\gtrsim  \beta^{\left(\frac{1}{2}+\frac{1}{4N+2}\right)\frac{\gamma-1}{\gamma}+\eta} \right\}.
\end{equation}
\item If $0<\gamma<1$, the condition \eqref{Eq a 2 C} is equivalent to
\begin{equation}\label{Ex2 a 2}
\vert \alpha \vert \gtrsim \beta^{\frac{1}{2}\frac{\gamma-1}{\gamma}}.
\end{equation}
Depending on the value of $m$, we can compare compare two conditions \eqref{Ex2 a 1} and \eqref{Ex2 a 2} as $\beta\to +\infty$. We have two cases as below.
 \begin{enumerate}
 \item If $m=0$, by the choice \eqref{Ex2 a 2}, there exists $\beta_{0}>0$  such that, for all $\lambda$ belonging to the set
\begin{equation}\label{Eq Pol Omega 2}
\Omega= \left\{\alpha+i\beta \in \C: \beta> \beta_{0} \text{ and } \vert \alpha \vert> \beta^{\frac{1}{2}\frac{\gamma-1}{\gamma}} \right\},
\end{equation}
our problem has the decay
\[\frac{\Vert (H_{V}- \lambda)\Psi_{\lambda,N} \Vert}{\Vert \Psi_{\lambda,N} \Vert} = \left\{
\begin{aligned}
&\mathcal{O}\left(\beta^{-\left( 1+\frac{1}{\gamma}\right)N-\frac{1}{\gamma}}\right), \qquad &\text{if } \vert \alpha \vert>\beta,\\
&\mathcal{O}\left(\beta^{-\frac{1}{2}\frac{1-\gamma}{\gamma}}\right), \qquad &\text{if } \vert \alpha \vert\leq \beta.
\end{aligned}
\right.\]
\item  If $m>0$, by the choice \eqref{Ex2 a 1}, there exists $\beta_{0}>0$  such that, for all $\lambda$ belonging to the set
\begin{equation}\label{Eq Pol Omega 3}
\Omega= \left\{\alpha+i\beta \in \C: \beta> \beta_{0} \text{ and } \vert \alpha \vert> m \right\},
\end{equation}
our problem has the decay
\[\frac{\Vert (H_{V}- \lambda)\Psi_{\lambda,N} \Vert}{\Vert \Psi_{\lambda,N} \Vert} =\left\{
\begin{aligned}
&\mathcal{O}\left(\beta^{-\left( 1+\frac{1}{\gamma}\right)N-\frac{1}{\gamma}}\right), \qquad &\text{if } \vert \alpha \vert>\beta,\\
&\mathcal{O}\left(\beta^{-(N+1)\frac{1-\gamma}{\gamma}}\right), \qquad &\text{if } \vert \alpha \vert\leq \beta.
\end{aligned}
\right.\]
\end{enumerate}
\end{enumerate}

The reader is also invited to compare our results with the application of the same method for Schr\"{o}dinger operators with the polynomial potential $V(x):=ix^{\gamma}$ with $\gamma\geq 1$ in \cite[Ex.~5.3]{KS19}. We see that the pseudospectra of the Dirac operators 
are larger than those of the Schr\"{o}dinger operators. While the outcome $\alpha$ is kept between two curves in the Schr\"{o}dinger case, the outcome $\alpha$ in the Dirac case is just bounded from below by a curve. Technically, 
this can be explained by the appearance of $\lambda$ in the denominator of the estimate $\mathrm{Re}\, \left(\lambda \psi_{-1}^{(1)}\right)$ for the Schr\"{o}dinger operator, in which the above bound of $\alpha$ is employed (to be clear, \cite[Est.~(5.9)]{KS19}). Furthermore, \cite[Ex. 5.3]{KS19} only investigates the case $\gamma\geq 1$, while ours produce the results for even $0<\gamma<1$. Finally, the decay of the problem in \cite{KS19} is attained only when $N$ large enough, while our method gives us the decay even for small $N$. In Figure~\ref{fig: Pol}, we see some representatives for the shape of $\Omega$ corresponding to the power $\gamma$ and the value of $m$ (when $0<\gamma<1$). For the faster growing of the polynomial, the pseudospectrum region $\Omega$ stands further away the axis $\alpha=0$.
\begin{figure}[h!]
 \centering
\begin{subfigure}[c]{0.3\textwidth}
\centering
\includegraphics[width=1\textwidth]{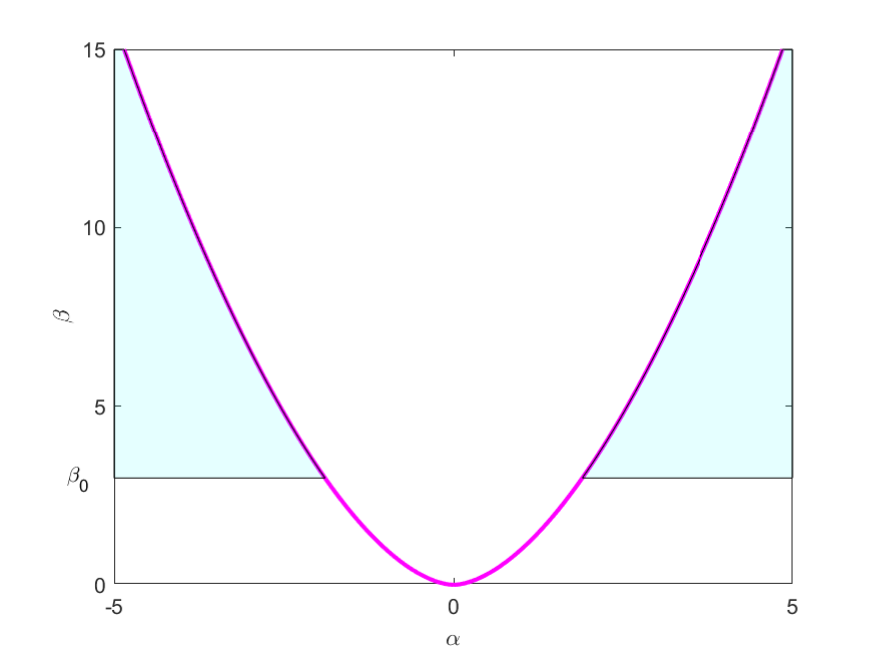}
\caption{$\gamma=2$.}
\end{subfigure}
~
\begin{subfigure}[c]{0.3\textwidth}
\centering
\includegraphics[width=1\textwidth]{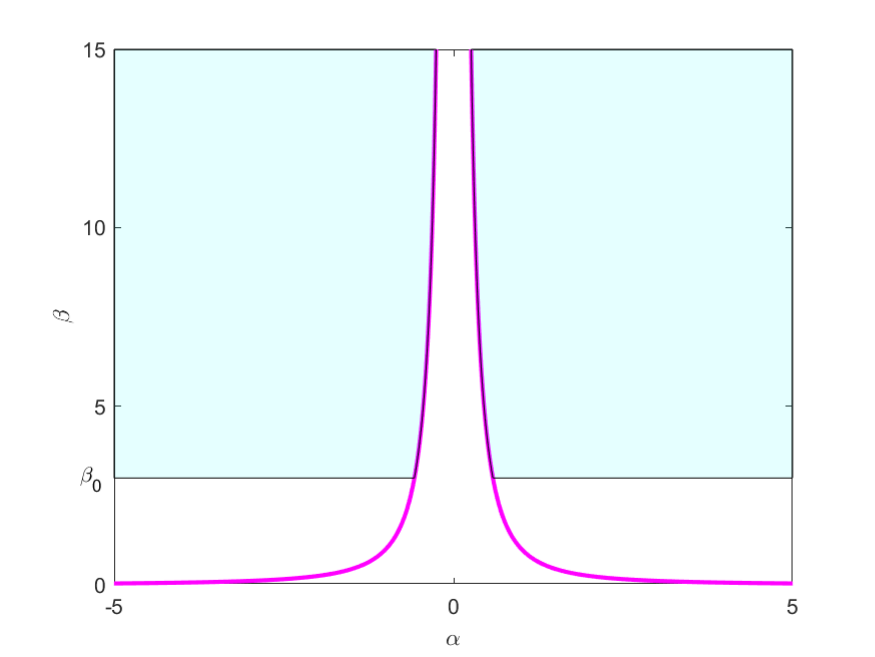}
\caption{$\gamma=\frac{1}{2}$ and $m=0$.}
\end{subfigure}
~
\begin{subfigure}[c]{0.3\textwidth}
 \centering
\includegraphics[width=1\textwidth]{Pol3}
\caption{$\gamma=\frac{1}{2}$ and $m>0$.}
\end{subfigure}
\captionsetup{singlelinecheck=off}
\caption[foo bar]{Illustrations of the shapes of $\Omega$ (in cyan color) associated with the potential $V$ given in \eqref{Eq Polynomial C} corresponding to $\gamma$ and the value of $m$. The magenta curves are, respectively, 
\begin{enumerate}[(A)]
  \item $\gamma=2$: $\vert \alpha \vert=\beta^{\frac{1}{4}+\frac{1}{3}}$, here we took $\eta$ in \eqref{Eq Pol Omega} such that $\frac{1}{8n+4}+\eta=\frac{1}{3}$,
  \item $\gamma=\frac{1}{2}$ and $m=0$: $\vert \alpha \vert=\beta^{-\frac{1}{2}}$,
  \item $\gamma=\frac{1}{2}$ and $m>0$: $\vert \alpha \vert= m$.
\end{enumerate}}\label{fig: Pol}
\end{figure}
\end{example}

\begin{example}
The next example that we want to study is the potential whose $\mathrm{Im}\, V_{11}=\mathrm{Im}\, V_{22}$ is an exponential function:
\begin{equation}\label{Eq Exp C}
 V(x) := \begin{pmatrix}
i e^{x^{\gamma}} && u(x)\\
u(x)  && i e^{x^{\gamma}}
\end{pmatrix},
\end{equation}
where $\gamma>0$, $u\in W_{\mathrm{loc}}^{N,2}(\R_{+})$, with $N\in \N_{1}$, such that $\vert u(x) \vert= o(e^{x^{\gamma}})$ as $x\to +\infty$. All conditions of Assumption \ref{Assump General C} are satisfied with $\nu=\gamma-1$ and any $\varepsilon_{1}>0$. Given $\beta>1$, the turning point $x_{\beta}>0$ is determined by the relation
$x_{\beta}=\ln(\beta)^{\frac{1}{\gamma}}.$
For any $c\in(0,1)$, thanks to \eqref{Eq Exponential C}, there is some $C>0$ such that
\begin{align*}
\kappa(\beta,c) = \mathcal{O}\left(\exp\left(-C \beta \ln(\beta)^{\frac{1}{\gamma}-1}\right)\right).
\end{align*}
Again, since $\mathrm{Re}\, V_{11}=\mathrm{Re}\, V_{22}=0$ the conditions in \eqref{Eq a 1 C} are equivalent to the fact \eqref{Ex2 a 1}. We may take $\alpha$ that is, with some $\varepsilon>0$,
$ \vert \alpha \vert \gtrsim \beta^{\frac{1}{2}+\varepsilon} \ln(\beta)^{\frac{1}{2}\frac{\gamma-1}{\gamma}}$
such that the condition \eqref{Ex2 a 2} is satisfied. Under this choice of $\alpha$, we can bound above $\sigma^{(N)}(\beta)$ 
\begin{align*}
\sigma^{(N)}(\beta) &\lesssim \left\{
\begin{aligned}
&\beta^{-N}\ln(\beta)^{\frac{2(\gamma-1)}{\gamma}N}, \qquad &&\text{ if } \vert \alpha \vert>\beta \text{ and } \gamma \geq 1,\\
&\beta^{-N} \ln(\beta)^{\frac{\gamma-1}{\gamma}(N+1)}, \qquad &&\text{ if } \vert \alpha \vert>\beta \text{ and } \gamma <1,\\
&\beta^{\frac{1}{2}-\varepsilon(2N+1)} \ln(\beta)^{\frac{1}{2}\frac{\gamma-1}{\gamma}}, \qquad &&\text{ if } \vert \alpha \vert\leq \beta .
\end{aligned}
\right.
\end{align*}
In order to get the decay of $\sigma^{(N)}(\beta)$ as $\beta\to+\infty$, we choose
\[ \varepsilon=\frac{1}{4N+2}+\eta, \qquad \text{ with } \eta>0. \]
In conclusion: for any $\eta>0$,  there exist $\beta_{0}>0$ and a family $(\Psi_{\lambda,N})_{\lambda\in \Omega}$ such that
\[\frac{\Vert (H_{V}- \lambda)\Psi_{\lambda,N} \Vert}{\Vert \Psi_{\lambda,N} \Vert} = \left\{
\begin{aligned}
&\mathcal{O}\left(\beta^{-N}\ln(\beta)^{\frac{2(\gamma-1)}{\gamma}N}\right), \qquad &&\text{ if } \vert \alpha \vert>\beta \text{ and } \gamma \geq 1,\\
&\mathcal{O}\left(\beta^{-N} \ln(\beta)^{\frac{\gamma-1}{\gamma}(N+1)}\right), \qquad &&\text{ if } \vert \alpha \vert>\beta \text{ and } \gamma <1,\\
&\mathcal{O}\left(\beta^{-\eta(2N+1)} \ln(\beta)^{\frac{1}{2}\frac{\gamma-1}{\gamma}}\right), \qquad &&\text{ if } \vert \alpha \vert\leq \beta,
\end{aligned}
\right.\]
where
\begin{equation}\label{Eq Exp Omega C}
\Omega:=\left\{ \alpha+i\beta \in \C: \beta> \beta_{0} \text{ and } \vert \alpha \vert \gtrsim   \beta^{\frac{1}{2}+\frac{1}{4N+2}+\eta} \ln(\beta)^{\frac{1}{2}\frac{\gamma-1}{\gamma}}\right\}.
\end{equation}
In Figure~\ref{Fig: Exp}, some sketches are created for the the imagination of the pseudospectral region $\Omega$ in the exponential cases. As in the polynomial cases, the higher $\gamma$ is, the further $\Omega$ stays away from the axis $\alpha=0$.
\begin{figure}[h!]
 \centering
 \begin{subfigure}[c]{0.3\textwidth}
\centering
\includegraphics[width=1\textwidth]{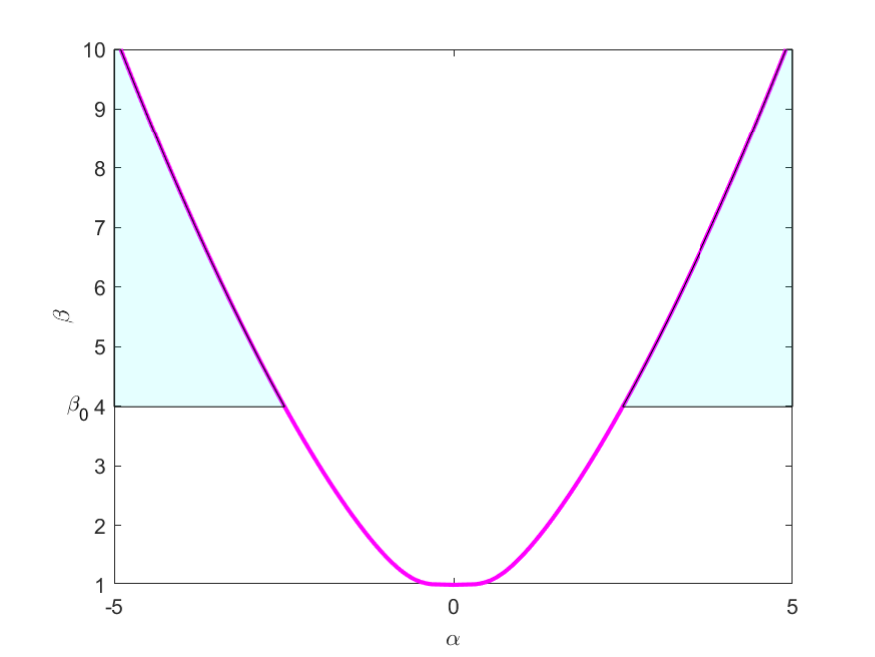}
\caption{$\gamma=2$.}
\end{subfigure}
~
\begin{subfigure}[c]{0.3\textwidth}
\centering
\includegraphics[width=1\textwidth]{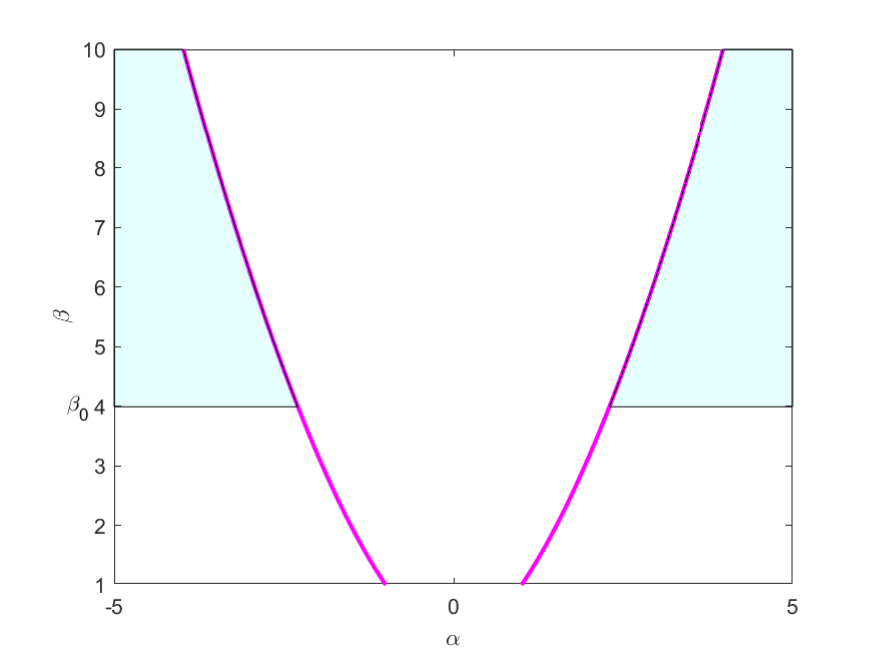}
\caption{$\gamma=1$.}
\end{subfigure}
~
\begin{subfigure}[c]{0.3\textwidth}
 \centering
\includegraphics[width=1\textwidth]{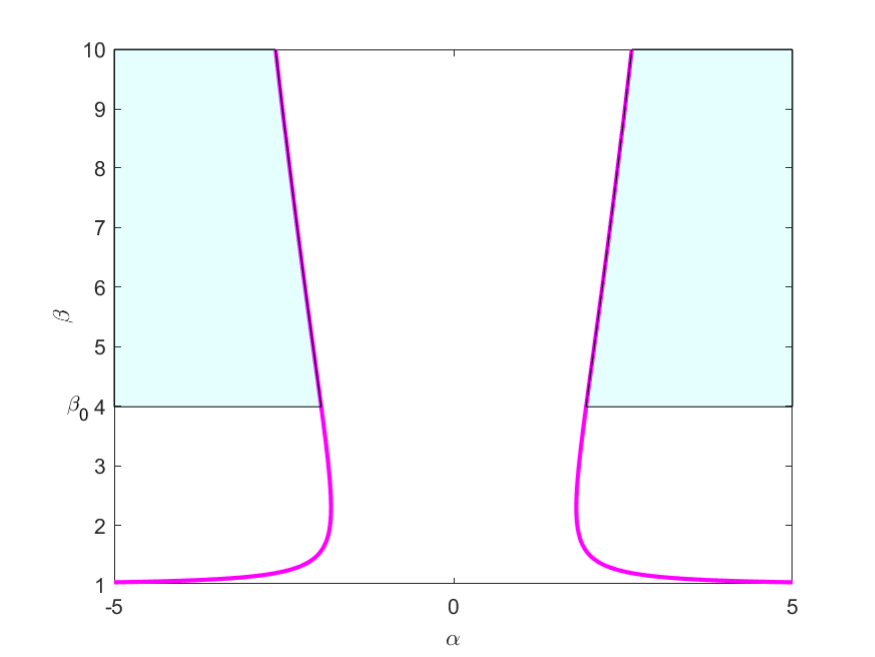}
\caption{$\gamma=\frac{1}{2}$.}
\end{subfigure}
\captionsetup{singlelinecheck=off}
\caption[foo bar]{Illustrations of the shapes of $\Omega$ (in cyan color) associated with the potential $V$ given in \eqref{Eq Exp C} corresponding to $\gamma$. Here we assume that $N$ is large enough such that we can take $\eta$ small enough satisfying $\frac{1}{4N+2}+\eta=\frac{1}{10}$ 
on the right-hand side of \eqref{Eq Exp Omega C}. The magenta curves are, respectively,
  \begin{enumerate}[(A)]
  \item $\gamma=2$: $\vert \alpha \vert=\beta^{\frac{1}{2}+\frac{1}{10}}\ln(\beta)^{\frac{1}{4}}$,
  \item $\gamma=1$: $\vert \alpha \vert=\beta^{\frac{1}{2}+\frac{1}{10}}$,
  \item $\gamma=\frac{1}{2}$: $\vert \alpha \vert=\beta^{\frac{1}{2}+\frac{1}{10}}\ln(\beta)^{-\frac{1}{2}}$.
  \end{enumerate}
}\label{Fig: Exp}
\end{figure}
\end{example}

\section*{Acknowledgement}
The authors were supported by the EXPRO grant
number 20-17749X of the Czech Science Foundation (GA\v{C}R).
\bibliographystyle{abbrv}
\bibliography{Ref22}
\end{document}